\documentclass[authorcolumns,numberwithinsect]{no-lipics-v2022}

\let\mathscr\relax

\usepackage{tcolorbox}

\usepackage{todonotes}
\usepackage{graphicx}

\nolinenumbers

\newcommand{\uncolholant}{\textup{\textsf{Holant}}}

\newcommand{\uni}[1]{#1_{\text{\sf{uni}}}}

\newcommand{\classHolant}{\uni{\text{\sc{p-Holant}}}^{2}}

\newcommand{\constclassHolant}[1]{\uni{\text{\sc{p-Holant}}}^{#1}}

\newcommand{\unbndclassHolant}{\text{\sc{Holant}}}

\newcommand{\uncolholantprob}{\text{\sc{p-Holant}}}

\newcommand{\constuncolholantprob}[1]{\text{\sc{p-Holant}}^#1}

\newcommand{\uncolunbndholantprob}{\text{\sc{p-Holant}}^{\sf{UBD}}}

\newcommand{\uncolholantprobuni}
{\text{\sc{p-Holant}}_{\textsf{uni}}}

\newcommand{\Sub}{\#\mathsf{Sub}}
\newcommand{\Hom}{\#\mathsf{Hom}}
\newcommand{\Emb}{\#\mathsf{Emb}}
\newcommand{\Aut}{\#\mathsf{Aut}}

\newcommand{\homs}[2]{\mathsf{Hom}{(#1 \to #2)}}
\newcommand{\subs}[2]{\mathsf{Sub}{(#1 \to #2)}}
\newcommand{\embs}[2]{\mathsf{Emb}{(#1 \to #2)}}
\newcommand{\auts}[1]{\mathsf{Aut}#1}
\newcommand{\surhoms}[2]{\mathsf{SurHom}{(#1 \to #2)}}

\newcommand{\Homprob}{\#\textup{\textsc{Hom}}}

\newcommand{\kgraphsBHS}{\mathcal{G}_{k}(\mathcal{S};\hs)}
\newcommand{\kgraphsGamma}{\mathcal{G}_{\leq k}(\mathcal{S};\hs)}

\newcommand{\colGraph}[1]{(#1, \nu_{#1})}

\newcommand{\hs}{\text{\sf{r}}}

\newcommand{\fptlinred}{\leq^{\mathsf{FPT-lin}}_{\mathsf{T}}}

\author{Panagiotis Aivasiliotis}{Hasso Plattner Institute, University of Potsdam}{panos.aivasiliotis@hpi.de}{}{}

\author{Andreas Göbel}{Hasso Plattner Institute, University of Potsdam}{andreas.goebel@hpi.de}{}{Funded by the Postdoc Network Brandenburg.}

\author{Marc Roth}{School of Electronic Engineering and Computer Science, Queen Mary University of London}{m.roth@qmul.ac.uk}{}{}

\authorrunning{P. Aivasiliotis, A. Göbel and M. Roth}

\title{Symmetric Parameterised Holants on Hypergraphs:
Towards a Classification for Parameterised VCSPs}

\titlerunning{Parameterised Holant Problems on Hypergraphs}


\begin{document}

\maketitle
\begin{abstract}
We study the complexity of the parameterised counting constraint satisfaction problem: given a set of constraints over a set of variables and a positive integer $k$, how many ways are there to assign $k$ variables to $1$ (and the others to $0$) such that all constraints are satisfied.
While this problem, and its decision version, received significant attention during the last two decades, existing work has so far exclusively focused on restricted settings such as finding and counting homomorphisms between relational structures due to Grohe (JACM 2007) and Dalmau and Jonsson (TCS 2004), or the case of finite constraint languages due to Creignou and Vollmer (SAT 2012), and Bulatov and Marx (SICOMP 2014).

In this work, we tackle a more general setting of parameterised (counting) valued constraint satisfaction problems (VCSPs) with infinite constraint languages: we allow our constraints to be chosen from an infinite set of permitted constraints and we allow our constraints to map an assignment of its variables not only to $\mathsf{True}$ or $\mathsf{False}$, but to arbitrary values. In this setting we are able to model and classify significantly more general problems such as (weighted) parameterised factor problems on hypergraphs and counting weight-$k$ solutions of systems of linear equations, none of which are captured by existing complexity classifications of parameterised constraint satisfaction problems.

On a formal level, we express parameterised VCSPs as parameterised \emph{holant problems} on uniform hypergraphs, and we establish complete and explicit complexity dichotomy theorems for this family of problems both w.r.t.\ classical complexity theory ($\mathrm{P}$ vs.\ $\#\mathrm{P}$) and parameterised complexity ($\mathrm{FPT}$ vs.\ $\#\mathrm{W}[1]$). For resolving the $\mathrm{P}$ vs.\ $\#\mathrm{P}$ question, we mainly rely on the use of hypergraph gadgets, the existence of which we prove using properties of degree sequences necessary for realisability in uniform hypergraphs. As a technical highlight, we also employ Curticapean's ``CFI Filters'' (SODA 2024) --- named after the Cai-Fürer-Immermann construction for bounding the expressiveness of the Weisfeiler-Leman heuristic --- to establish polynomial-time algorithms for isolating vectors in the homomorphism basis of some of our holant problems. For the $\mathrm{FPT}$ vs.\ $\#\mathrm{W}[1]$ question, we build upon the recently established combinatorial toolkit for parameterised holants on the special case of graphs by Aivasiliotis et al. (ICALP 2025) and also rely on an extension of the framework of the homomorphism basis due to Curticapean, Dell and Marx (STOC 17) to uniform hypergraphs. 
\end{abstract}
\newpage

\section{Introduction}
Constraint satisfaction problems (``CSPs'') and their counting analogues (``$\#$CSPs'') belong to the most well-studied computational problems in algorithms and complexity theory: given a set of constraints over a set of variables, find one of (or count) the assignments of variables to elements of a finite set $D$ such that all constraints are satisfied. CSPs and $\#$CSPs are extremely rich in expressibility, allowing us not only to model a multitude of fundamental problems such as the Boolean satisfiability problem, the graph homomorphism problem, and the computation of partition functions, but they also provide a uniform language for studying and analysing seemingly very different problems (see e.g.\ Bulatov's survey on CSPs~\cite{BulatovSurvey18}). At the same time, families of $(\#)$CSPs are robust enough in structure to often allow for a \emph{precise and exhaustive} complexity classification of their members, which usually comes in form of a \emph{complexity dichotomy theorem} stating that each member of a family of $(\#)$CSPs is either solvable in polynomial time, or $\mathrm{NP}$-hard.\footnote{This property is rather surprising, given that Ladner's Theorem states that, assuming $\mathrm{P}\neq \mathrm{NP}$, the class $\mathrm{NP}$ contains an infinite hierarchy of complexities between $\mathrm{P}$ and $\mathrm{NPC}$ \cite{Ladner75}.}

One of the earliest examples of a comprehensive analysis of a CSP is Schaefer's Dichotomy~\cite{schaefer1978complexity}, stating that each member of the generalised satisfiability problem is either solvable in polynomial time or $\mathrm{NP}$-complete. Over the course of the following 20 years a significant amount of research had been undertaken with the goal of lifting Schaefer's result from the Boolean domain to arbitrary finite domains, culminating in the formulation of the Feder-Vardi-Conjecture in 1998 \cite{FederV98}. It then took almost another 20 years until the conjecture was proved independently by Bulatov \cite{Bulatov17} and Zhuk \cite{Zhuk17} in 2017.

The quest for counting constraint satisfaction problems ($\#$CSPs) has seen similar success: Creignou and Hermann~\cite{creignou1996complexity} established a counting version of Schaefer's dichotomy in 1996, and the classification for $\#$CSPs over arbitrary domains was established again by Bulatov~\cite{Bulatov13} whose result was then simplified and its tractability criterion shown to be decidable by Dyer and Richerby~\cite{DyerR13}. Moreover, the classification was eventually extended and strenghthened to $\#$CSPs with complex weights by Cai and Chen~\cite{CaiC17}. In the present work, we study the version of $\#$CSP in which we only count ``weight-$k$-solutions'', that is assignments that set precisely $k$ variables to $1$; this version is also known as \emph{parameterised} $\#$CSP and we focus on the Boolean domain, i.e., \#CSPs defined over $D = \{0,1 \}$.

\subsection{Parameterising CSPs by Solution Size}
The \emph{weighted satisfiability problem} $\textsc{WSAT}$ asks, given a Boolean formula $F$ and a positive integer $k$, whether $F$ contains a satisfying assignment of Hamming weight $k$, that is, whether it is possible to satisfy the formula by setting precisely $k$ variables to $1$. Similarly, its counting version $\#\textsc{WSAT}$ asks to compute the number of such assignments. $\textsc{WSAT}$ and $\#\textsc{WSAT}$ are the foundational problems for the complexity theory of parameterised decision and counting problems: depending on the allowed structure of the input formula\footnote{The ``W'' in the $\mathrm{W}$-hierarchy stands for ``weft'', which describes a structural property of boolean circuits.}, they define, respectively, the $\mathrm{W}$ and the $\#\mathrm{W}$-hierarchies, including the classes $\mathrm{W}[1]$ and $\#\mathrm{W}[1]$, which can be considered the parameterised decision and counting equivalents of $\mathrm{NP}$. We refer the reader to the standard textbook of Flum and Grohe~\cite{FlumG06} for a comprehensive introduction to structural parameterised complexity theory.

Notably, $\textsc{WSAT}$ and $\#\textsc{WSAT}$ constitute the most rudimentary instances of parameterised CSP and $\#$CSP, and it comes to no surprise that analogues of Schaefer's dichotomy, both for decision and for counting, have already been established for those problems~\cite{Marx05,bulatov2014constraint,CreignouV15}. The key difference in the parameterised setting is that the notion of tractability is relaxed from polynomial-time algorithms to  \emph{fixed-parameter tractable} ``FPT'' algorithms, the requirement for which is a running time bound of $f(k)\cdot n^{O(1)}$, where $n$ is the input size and $k$ is the problem parameter --- in this work, $k$ will always be the Hamming weight of the sought solutions. The function $f$ can be any computable function; in other words, we allow for a super-polynomial overhead only in the problem parameter, reflecting the assumption that the parameter is significantly smaller than the input size in problem instances of interest (see \cite{bulatov2014constraint} for a discussion on this assumption in the context of parameterised CSPs). 

When studying families of parameterised CSPs, the goal is, thus, to classify the problems into instances that are FPT and instances that are hard for a class in the $\mathrm{W}$ (or $\#\mathrm{W}$) hierarchy. In particular, it is well-known that under the \emph{exponential time hypothesis} (ETH) \cite{ImpagliazzoP01,ImpagliazzoPZ01}, which asserts that 3-SAT cannot be solved in time $\mathsf{exp}(o(n))$ (where $n$ is the number of variables of the input formula), $(\#)\mathrm{W}[1]$-hard problems are not in FPT \cite{Chenetal05,Chenetal06,CyganFKLMPPS15}.  In addition to the aforementioned analogues of Schaefer's dichotomy, such classifications have also been achieved for the problems of finding and counting homomorphisms between relational structures, which can both be expressed as parameterised $(\#)$CSPs (see \cite{Grohe07,dalmau2004complexity}).

However, all existing results are limited by either focusing on a special case (such as simple, undirected graph homomorphisms) or by the restriction to finite constraint languages. For example, the following natural counting problems cannot be expressed in existing frameworks for parameterised $\#$CSPs:\footnote{We note that special cases of $\#\textsc{CodeWord}_p$ and of $\#\textsc{Factor}(S)$, such as restrictions to certain input matrices and to graphs, are expressible in existing frameworks such as in~\cite{aivasiliotis_et_al:LIPIcs.ICALP.2025.7}. However, the complexity of both problems has not been fully resolved yet.}
\begin{itemize}
    \item Let $p$ be a prime. The problem $\#\textsc{CodeWord}_p$ gets as input a matrix $A$ over $\mathbb{Z}/p\mathbb{Z}$ and a positive integer $k$, and the goal is to compute the number of solutions to the system of linear equations $A\vec{x}=\vec{0}$ (also over $\mathbb{Z}/p\mathbb{Z}$) that set precisely $k$ variables to $1$ (and the other variables to $0$). Note that $\#\textsc{CodeWord}_2$ is identical to the counting version of $\textsc{Exact-Even-Set}$ (see e.g.\ \cite{DowneyFVW99}).
    \item Let $S\subseteq \mathbb{N}$. The problem $\#\textsc{Factor}(S)$ gets as input a hypergraph $H$ and a positive integer~$k$, and the goal is to compute the number of $k$-edge-subsets of $H$ that induce a hypergraph in which the degree of every vertex is in~$S$. Given $d\in \mathbb{N}$, we denote $\#\textsc{Factor}_d(S)$ for the restriction to $d$-uniform hypergraphs. For example, $\#\textsc{Factor}_d(\{0,1\})$ is identical to the problem of counting $k$-matchings in $d$-uniform hypergraphs, where a matching in a hypergraph is a set of pairwise disjoint hyperedges.
\end{itemize}

In this work, we introduce and study a version of parameterised $\#$CSP which is powerful enough in expressibility to model the previous problems ---in fact, we will consider a much more general model which can also describe constraints that map assignments to arbitrary numbers (rather than to just $\mathsf{True}$ or $\mathsf{False}$). Moreover, we will be able to provide complete complexity classifications for our model which reveal the complexity of the aforementioned examples as straightforward special cases.

\paragraph*{Parameterised Symmetric Valued $\#$CSPs and Hypergraph Holants}
For reasons of notational consistency with existing work on (parameterised) CSPs (see, e.g., \cite{jerrum2017counting}) we will first introduce our setting of interest via so-called \emph{Valued} CSPs (``VCSPs'').\footnote{Being aware of another work (see \cite{kolmogorov2017complexity}) studying a problem under the name Valued CSPs, we would like to point out that their problem is incomparable to ours and that our definition complies with \cite{jerrum2017counting}.}

Let $\mathcal{F}$ denote a finite set of symmetric functions $f : \{0, 1\}^\ast \to \mathbb{C}$; by symmetric we mean that $f$ is invariant under permutation of entries of input vectors. 

\begin{remark}
Note that, a symmetric function that is not associated with a specific arity, i.e., the length of its input is arbitrary, depends only on the Hamming weight of its input. Hence, such functions can be equivalently defined over natural numbers.    
\end{remark}

An instance $\mathcal{I}$ of parameterised symmetric $\textsc{p-VCSP}(\mathcal{F})$ consists of a positive integer $k$, a set $X$ of $n$ \textit{variables} $x_1, \dots, x_n$ and a set $C$ of $m$ \textit{constraints} of the form $\langle f, (x_{i_1}, ..., x_{i_r})\rangle$ for some $r\in \mathbb{N}$, where $f \in \mathcal{F}$. The vector $(x_{i_1}, \dots, x_{i_r})$ is known as the \textit{scope} of the constraint. Given an instance $\mathcal{I}=(X,C,k)$, the task is to compute

\[
Z(\mathcal{I}) = \sum_{\substack{\alpha : X \to \{0,1\}\\\sum_{x\in X}\alpha(x)=k}}\prod_{\langle f, (x_{i_1}, \dots, x_{i_r})\rangle \in C}f(\alpha(x_{i_1}), \dots, \alpha(x_{i_r}))\,.
\]

For example, if $\mathcal{F}$ contains precisely the function $\vec{x} \mapsto \sum_{x\in \vec{x}} x\mod p$, then $\textsc{p-VCSP}(\mathcal{F})$ is identical to $\#\textsc{CodeWord}_p$. If $\mathcal{F}$ contains precisely the function $\mathsf{hw}_{\leq 1}$ which maps a vector $\vec{x}\in \{0,1\}^\ast$ to $1$ if it has at most one $1$ entry, and to $0$ otherwise, then $\textsc{p-VCSP}(\mathcal{F})$ is identical to $\#\textsc{Factor}(\{0,1\})$.

Our goal is to understand the complexity of  $\#\textsc{p-VCSP}(\mathcal{F})$ for any set of symmetric constraint functions $\mathcal{F}$ - recall that symmetric constraints only depend on the Hamming weight of their assignment. To this end, we introduce and rely on a generalisation of the parameterised holant framework \cite{Curticapean15,aivasiliotis_et_al:LIPIcs.ICALP.2025.7} from graphs to hypergraphs, which enables us to naturally model VCSPs in a way that allows us to leverage a variety of powerful tools that have been established in parameterised counting complexity theory over the past years. We proceed by formally introducing the holant framework, after which we provide details on the equivalence between VCSPs and parameterised holants on hypergraphs.

Holant problems, implicitly defined by Valiant in his seminal paper on \textit{holographic reductions} \cite{Valiant08} and formally introduced by Cai, Lu and Xia \cite{cai2009holant}, yield an avenue for modeling various counting problems, including problems that cannot be formulated as a $\#$CSP problem, such as the generating function of perfect matchings in a graph (see \cite{CaiG22,jerrum2017counting}). 
Similarly to VCSPs, holant problems assume oracle access to a finite collection of functions, which by convention we call \textit{signatures} and denote them with $s$ instead of $f$. 

An instance of a holant problem with respect to a finite set $\mathcal{S}$ of signatures is a \textit{signature grid} that consists of a (simple) undirected graph $G$ and an assignment of signatures from $\mathcal{S}$ to the vertices of $G$. For a vertex $v \in V(G)$, we denote by $s_v$, the signature that is assigned to $v$. The objective is to compute the holant value (or simply, the \textit{holant}) of the signature grid $(G, \{s_v\}_{v \in V(G)})$ which is given as follows
\begin{equation}\label{eq:introHolantValue}
\mathsf{Holant}(G, \{s_v\}_{v \in V(G)}) = \sum_{\alpha : E(G) \to \{0,1\}}\prod_{v \in V(G)}s_v(\alpha|_{E_v(G)})\,,
\end{equation}
where $\alpha|_{E_v(G)}$ is the restriction of $\alpha$ to the edges that are incident to $v$. The focus of this article is on the following holant variant.

The \textit{parameterised holant problem} on a finite set $\mathcal{S}$ of symmetric signatures $s : \{0,1\}^* \to \mathbb{C}$, takes as input a signature grid $\Omega = (G, \{s_v\}_{v \in V(G)})$, and a number $k \in \mathbb{Z}_{\geq 0}$ and computes
\[
\uncolholant(\Omega, k) = \sum_{\substack{\alpha : E(G) \to \{0,1\} \\ \sum_{e}\alpha(e) = k}}\prod_{v \in V(G)}s_v(\alpha|_{E_G(v)})\,,
\]

\begin{remark}
Following the equivalent definition of symmetric signatures $s : \{0,1\}^* \to \mathbb{C}$ as $s : \mathbb{N} \to \mathbb{C}$, we note that $\uncolholant(\Omega, k)$ can be written in a simplified form by observing that for any assignment $\alpha : E(G) \to \{0,1\}$ s.t. $\sum_{e}\alpha(e) = k$ and any vertex $v \in V(G)$, by letting $A = \{e \in E(G)\,|\,\alpha(e) = 1\}$, we have that the number of 1s in $\alpha |_{E_G(v)}$ is precisely equal to $|A\,\cap\, E_G(v)|$. Hence,
\[
\uncolholant(\Omega, k) = \sum_{\substack{A \subseteq E(G) \\ |A| = k}}\prod_{v \in V(G)}s_v(|A\,\cap\,E_G(v)|)\,.
\]
\end{remark}

The integer $k$ is called the \emph{parameter} of the problem instance, and, in addition to classifying the polynomial-time tractable cases, we are interested in \emph{fixed-parameter tractable} (FPT) algorithms. Holant instances can naturally be extended to signature grids over \textit{hypergraphs} $(V, E)$, where $V$ is the set of vertices and $E \subseteq 2^{V}$ is the \textit{multiset} of \textit{hyperedges}\footnote{Note that a hyperedge is a set and \textit{not} a multiset, but multiple copies of the same hyperedge are allowed.}, that generalise graphs. The holant (value) of a signature grid over an underlying hypergraph is defined identically to \eqref{eq:introHolantValue}, where a hyperedge $e$ is incident to a vertex $v$ if $v \in e$.  

Thus far, parameterised holant problems may take as input a signature grid with an underlying hypergraph of arbitrarily large rank\footnote{Recall that the rank of a hypergraph $G$ is given by $\max_{e \in E(G)}|e|$.}. Naturally, we are also interested in the case in which the rank of the input hypergraph is \textit{bounded}. There are two ways by which we may \textit{bound} the rank of the instance's underlying hypergraph, namely, either to include the rank in the parametrisation along with $k$, or to allow only for hypergraphs of fixed rank, that is, their rank is at most some \textit{constant} $d$ (the latter of which is in fact the \textit{slice} of the former). In particular, the most interesting variant considering fixed rank is the one that is restricted to \textit{uniform} hypergraphs. The three aforementioned variants are formally stated below\footnote{We point out that parameterised holant problems $\textsc{p-Holant}$ have been traditionally considered in a \textit{colourful} setting in which edges are coloured, while uncoloured holant problems are typically denoted as $\textsc{p-UnColHolant}$ (see e.g., \cite{Curticapean15}, \cite{aivasiliotis_et_al:LIPIcs.ICALP.2025.7}). In this work, we only consider the uncoloured setting, thus we refrain from indicating in the notation of our holant problems that are uncoloured.} --- we emphasize that we assume signature grids over hypergraphs that feature \textit{no} multiple hyperedges for everything that follows. 

\begin{definition}[Parameterised Holant Problems on Hypergraphs]
    Let $\mathcal{S}$ be a finite set of signatures. 
\begin{itemize}
\item    
    The problem $\uncolunbndholantprob(\mathcal{S})$ expects as input a positive integer $k$ and a signature grid $\Omega=(G,\{s_v\}_{v\in V(G)})$ over $\mathcal{S}$. The output is $\uncolholant(\Omega, k)$ and the problem parameter is $k$. 
\item 
    The problem $\uncolholantprob(\mathcal{S})$ expects as input a positive integer $k$ and a signature grid $\Omega=(G,\{s_v\}_{v\in V(G)})$ over $\mathcal{S}$. The output is $\uncolholant(\Omega, k)$ and the problem parameters are $k$ and the rank of $G$. 
\item For any constant $d \in \mathbb{Z}_{\geq 2}$, the problem $\uni{\constuncolholantprob{d}}(\mathcal{S})$ expects as input a positive integer $k$ and a signature grid $\Omega=(G,\{s_v\}_{v\in V(G)})$ over $\mathcal{S}$, such that $G$ is $d$-uniform, that is, each hyperedge has cardinality $d$. The output is $\uncolholant(\Omega, k)$ and the problem parameter is $k$. 
\end{itemize}
\end{definition}

\begin{remark}[On infinite domains] We emphasise that our definition of parameterised holants on hypergraphs is, up to the extension from graphs to hypergraphs, \textit{identical} to the definition in existing literature \cite{aivasiliotis_et_al:LIPIcs.ICALP.2025.7,Curticapean15}. As discussed in \cite{aivasiliotis_et_al:LIPIcs.ICALP.2025.7}, computable signatures with infinite domains ---though unusual in classical holant theory--- yield a non-trivial well-defined framework that is able to capture problems that would otherwise require an infinite number of (finite-arity) signatures (see, e.g., $\#\textsc{CodeWord}_p$ and $\#\textsc{Factor}(S)$ discussed earlier). Even more importantly, finite-arity signatures can only capture those instances for which the underlying hypergraph is of bounded vertex-degree. In the parameterised setting, such a restriction above is known to admit trivial FPT-algorithms (a detailed discussion about the latter point is formally provided in \cite{aivasiliotis_et_al:LIPIcs.ICALP.2025.7} as well as briefly in \Cref{sec:Conclusion}). Thereby a study on infinite domains is crucial for establishing the complexity of important families of problems that remained uncaptured by existing frameworks. 
\end{remark}

\begin{remark}[On Boolean domain and symmetric signatures]
In general, holant problems allow for general domains (see, e.g., \cite{CaiG21}) and signatures that do not need to be symmetric~\cite{LinW17}. However, even in the classical setting, only partial results are known for symmetric functions on general domains \cite{cai2025holant} or for asymmetric functions on Boolean domains~\cite{meng2025fp,MengWXZ25}. Therefore, this work focusses on symmetric Boolean signatures not because the asymmetric case is less interesting, but as a natural starting point for investigating parameterised holants on hypergraphs.
% The framework of holant problems is an effective way to model a wide range of counting problems, and it has been shown to admit dichotomies for numerous general families of signatures, including, but not limited to, Boolean symmetric signatures \cite{CaiGW16} and non-negative real-valued signatures \todo{Maybe skip this here, since we discuss this later. Or bring the later text here.}
\end{remark}

% \subsection{On the Relationship Between CSPs and Holants}
\subsection{Equivalence of VCSPs and Holant Problems on Hypergraphs}\label{appendix:VCSPsHolants}
We show in a very simple reduction, that holant problems on hypergraphs are VCSPs (and vice versa). In particular, here we show the equivalence for the respective parameterised problems. We stress that for the aforementioned equivalence to hold, we need to modify our definition of holant problems such that the underlying hypergraphs allow for multi-hyperedges, i.e., multiple copies of the same hyperedge may appear. However, this is only a technical subtlety that we do not need to consider in the rest of the paper. The reason is that all of our hardness results for the \textit{restricted} holant problems transfer trivially to the more general holant problems that allows for multi-hyperedges and as we show in \Cref{sec:TractabilitySubsection} our tractability results still apply as well. 

We let $\mathcal{F}$ be a finite set of functions $f : \{0,1\}^r \to \mathbb{C}$ and recall that, given a set $X$ of $n$ variables $x_1, \dots, x_n$, a set $C$ of $m$ constraints of the form $\langle f, (x_{i_1}, \dots, x_{i_r})\rangle$, where $f \in \mathcal{F}$, and a positive integer $k$, the problem $\textsc{p-VCSP}(\mathcal{F})$ computes
\[
Z(X, C, k) = \sum_{\substack{\alpha : X \to \{0,1\} \\ \sum_{e}\alpha(e) = k}}\prod_{\langle f, (x_{i_1}, \dots, x_{i_r})\rangle \in C}f(\alpha(x_{i_1}), \dots, \alpha(x_{i_r}))\,.
\]

Given $\mathcal{I} = (X, C, k)$, we construct a hypergraph $G(\mathcal{I})$, that will be used as the underlying hypergraph of the intended holant instance, as follows. Each constraint $c \in C$ corresponds to a (unique) vertex $v_c \in G(\mathcal{I})$ and each variable $x \in X$ corresponds to a (unique) hyperedge $e_x \in E(G(\mathcal{I}))$ given as $\{v_c \mid \text{$x$ is in the scope of $c$}\}$. Note that, there may be variables $x \neq x'$ such that $e_x = e_{x'}$. We allow $G(\mathcal{I})$ to contain multiple copies of the same hyperedge, such that each variable $x$ corresponds uniquely to (some copy of) $e_x$.

We let $\Omega(\mathcal{I})$ denote the signature grid obtained from $G(\mathcal{I})$ by equipping, for each $c = (f, (x_{i_1}, \dots, x_{i_r})) \in C$, the variable $v_c$ with the signature $s_{v_c} = f$. We also assume that the order of the arguments of $f$  is preserved, that is, if $x$ is the $j$-th argument of $f$, then $e_x$ is the $j$-th argument of $s_{v_c}$. It can be readily verified that
$Z(\mathcal{I}) = \uncolholant(\Omega(\mathcal{I}), k)\,.$ 

\begin{remark}\label{rem:introVSCPsToHolants}
For $x \in X$, we write $\textsf{occ}_{\mathcal{I}}(x)$ for the \textit{occurence} of $x$ in $\mathcal{I}$, that is, the number of constraints $\langle f, (x_{i_1}, \dots, x_{i_r})\rangle \in C$ that contain $x$ in their scope, that is, $x = x_{i_j}$, for some $1 \leq j \leq r$.
For any variable $x \in X$, we have that $\textsf{occ}_{\mathcal{I}}(x)$ is equal to the size of $e_x$. Furthermore, for any constraint $c \in C$, we have that $|c|$, that is, the number of variables in the scope of $c$, is equal to the number of hyperedges in $G(\mathcal{I})$ that are incident to $v_c$. 
\end{remark}

Furthermore, the reverse direction, that is, formulating a (parameterised) generalised holant problem as a VSCP can be achieved in a similar fashion.

\subsection{Existing Results for Parameterised Holants on Graphs}

The parameterised complexity of the problem $\uni{\constuncolholantprob{2}}(\mathcal{S})$ was recently classified by Aivasiliotis et al.\ \cite{aivasiliotis_et_al:LIPIcs.ICALP.2025.7}. The classification criterion depends on the \textit{type} of the set of allowed signatures, as defined in \cite{aivasiliotis_et_al:LIPIcs.ICALP.2025.7}, which we also define here momentarily, for reasons of self-containment. In particular, only signature sets that exclude signatures $s$ with $s(0) = 0$ are assigned a type in~\cite{aivasiliotis_et_al:LIPIcs.ICALP.2025.7} and the reason is that --- roughly speaking --- signatures $s$ with $s(0) = 0$ do not affect the \emph{parameterised} complexity of the problem, as it was shown in \cite{aivasiliotis_et_al:LIPIcs.ICALP.2025.7}. In other words, any fixed-parameter tractable finite set of signatures remains fixed-parameter tractable even if we add to it any number of signatures $s$ with $s(0) = 0$. However, we will see that this is not necessarily true for polynomial-time tractability, which is why we explicitly include the case of $s(0)=0$ in the present paper.

The type of a signature set is decided by the \textit{signature fingerprint} function, which is a combinatorial sum that is reminiscent of the Möbius function of partition lattices, given as follows.

\begin{definition}[Signature Fingerprints \cite{aivasiliotis_et_al:LIPIcs.ICALP.2025.7}]\label{def:fingerprint_intro}
Let $a$ be a positive integer and let $s$ be a signature. If $s(0)\neq 0$ then the \emph{fingerprint} of $a$ and $s$ is defined as
    \[\chi(a,s) := \sum_{\sigma} (-1)^{|\sigma|-1} (|\sigma|-1)! \cdot \prod_{B \in \sigma} \frac{s(|B|)}{s(0)} \,,\]
    where the sum is over all partitions $\sigma$ of $[a]$ and $|\sigma|$ denotes the number of blocks $B \in \sigma$. For technical reasons, if $s(x)=0$ for all $x$, we set $\chi(a,s)=0$. In all remaining cases, i.e., $s(0)=0$ but $s(x)\neq 0$ for some $x>0$, the fingerprint is undefined and we set $\chi(a,s)=\bot$.
\end{definition}

\begin{definition}[Types of Signature Sets \cite{aivasiliotis_et_al:LIPIcs.ICALP.2025.7}]\label{def:sigtype_intro}
    Let $\mathcal{S}$ be a finite set of signatures. We say that $\mathcal{S}$ is of type\footnote{We point out that we have renamed the types of signatures to avoid introducing further concepts.}
    \begin{itemize}
        \item[1.] $\mathbb{T}[1]$ if $\chi(a,s)=0$ for all $s\in \mathcal{S}, a\geq 2$, 
        \item[2.] $\mathbb{T}[2]$ if $\chi(a,s)=0$ for all $s\in \mathcal{S}, a\geq 3$, but there exists $s\in \mathcal{S}$ with $\chi(2,s)\neq 0$, and
        \item[3.] $\mathbb{T}[\infty]$ otherwise, i.e., there exists $s\in \mathcal{S}$ and $a\geq 3$ such that $\chi(a,s)\neq 0$.
    \end{itemize}
\end{definition}

It was shown in~\cite{aivasiliotis_et_al:LIPIcs.ICALP.2025.7} that there are infinitely many signature sets of each type.
We are now ready to state the known complexity dichotomy for $\uni{\constuncolholantprob{2}}(\mathcal{S})$. In particular, it is only the signature sets of the first type that render the problem tractable. However, the third type of signature sets allow for tighter conditional lower bounds.

\begin{theorem}[Complexity Dichotomy for $\uni{\constuncolholantprob{2}}(\mathcal{S})$ \cite{aivasiliotis_et_al:LIPIcs.ICALP.2025.7}]\label{thm:main_uncol}
    Let $\mathcal{S}$ be a finite set of signatures. We set $\mathcal{S}_0 = \{s \in \mathcal{S} \mid s(0) = 0\}$.
    \begin{itemize}
        \item[1.] If $\mathcal{S} \setminus \mathcal{S}_0$ is of type $\mathbb{T}[1]$, then $\uni{\constuncolholantprob{2}}(\mathcal{S})$ can be solved in FPT-near-linear time.
        \item[2.] Otherwise $\uni{\constuncolholantprob{2}}(\mathcal{S})$ is $\#\mathrm{W}[1]$-complete. If, additionally, $\mathcal{S} \setminus \mathcal{S}_0$ is of type $\mathbb{T}[\infty]$, then $\uni{\constuncolholantprob{2}}(\mathcal{S})$ cannot be solved in time $f(k)\cdot |V(\Omega)|^{o(k/\log k)}$, unless ETH fails. \qed
    \end{itemize}
\end{theorem}

Finally, the tractability criteria above are in fact explicit, in the sense that for a set $\mathcal{S}$ of signatures of type $\mathbb{T}[1]$ and a signature $s \in \mathcal{S}$, we know that $s$ has the following explicit representation: $s(n) = \alpha^n$, for some $\alpha \in \mathbb{C}$. The above implies that $s(0) = 1$ must be assumed, but in the context of holant problems it is well-known that scaling a signature by some constant, does not affect the complexity of the problem.

\subsection{Further Related Work on Holants}
In this part, we explore some well established approaches that appear, at first glance, to allow for deriving complexity classifications for holants on hypergraphs; however, as we argue below, they turn out to be insufficient for our (parameterised and classical) settings.

Firstly, as it has already been pointed out, holant problems on hypergraphs are equivalent to a natural class of symmetric counting constraint satisfaction problems (see \cref{appendix:VCSPsHolants}). Furthermore, readers familiar with classical holant theory might be aware that the extension of signature grids from graphs to hypergraphs has only limited benefits since classifications for higher arity $\#$CSPs had already been known when the holant framework was introduced --- e.g. the dichotomy for $\#$CSP~\cite{Bulatov13} first appeared in ICALP'08 while for holants~\cite{CaiGW16} in STOC'13. However, in the parameterised world where we only consider assignments that map precisely $k$ variables/edges to $1$, the situation is flipped: not much is known about the complexity of $\#$CSP for higher arities and the framework of parameterised holants on hypergraphs allows us to tackle such problems.

Secondly, holants on hypergraphs are also known to reduce, in the classical sense, to holants on graphs via incidence graphs; we provide the details for the sake of self-containment:
To this end, we let $\mathcal{S}$ be a finite set of signatures. We consider a hypergraph $G$ and an assignment $\{s_v\}_{v \in V(G)}$ of signatures from $\mathcal{S}$ to its vertices. It is folklore that any hypergraph can be represented as a bipartite incidence graph $B$ as follows. We let $V(B) = L\,\dot\cup\,R$ and consider two arbitrary bijections $\pi_L : V(G) \to L$ and $\pi_R : E(G) \to R$. Furthermore, for each $v \in V(G)$ and each $e \in E(G)$ such that $v \in e$, we add $\{\pi_L(v), \pi_R(e)\}$ to $E(B)$. Then, we consider the following assignment $\{s'_v\}_{v \in V(B)}$ of signatures to $V(B)$: for each $v \in L$ we assign to $v$ the signature of $\pi_L^{-1}(v)$ and for each $u \in R$, we assign to $u$ the signature $=_\ell : \{0,1\}^\ell \to \{0,1\}$ where $\ell = |\pi^{-1}_R(e)|$ and $=_\ell$ evaluates to one if and only if all of its arguments are equal. It can be readily verified that 

\begin{equation}\label{eq:introHolantReduction}
\uncolholant(G, \{s_v\}_{v \in V(G)}) = \uncolholant(B, \{s'_v\}_{v \in V(B)})\,,
\end{equation}
where $s'_v$ is either in $\mathcal{S}$ or is an equality signature $=_{\ell}$ for some $\ell \in \mathbb{N}$. In other words, holants on hypergraphs reduce to holants on (bipartite) graphs assuming ---besides $\mathcal{S}$--- freely available equality signatures.

It is known that signatures for equality (i.e., $=_\ell$) are part of the tractable family of affine signatures, implying the easiness of holant problems on hypergraphs. Moreover, hardness of holants on graphs transfer directly to holant problems on hypergraphs since the former is a special case of the latter.

However, we point out that the aforementioned trivial equivalence already becomes incomplete when introducing even mild natural restrictions to the problem, e.g., restricting holants to uniform hypergraphs of fixed rank (i.e., of upper-bounded maximum hyperedge-size), or, more importantly, by not allowing freely available equality signatures, which cases are considered and remedied in this work. Note that the latter case is realised when we only consider signatures of infinite domain which are our main consideration. It is easy to see that such signatures cannot model equality signatures.

\subsection{Our Contributions}
We provide exhaustive complexity dichotomy results both with respect to classical complexity theory ($\mathrm{P}$ vs.\ $\#\mathrm{P}$) and with respect to parameterised complexity theory ($\mathrm{FPT}$ vs.\ $\#\mathrm{W}[1]$). 

We first state our main results in the parameterised setting. As it has been already pointed out, the problem $\uni{\constuncolholantprob{d}}(\mathcal{S})$ is a ``slice'' of $\uncolholantprob(\mathcal{S})$. The above implies that we need not state two complexity dichotomies, since an FPT algorithm for the latter problem implies an FPT algorithm for the former (w.r.t. the respective parametrisation) and if any slice of $\uncolholantprob(\mathcal{S})$ is hard, then $\uncolholantprob(\mathcal{S})$ must already be hard.
Hence, our main result ---stated formally below--- consists of 
\begin{itemize}
\item[1.] an FPT algorithm for $\uncolholantprob(\mathcal{S})$, for finite sets of signatures $\mathcal{S}$ of type $\mathbb{T}[1]$, which also addresses the tractable case in which $\mathcal{S}$ contains signatures $s$ with $s(0) = 0$. 

\item[2.] $\#\mathrm{W}[1]$-hardness results for the problem $\uni{\constuncolholantprob{d}}(\mathcal{S})$, for any $d \geq 2$.
\end{itemize}

We recall that the parameterised complexity in the special case of $\uni{\constuncolholantprob{2}}(\mathcal{S})$, that is, the case of graphs, has already been established in~\cite{aivasiliotis_et_al:LIPIcs.ICALP.2025.7}.

\begin{theorem}[Parameterised Complexity Dichotomy]\label{thm:INTROmainParamThm}\label{main_para_intro}
Let $\mathcal{S}$ be a finite set of signatures. Let furthermore $\mathcal{S}_0 = \{s \in \mathcal{S} \mid s(0) = 0\}$. 
\begin{itemize}
\item[1.] If $\mathcal{S} \setminus \mathcal{S}_0$ is of type $\mathbb{T}[1]$, then $\uncolholantprob(\mathcal{S})$ can be solved in FPT-(near)-linear time, that is, in time $f(k, \hs) \cdot \tilde{O}(|\Omega|)$, where $\hs$ is the rank of the underlying hypergraph of $\Omega$.
\item[2.] If $\mathcal{S}\setminus \mathcal{S}_0$ is not of type $\mathbb{T}[1]$, then $\uncolholantprob(\mathcal{S})$ is $\#\mathrm{W}[1]$-hard.\\ In particular, for any $d \in \mathbb{Z}_{\geq 2}$, the problem $\uni{\constuncolholantprob{d}}(\mathcal{S})$ is $\#\mathrm{W}[1]$-hard. Furthermore, if $\mathcal{S} \setminus \mathcal{S}_0$ is of type $\mathbb{T}[\infty]$, then $\uni{\constuncolholantprob{d}}(\mathcal{S})$ cannot be solved in time $f(k)\cdot |V(\Omega)|^{o(k/\log(k))}$, for any computable function $f$, unless ETH fails. \qed
\end{itemize}
\end{theorem} 
Note that the previous result implies that all FPT cases are in fact FPT-near-linear time solvable.
We also classify the aforementioned problems in the non-parameterised setting. We point out that such a classification is novel \textit{even} for the case of graphs (that is, $d = 2$). 

\begin{theorem}[Classical Complexity Dichotomy]\label{thm:INTROPerfectMatchingReduction}\label{main_classical_intro}
Let $\mathcal{S}$ be a finite set of signatures and let $d\in \mathbb{Z}_{\geq 2}$.
\begin{itemize}
\item[1.] If $\mathcal{S}$ is of type $\mathbb{T}[1]$ (implying that all signatures $s\in \mathcal{S}$ satisfy $s(0)\neq 0$), then $\uni{\constuncolholantprob{d}}(\mathcal{S})$ can be solved in polynomial time.
\item[2.] Otherwise, $\uni{\constuncolholantprob{d}}(\mathcal{S})$ is $\#\mathrm{P}$-hard. \qed
\end{itemize}
\end{theorem}  
Note that the second part (2.) in the previous result implies that $\uni{\constuncolholantprob{d}}(\mathcal{S})$ is $\#\mathrm{P}$-hard whenever $\mathcal{S}$ contains a signature $s$ with $s(0)=0$ but $s(x)\neq 0$ for some $x>0$.

\begin{remark}[Finite signature sets vs.\ finite constraint languages]
   Finite signature sets are not the same as finite constraint languages, since our signatures have domain $\{0,1\}^\ast$, but constraints come with a fixed arity. 
\end{remark}

As we explain in \cref{appendix:VCSPsHolants} our theorems immediately imply the analogous classifications for (the classical and parameterised) complexity of $\textsc{p-VCSP}$.
We now give two applications of our main results:

\begin{corollary}
    For each prime $p$, the problem $\#\textsc{Codeword}_p$ is $\#\mathrm{P}$-hard and $\#\mathrm{W}[1]$-hard. Moreover, assuming ETH, $\#\textsc{Codeword}_p$ cannot be solved in time $f(k)\cdot |A|^{o(k/\log k)}$ for any function $f$, where $A$ is the input matrix and $k \in \mathbb{N}$ is the problem parameter.
\end{corollary}
\begin{proof}
    The problem $\#\textsc{Codeword}_p$ reduces from $\uncolholantprob(\{s_p\})$ where $s_p(x)=1$ if $x \equiv 0 \mod p$ and $s_p(x)=0$ otherwise. To verify the existence of this reduction, observe that for each vertex $v$ of the signature grid and its incident edges $e_1,\dots,e_\ell$, we just need to add the equation $\sum_{i=1}^\ell e_i \equiv 0 \mod p$; in fact, this reduction applies in both directions if we allow multiple copies of the same hyperedge in the signature grid (this happens if the matrix of the system has two identical columns).
    
    Clearly, $s_p(0)\neq 0$. Now consider $\chi(p,s_p)$. There is only one partition of $[p]$ that contains a block of size $p$: this is the coarsest partition $\top$ containing only one block $[p]$. Clearly, $s_p(|[p]|)=s_p(p)=1$. All further partitions of $[p]$ contain a block $B$ with $1\leq |B|\leq p-1$, hence $s(|B|)=0$. Consequently,
    \[\chi(p,s_p)= (-1)^{|\top|-1} (|\top|-1)! \cdot 1 \neq 0\,.\]
    For $p>2$, this shows that $\{s_p\}$ is of type $\mathbb{T}[\infty]$. For $p=2$, it is easy to verify that $\chi(3,s_p)\neq 0$, hence the type is also $\mathbb{T}[\infty]$. Intractability then immediately follows from Theorems~\ref{main_para_intro} and~\ref{main_classical_intro}, and from the fact that $\uni{\constuncolholantprob{d}}(\mathcal{S})$ reduces to $\uncolholantprob(\mathcal{S})$ for any $d$ and $\mathcal{S}$.
\end{proof}

For the next application, we consider the problem of counting $k$-matchings in $d$-uniform hypergraphs. It is well-known that this problem is $\#\mathrm{P}$-hard since it subsumes as a special case the problem of counting perfect matchings in $d$-uniform hypergraphs for which $\#\mathrm{P}$-hardness is known (in particular, the case $d = 2$ was shown by Valiant \cite{valiant1979complexity}, while the case $d > 2$ was shown later by Creignou \cite{creignou1996complexity}). For $d=2$, it took almost a decade to strengthen $\#\mathrm{P}$-hardness to $\#\mathrm{W}[1]$-hardness.\footnote{The $\#\mathrm{W}[1]$-hardness of counting $k$-matchings in graphs was first conjectured by Flum and Grohe in 2004~\cite{FlumG04}, and proved by Curticapean in 2013~\cite{Curticapean13}.} We obtain $\#\mathrm{W}[1]$-hardness for all $d\geq 2$ as an easy application of our main result.

\begin{corollary}\label{cor:matchingsComplexityClass}
    For each $d\in \mathbb{Z}_{\geq 2}$, the problem of counting $k$-matchings in a $d$-uniform hypergraph $G$ is $\#\mathrm{W}[1]$-hard and cannot be solved in time $f(k)\cdot |G|^{o(k/\log k)}$ for any function~$f$, unless ETH fails.
\end{corollary}
\begin{proof}
    We just need to show that the signature set $\{\mathsf{hw}_{\leq 1}\}$ is of type $\mathbb{T}[\infty]$. Here, $\mathsf{hw}_{\leq 1}(x)=1$ if $x\leq 1$ and $\mathsf{hw}_{\leq 1}(x)=0$ otherwise. The claim then follows immediately from Theorem~\ref{main_para_intro}, since the problem is identical to $\uni{\constuncolholantprob{d}}(\{\mathsf{hw}_{\leq 1}\})$. It is easy to verify that $\chi(3,\mathsf{hw}_{\leq 1})\neq 0$. Moreover, $\mathsf{hw}_{\leq 1}(0)\neq 0$. Hence $\{\mathsf{hw}_{\leq 1}\}$ is of type $\mathbb{T}[\infty]$.
\end{proof}

In particular, \Cref{cor:matchingsComplexityClass} also follows from the novel complete complexity classification of $\textsc{Factor}_d(S)$ which we establish below. 

\begin{corollary}
Let $S \subseteq \mathbb{Z}_{\geq 0}$ be a finite set. For any $d \in \mathbb{Z}_{\geq 2}$,
\begin{enumerate}
\item $\text{\sc{Factor}}_d(S)$ is $\#\mathrm{P}$-hard.
\item If $0 \notin S$, $\text{\sc{Factor}}_d(S)$ is solvable in FPT-(near)-linear time, that is, in time $f(k) \cdot \tilde{O}(|G|)$, for some computable function $f$.
\item If $0 \in S$, $\text{\sc{Factor}}_d(S)$ is $\#\mathrm{W}[1]$-hard and cannot be solved in time $f(k) \cdot |V(G)|^{o(k/\log(k))}$, for any computable function $f$, unless ETH fails.
\end{enumerate}
\end{corollary}
\begin{proof}
Let $(G, k)$ be an instance of $\textsc{Factor}_d(S)$. Consider the signature $t_S : \mathbb{Z}_{\geq 0} \to \{0,1\}$ where $t_S(x) = 1$ if and only if $x \in S$. It is easy to verify that $\textsc{Factor}_d(S)$ is equivalent\footnote{Clearly, with respect to parameterised as well as non-parameterised polynomial-time Turing-reductions.} to $\uni{\constuncolholantprob{d}}(\{t_S\})$ via a reduction that associates to $(G, k)$ a holant instance $(\Omega, k)$, where the underlying hypergraph of $\Omega$ is $G$ and each vertex of $\Omega$ is equipped with signature $t_S$.

First, note that if $0 \notin S$ (implying that $t_S(0) = 0)$, it follows from \Cref{thm:INTROmainParamThm} that $\textsc{Factor}_d(S)$ is solvable in FPT-(near)-linear time. Furthermore, \Cref{thm:INTROPerfectMatchingReduction} implies that for any $d \in \mathbb{Z}_{\geq 2}$, $\textsc{Factor}_d(S)$ is $\#\mathrm{P}$-hard.

Now, assume that $0 \in S$. It has been shown in (the full version of) \cite[Corollary 6.22]{aivasiliotis_et_al:LIPIcs.ICALP.2025.7} that $t_S$ is of type $\mathbb{T}[\infty]$. Hence, \Cref{thm:INTROmainParamThm} implies that for any $d \in \mathbb{Z}_{\geq 2}$, $\textsc{Factor}_d(S)$ is $\#\mathrm{W}[1]$-hard and cannot be solved in time $f(k)\cdot |V(G)|^{o(k/\log(k))}$, for any computable function $f$, unless ETH fails. Furthermore, \Cref{thm:INTROPerfectMatchingReduction} implies that for any $d \in \mathbb{Z}_{\geq 2}$, $\textsc{Factor}_d(S)$ is $\#\mathrm{P}$-hard.

\end{proof}

\subsection{Our Techniques}

Before we conclude our introduction, we give a brief exposition of the machinery we are employing for deriving our results.

\subsubsection{Gadgets}

One of the tools that we employ are \textit{gadgets}. Gadgets have mainly been used in counting problems outside of the holant framework. We point out that our gadgets will be much more general than what the familiar reader might know as \textit{matchgates} which have been extensively used within the holant framework, the details of which are beyond the scope of this paper.

Roughly speaking, gadgets are signature grids $\Omega_F$ that are \textit{embedded} on other signature grids $\Omega_G$, typically by identifying vertices of (a copy of) $\Omega_F$ with vertices of $\Omega_G$ and possibly forming new hyperedges (e.g., between different copies of gadgets). The above implies that gadgets induce a mapping between signature grids (possibly allowing different signatures). In this work, the construction of gadgets is based on the realisability of \textit{degree sequences} of certain suitable hypergraphs that also meet some structural properties. For example, some of our gadgets rely on the existence of infinitely many $d$-uniform $b$-regular hypergraphs for $d, b \geq 2$ that are also connected.

For our parameterised holant problems, we will consider a more general mapping $f$ that maps a parameterised holant instance $(\Omega, k)$ (instead of just a signature grid) into another $(\Omega', k') = f(\Omega, k)$, that can be computed efficiently with respect to parameterised or classical complexity. The choice of the new parameter $k'$ is based on the following observation: the most useful gadgets are naturally those that preserve the holant value, that is, $\uncolholant(\Omega', k') =   \uncolholant(\Omega, k)$, hence the need for possibly choosing a different parameter. However, the previous requirement is rather strict, and we would also be content if the holant values of the respective holant instances are \textit{efficiently comparable}, in the sense that there is an efficient way to compute $\uncolholant(\Omega, k)$ by computing only $\uncolholant(\Omega', k')$. For our problems, we devise gadgets such that 
\[
\uncolholant(\Omega', k') = \alpha(\Omega, k)\cdot\uncolholant(\Omega, k) + \beta(\Omega, k)\,,
\]
where $\alpha, \beta$ are efficiently computable functions, which implies that we can efficiently compute $\uncolholant(\Omega, k)$ via computing $\uncolholant(\Omega', k')$. More formally, we are interested in those gadgets that induce a mapping $f$ which will imply a polynomial time (or FPT) Turing-reduction between two holant problems (that may not necessarily allow for the same signatures).  

\subsubsection{Graph Motif Parameters}

Another tool in our machinery is the framework of \textit{graph motif parameters} and the property of \textit{complexity monotonicity} that they enjoy, introduced by Curticapean, Dell and Marx \cite{CurticapeanDM17}. 

A \textit{graph parameter} is any function $f : \mathcal{G} \to \mathbb{Q}$ over graphs that is invariant under isomorphisms. A \textit{graph motif parameter} is a graph parameter $f$ that admits the following expansion: there exist pair-wise non-isomorphic \textit{pattern} graphs $H_1, \dots, H_t$ and \textit{coefficients} $\zeta_1, \dots, \zeta_t \in \mathbb{Q}\setminus\{0\}$ such that for any graph $G$, $f(G) = \sum_{i = 1}^{t}\zeta_i\cdot\#\homs{H_i}{G}$, where $\#\homs{H}{G}$ counts the number of homomorphisms, i.e., edge-preserving mappings from $V(H)$ to $V(G)$. Is it known that such an expansion is unique \cite{Lovasz12}. We will sometimes refer to it as the expansion of $f$ into the \textit{homomorphism basis} implying that homomorphism counts span the vector space of graph motif parameters, the details of which are beyond the scope of this introduction.

A well-known result, first established in~\cite{DiazST02} (see also~\cite{CurticapeanDM17}), states that evaluating $\#\homs{H}{G}$ can be done in time $\mathsf{poly}(|V(H)|)\cdot n^{\mathsf{tw}(H)+1}$, where $\mathsf{tw}(H)$ is the treewidth of $H$ and $n = |V(G)|$. Hence, any fixed graph motif parameter can be evaluated in polynomial time, since the pattern graphs are fixed. It is therefore more interesting to investigate the complexity of evaluating a graph motif parameter $f$ from an infinite ---and assumed recursively enumerable--- family $\mathcal{F}$ of graph motif parameters. 

From our discussion above, it follows that if we can bound the treewidth of the pattern graphs featured in \textit{any} graph motif parameter $f \in \mathcal{F}$, then the problem of evaluating any $f \in \mathcal{F}$ is computable in polynomial time.
From the perspective of parameterised complexity, the result above further implies quite straightforwardly that the problem of evaluating $f \in \mathcal{F}$, when parameterised by the maximum number of vertices of the pattern graphs in $f$, is fixed-parameter tractable when the treewidth of pattern graphs in any $f \in \mathcal{F}$ is bounded. In a remarkable result due to Curticapean, Dell and Marx~\cite{CurticapeanDM17} it was shown that the boundedness of the treewidth is the right dichotomy criterion for the aforementioned parameterised problem, in the sense that, if $\mathcal{F}$ contains graph motif parameters with pattern graphs of unbounded treewidth, then the problem is $\#\mathrm{W}[1]$-hard. 
Specifically, it is shown in~\cite{CurticapeanDM17} that, if we can evaluate a graph motif parameter $f$, then we can also evaluate the mapping $\#\homs{H}{\star} : G \mapsto \#\homs{H}{G}$, via FPT Turing-reductions, for any pattern graph $H$ of $f$ as long as its corresponding coefficient is non-zero. This property was coined as \textit{complexity monotonicity} and since its inception, it has yielded an avenue for charting the complexity of many counting problems that can be formulated within the framework of graph motif parameters. In particular, the problem of evaluating a graph motif parameter $f \in \mathcal{F}$ then becomes a combinatorial problem ---arguably of very challenging nature---, since one has to investigate which graphs are supported in the expansion of $f$ in the homomorphism basis, i.e., their coefficient is non-zero. 

Interestingly, Curticapean showed very recently that ---under some technical restrictions--- a polynomial-time analogue of complexity monotonicity is feasible, that is, extracting homomorphism counts from the homomorphism expansion of a graph motif parameter can be done in polynomial time \cite{curticapean2024count}. If the technical conditions apply, this result can be used to infer $\#\mathrm{P}$-hardness for the problem of evaluating graph motif parameters.

The property of complexity monotonicity also applies to graph motif parameters over more general structures, e.g., hypergraphs that may also feature vertex-colourings and it is otherwise known as \textit{Dedekind interpolation} (see e.g., the full version of \cite[Theorem 18]{10.1145/3564246.3585204}).

In particular, signature grids can be equivalently seen as vertex-coloured hypergraphs where the colours are given by the signatures. Letting $\mathcal{G}(\mathcal{S})$ denote the set of all signature grids over $\mathcal{S}$, we can show that for any $k \in \mathbb{Z}_{\geq 0}$, the mapping $f : \Omega \in \mathcal{G}(\mathcal{S}) \to \uncolholant(\Omega, k)$ is a graph motif parameter (see \Cref{sec:HolantsInHomBasis} for a proof) and thus we can formulate and study holant problems within the framework of graph motif parameters.

\subsection{Conclusion and Future Work}\label{sec:Conclusion}

Following the recent advancements of Aivasiliotis et al.\ ~\cite{aivasiliotis_et_al:LIPIcs.ICALP.2025.7}, in the line of work for charting the complexity of parameterised holant problems introduced by Curticapean~\cite{Curticapean15}, we introduce generalised parameterised holant problems that can model parameterised VCSPs. The latter can be seen as the counting counterpart (that assumes complex values) of the parameterised satisfiability problem introduced and classified by Marx in \cite{Marx05}, the parametrisation of which is given by the solution-size which is restricted to some $k \in \mathbb{Z}_{\geq 0}$. 

We fully classify the complexity of parameterised VCSPs (as well as the complexity of the non-parameterised equivalent problems) for unbounded symmetric constraints. In particular, several natural variants are also considered subject to additional restrictions, e.g., when no two variables can appear in the same set of constraints and/or variables appear in exactly $d$ constraints, that is the \textit{occurrence} of every variable is precisely $d$.

As mentioned above, in our setting we assumed unbounded arity (of constraints) and in particular, we further assumed bounded occurrence of (variables). For the rest of this section, we would like to elaborate on the other possible configurations of the problem with respect to the (un)boundedness of the arity and the occurrence. 

First, in the context of unbounded arity, we show that trivially tractable instances of the problem already become $\mathrm{W}[2]$-hard if we assume unbounded occurrences (see \Cref{sec:W2hardness}). In this context, an FPT Turing-reduction from parameterised VCSPs to the problem of evaluating graph motif parameters would imply the collapse of $\mathrm{W}$-hierarchy, since the expansion of instances of the latter problem into the homomorphism basis includes patterns, the size of which is not bounded by the (only) parameter $k$. This shows that a novel framework becomes necessary for tackling those instances of parameterised \#CSPs, which is an exciting avenue for future work.

Next, in the context of bounded arity, if we further assume that occurrence is bounded then we can infer trivial fixed-parameter tractability results (for all instances) which follow from standard algorithmic techniques in parameterised complexity such as the bounded search-tree paradigm.

The last case, that is the combination of bounded arity and unbounded occurrence, is not remedied in this work as it is orthogonal to our main setting of interest. Hence we leave this family of parameterised \#CSPs open for investigation in future work.

\newpage

\tableofcontents

\newpage

\section{Technical Overview}
We proceed with giving an overview of our technical contributions hoping that it will facilitate the navigation of the reader through the proofs and the techniques therein. In this section, we will focus solely on the analysis of our hardness results, which we consider the heart of our contributions. Our upper bounds, although novel, follow from standard algorithmic techniques, for the details of which we refer the reader to \Cref{sec:TractabilitySubsection}.
\subsection{Parameterised Complexity: \ensuremath{\#\mathrm{W}[1]}-Hardness and Lower Bounds}\label{sec:introHardnessParam}
The $\#\mathrm{W}[1]$-hard instances of $\uncolholantprob(\mathcal{S})$ (recall that they are defined over arbitrary hypergraphs) are essentially determined by the hard instances of $\uni{\constuncolholantprob{2}}(\mathcal{S})$, since the latter problem is subsumed by the former and since the tractable instances of $\uncolholantprob(\mathcal{S})$ match those of $\uni{\constuncolholantprob{2}}(\mathcal{S})$ (see \Cref{sec:TractabilitySubsection} for the aforementioned tractability results). Both, more interesting and more challenging is the case of uniform hypergraphs, that is, the hardness of $\uni{\constuncolholantprob{d}}(\mathcal{S})$, for $d \geq 3$, which ---overviewed below--- is formally established in \Cref{sec:hardnessSubsection}.

Our proof for the hardness of $\uni{\constuncolholantprob{d}}(\mathcal{S})$ relies generally on a case distinction along additional properties of the set of allowed signatures.
The first, and also the simplest case, assumes that $\mathcal{S}$ contains a signature $s$ with $s(1) \neq 0$. Under this assumption, we can construct a simple gadget to embed on the underlying graph of instances of $\uni{\constuncolholantprob{2}}(\mathcal{S})$ and transform those instances into instances of $\uni{\constuncolholantprob{d}}(\mathcal{S})$, for any $d \geq 3$, thus yielding a reduction from $\uni{\constuncolholantprob{2}}(\mathcal{S})$ to $\uni{\constuncolholantprob{d}}(\mathcal{S})$. We outline the gadget construction below.

\begin{proof}[Gadget construction for signature sets containing $s$ with $s(1) \neq 0$] Let $s$ be a signature and let $(\Omega, k)$ be an instance of $\uni{\constuncolholantprob{2}}(\{s\})$ with underlying graph $G$. Let $G'$ denote the $d$-uniform hypergraph obtained from $G$ by adding to each edge $e \in E(G)$, $d-2$ new vertices (that is, all new vertices have degree 1). We further equip the new vertices with signature $s$ and obtain a signature hypergrid $\Omega'$. Note that $(\Omega', k) \in \uni{\constuncolholantprob{d}}(\{s\})$. We have
    \begin{align*}
    \uncolholant(\Omega', k) & = \sum_{\substack{A \subseteq E(G') \\ |A| = k}}\prod_{v \in V(G')}s(|A \cap E_{G'}(v)|) \\ & = s(1)^{k(d-2)} \underbrace{\sum_{\substack{A \subseteq E(G') \\ |A| = k}}\prod_{v \in V(G)}s(|A\cap E_{G'}(v)|)}_{= \uncolholant(\Omega, k)}\,.\end{align*}
    So, as long as $s(1) \neq 0$, it follows that $\uncolholant(\Omega, k) = s(1)^{-k(d-2)}\cdot\uncolholant(\Omega', k)$ which thus can be computed in polynomial time using our oracle for $\uncolholant(\Omega', k)$.
\end{proof}
Now, the question is whether the previous analysis can be adapted for any signature set of type $\mathbb{T}[2]$ or $\mathbb{T}[\infty]$, including signatures $s$ with $s(1)=0$. The answer is partially affirmative, that is, we can construct more complicated gadgets for signatures that satisfy $s(2) \neq 0$ (and not necessarily $s(1) \neq 0$) and thus derive the desired reduction from $\uni{\constuncolholantprob{2}}(\mathcal{S})$ to $\uni{\constuncolholantprob{d}}(\mathcal{S})$ for those signatures (in particular, any set $\mathcal{S}$ of type $\mathbb{T}[2]$ contains such signatures), as outlined below.

\begin{proof}[Gadget construction for signature sets containing $s$ with $s(2) \neq 0$] Let $B$ be a connected $d$-uniform hypergraph with precisely two vertices $x, y$ having degree 1 (the existence of which is established below). We assume that $x, y$ are adjacent (that is, they are contained in the same hyperedge) and that every other vertex has degree 2.
    Let $G'$ denote the hypergraph obtained by replacing each edge $\{u, v\} \in E(G)$ with a copy of $B$ by identifying $v$ (resp. $u$) with $x$ (resp. $y$). We also equip each new vertex with signature $s$, obtaining a new signature grid $\Omega'$.
    
    To compute $\uncolholant(\Omega', k')$, for any $k'$, we observe that any $k'$-set $A \subseteq E(G')$ that contains at least one but not all of the hyperedges of any copy of $B$, has zero contribution to $\uncolholant(\Omega', k')$. To see this, recall that $s(1) = 0$ and note that there would exist at least one vertex $w \in B$ ($w \neq x, y$) such that $|A \cap E_{G'}(w))| = 1$, since $B$ is connected.

    We set $k' = k \cdot |E(B)|$ and for each $e \in E(G)$, we write $B^e \subseteq G'$ for the copy of $B$ that $e$ is replaced with in $G'$. We let $\mathcal{A}'$ denote the set of all $k'$-sets $A' \subseteq E(G')$ for which there exist $e_1, \dots, e_k \in E(G)$ such that $A' = \bigcup_{i = 1}^{k}E(B^{e_i})$. It follows from the discussion above that

    \[
    \uncolholant(\Omega', k') = \sum_{ A'\in \mathcal{A}'}\prod_{v \in V(G')}s(|A' \cap E_{G'}(v)|)\,.
    \]
    We let $\binom{E(G)}{k}$ denote all $k$-subsets of $E(G)$.
    Let $a:\mathcal{A}' \to \binom{E(G)}{k}$ denote the mapping assigning a set $A'= \bigcup_{i = 1}^{k}E(B^{e_i})$ to the set $\{e_1, \dots, e_k\}$. By construction of $\mathcal{A}'$ it is easy to see that $a$ is a bijection. Moreover, observe that 
    \[
    \prod_{v \in V(G')}s(|A' \cap E_{G'}(v)|) = s(2)^{k(|V(B)|-2)}\prod_{v \in V(G)}s(|a(A') \cap E_G(v)|)\,.\] 
    So, we have 
    \begin{align*}
    \uncolholant(\Omega', k') & = \sum_{\substack{A' \in a(\mathcal{A}')}}s(2)^{k(|V(B)|-2)}\prod_{v \in V(G)}s(|A' \cap E_{G}(v)|) \\ & = s(2)^{k(|V(B)|-2)} \cdot \uncolholant(\Omega, k)\,,
    \end{align*}
    and since $s(2) \neq 0$, we have $\uncolholant(\Omega, k) = s(2)^{-k(|V(B)|-2)}\cdot \uncolholant(\Omega', k')$, which thus can be computed in polynomial time using our oracle for $\uncolholant(\Omega', k')$ (since $|V(B)| + |E(B)|$ is constant). 

    To complete the reduction, what remains is to show that such a hypergraph $B$ always exists.
    By \cite[Theorem 3]{frosini2021new}, it follows that any degree sequence with sufficiently many entries of $2$ and precisely $d-2$ entries of $1$ (and no entries larger than $2$) is realizable by a $d$-uniform hypergraph.
    Let hence $W$ be a $d$-uniform hypergraph of maximum degree 2 having precisely $d-2$ vertices of degree 1. First, we delete all connected components of $W$ in which each vertex has degree $2$; consequently, each remaining connected component contains at least one vertex of degree $1$.
    Finally, we obtain $B$ by first adding two fresh vertices $x$ and $y$, and then forming a new hyperedge that contains all $(d-2)$ degree-1 vertices along with $x$ and $y$. Since $d$ is a constant, it follows that we can also compute $B$ in constant time (irrespective of whether the mapping $d \mapsto B$ that constructs $B$ is even decidable).
\end{proof}

It seems that our gadgets cannot be extended so as to resolve the remaining hard signatures, which stems from certain issues arising when applying \cite[Theorem 3]{frosini2021new}, that is, there are certain divisibiliy requirements that \textit{cannot} be met and thus realisability of hypergraphs that our gadgets are based upon cannot be guaranteed. However, it turns out that all other cases can be remedied using the framework of graph motif parameters, in which cases the combinatorial analysis involved will turn out to be very clean. In particular, let $s$ be a signature such that $s(1) = s(2) = 0$ and let $b$ be the largest number such that $s(b) \neq 0$ and $s(i) = 0$, for each $0 < i < b$. We show that for any instance $(\Omega, k)$ of $\uni{\constuncolholantprob{d}}(\{s\})$, $d$-uniform $b$-regular hypergraphs on $k$ hyperedges survive in the expansion of $\uncolholant(\Omega, k)$ into the homomorphism basis. More formally, letting $\mathcal{F}_{d, b}$ denote the class of all $d$-uniform $b$-regular hypergraphs, we show that $\#\textsc{HOM}(\mathcal{F}_{d, b})$ reduces to $ \uni{\constuncolholantprob{d}}(\{s\})$. 

In particular, it is known that for $d \geq 2$ and $b \geq 3$, $\mathcal{F}_{d, b}$ contains hypergraphs of unbounded treewidth. Furthermore, we point out that our reductions above are linear FPT Turing-reductions, implying that via our reductions, apart from $\#\mathrm{W}[1]$-hardness, fine-grained lower bounds on the runtime complexity transfer as well. 

Hence, what remains is to investigate whether the known fine-grained conditional lower-bounds of $\uni{\constuncolholantprob{2}}(\mathcal{S})$ also hold for $\#\textsc{HOM}(\mathcal{F}_{d,b})$. We answer this in the affirmative. For this, we formally prove the technical lemma below that follows from \cite{dalmau2004complexity} and \cite{Marx10} and apply it for $\mathcal{F}_{d, b}$ which as shown in this work, contains hypergraphs the treewidth of which is linear in their size, yielding the (almost) matching lower bounds.

\begin{lemma}
Let $r\in \mathbb{Z}_{\geq 2}$ be a fixed integer, and let $\mathcal{H}$ be a recursively enumerable class of $r$-uniform hypergraphs. If $\mathcal{H}$ has unbounded treewidth, then $\#\text{\sc{Hom}}(\mathcal{H})$ is $\#\mathrm{W}[1]$-hard and cannot be solved in time $f(|H|)\cdot |V(G)|^{o(\mathsf{tw}(H)/\log(\mathsf{tw}(H)))}$, for any function $f$, unless ETH fails.  This remains true even if each input $(H,G)$ to $\#\text{\sc{Hom}}(\mathcal{H})$ satisfies that $G$ is $r$-uniform as well.
\end{lemma}

\subsection{Results On \ensuremath{\#\mathrm{P}}-hardness}

In \Cref{sec:NonParamHardness}, we study the ``classical complexity'' of the problem $\constclassHolant{d}(\mathcal{S})$. In particular, we show that \begin{itemize} 

\item[1.] all $\#\mathrm{W}[1]$-hard instances of $\uni{\constuncolholantprob{d}}(\mathcal{S})$ are $\#\mathrm{P}$-hard for $\constclassHolant{d}(\mathcal{S})$ and,

\item[2.] any signature set $\mathcal{S}$ containing a signature $s$ with $s(0) = 0$ renders the problem $\constclassHolant{d}(\mathcal{S})$ $\#\mathrm{P}$-hard, for any $d \geq 2$.
\end{itemize}

We point out that for the remaining cases, $\constclassHolant{d}(\mathcal{S})$ is solvable in polynomial time, since for these cases, the FPT algorithms for $\uni{\constuncolholantprob{d}}(\mathcal{S})$ (see \Cref{sec:TractabilitySubsection}) are in fact polynomial time.

\subsubsection{\ensuremath{\#\mathrm{P}}-hardness results for signature sets containing \ensuremath{s} with \ensuremath{s(0) = 0}} 

In \Cref{sec:HardnessZeroSigs} we show that $\constclassHolant{d}(\mathcal{S})$ is $\#\mathrm{P}$-hard whenever $\mathcal{S}$ contains a signature $s$ with $s(0) = 0$. As usual, for hardness results, it suffices to restrict to signature sets with a single signature $s$, in which case we also assume that $s(0) = 0$. 

We derive our hardness results by reducing from the problem $\#\textsc{PerfectMatching}^d$, which is given a $d$-uniform hypergraph $G$ to compute the number of perfect matchings of $G$, denoted as $\#\mathsf{PerfMatch}(G)$, where recall that a perfect matching in a hypergraph $G$ is a partition of $V(G)$ into pairwise disjoint hyperedges of $G$. The problem $\#\textsc{PerfectMatching}^d$ is known to be $\#\mathrm{P}$-hard for any $d \geq 2$ (in particular, the case $d = 2$ was shown by Valiant \cite{valiant1979complexity}, while the case $d > 2$ was shown later by Creignou and Hermann~\cite{creignou1996complexity}).

In a nutshell, for our reductions mentioned above, we proceed as follows. We construct gadgets ---that are based on the realisability of certain uniform regular hypergraphs--- which we embed on the instances $G$ of $\#\textsc{PerfectMatching}^d$ and obtain new $d$-uniform hypergraphs $G'$. Then, the vertices of $G'$ are equipped with signature $s$ yielding the signature grid~$\Omega$. We show that for a suitable choice of $k$, $\uncolholant(\Omega, k)$ is equal to $\#\mathsf{PerfMatch}(G)$, up to efficiently computable multiplicative factors, and thus we can efficiently extract $\#\mathsf{PerfMatch}(G)$ from $\uncolholant(\Omega, k)$, which proves correctness of the reduction.

\subsubsection{\ensuremath{\#\mathrm{P}}-hardness results for \ensuremath{\constclassHolant{d}(\mathcal{S})}: CFI Filters for \ensuremath{d = 2}}
The parameterised complexity of $\uni{\constuncolholantprob{2}}(\mathcal{S})$ was resolved using the framework of graph motif parameters. In \Cref{sec:HardnessCFI}, we show that the same framework can also be used to derive $\#\mathrm{P}$-hardness, following the recent results regarding the polynomial-time analogue of the property of complexity monotonicity via ``CFI-Filters'' by Curticapean \cite{curticapean2024count} and suitably adapting the combinatorial analysis of Aivasiliotis et al.\ \cite{aivasiliotis_et_al:LIPIcs.ICALP.2025.7} from the parameterised setting.

To this end, we let $\mathcal{G}$ denote the family of all pairwise non-isomorphic graphs and $\mathcal{G}(\mathcal{S})$ the family of all pairwise non-isomorphic signature grids $\Omega_G$ over $\mathcal{S}$ where $G \in \mathcal{G}$. As already mentioned, a signature grid can be equivalently seen as a vertex-coloured hypergraph, where the colours of the vertices are given by the signatures.

Recall that for any $k$, the mapping $f_{k, \mathcal{S}} : \Omega_G \in \mathcal{G}(\mathcal{S}) \mapsto \uncolholant(\Omega_G, k)$ is a graph motif parameter (see \Cref{sec:HolantsInHomBasis}), that is, for any $\Omega_G \in \mathcal{G}(\mathcal{S})$
\[
f_{k, \mathcal{S}}(\Omega_G) = \sum_{\Omega_H \in \mathcal{G}(\mathcal{S})}\zeta_{k, \mathcal{S}}(\Omega_H)\cdot\#\homs{\Omega_H}{\Omega_G}\,
\]
and there are finitely many signature grids $\Omega_H$ such that $\zeta_{k, \mathcal{S}}(\Omega_H) \neq 0$. 

For deriving hardness results, it suffices to consider signature sets $\mathcal{S}$ containing only one signature $s$. Clearly, for any two signature grids $\Omega_H, \Omega_G$ over a single signature $s$, it holds that $\#\homs{\Omega_H}{\Omega_G} = \#\homs{H}{G}$. Hence, we may simplify notation and redefine $f_{k, \mathcal{S}}$ as follows:
\[
f_{k, s}(G) = \sum_{H \in \mathcal{G}}\zeta_{k, s}(H)\cdot\#\homs{H}{G}\,,
\]
where we write $f_{k, s}(G)$ (resp. $\zeta_{k,s}(H)$) instead of $f_{k, \{s\}}(\Omega_G)$ (resp. $\zeta_{k, \{s\}}(\Omega_H)$).

The combinatorial analysis of the coefficients $\zeta_{k, s}(H)$ in \cite{aivasiliotis_et_al:LIPIcs.ICALP.2025.7},\footnote{We refer the reader also to the full version \cite{aivasiliotis2024parameterisedholantproblems} of~\cite{aivasiliotis_et_al:LIPIcs.ICALP.2025.7}.} revealed the following regarding the patterns surviving in the expansion of $f_{k, s}$ into the homomorphism basis:

\begin{lemma}[\cite{aivasiliotis2024parameterisedholantproblems}]
Let $\{s\}$ be of type $\mathbb{T}[2]$. Let $H$ be a clique on $d+1$ vertices and set $k = d^2-1$. It holds $\zeta_{k, s}(H) \neq 0$.    
\end{lemma}

\begin{lemma}[\cite{aivasiliotis2024parameterisedholantproblems}]\label{lem:IntroCliques}
Let $\{s\}$ be of type $\mathbb{T}[\infty]$ and $r \geq 3$ such that $\chi(r, s) \neq 0$. Let $H$ be a $r$-regular graph and set $k = |E(H)|$. It holds $\zeta_{k, s}(H) \neq 0$.
\end{lemma}

In essence, the above results imply that, whenever $\{s\}$ of type $\mathbb{T}[2]$ or $\mathbb{T}[\infty]$, there is a family of graphs $\mathcal{H}$  of unbounded treewidth, such that $\#\textsc{HOM}(\mathcal{H})$ reduces to $\uni{\constuncolholantprob{2}}(\mathcal{S})$. In particular, if $\{s\}$ is of type $\mathbb{T}[2]$ (resp. $\mathbb{T}[\infty]$), the aforementioned class of graphs $\mathcal{H}$ is the family of all $r$-regular graphs (resp. the family of all cliques). 

The result of Curticapean~\cite{curticapean2024count} implies that the above reduction is also a polynomial time Turing-reduction ---subject to some alterations to be discussed momentarily--- as long as the maximum degree of graphs in $\mathcal{H}$ is upper-bounded by some constant. Indeed, this applies to the case of $r$-regular graphs since $r$ is a constant depending only on $s$, but cliques do not meet this requirement. However, we are able to show the following result by adapting accordingly the combinatorial analysis of \cref{lem:IntroCliques}:

\begin{lemma} Let $\{s\}$ be of type $\mathbb{T}[2]$. Let $F$ be a 4-regular graph and set $k = 3\cdot |V(F)|$. It holds $\zeta_{k, s}(F) \neq 0$.    
\end{lemma}

One more technicality that differentiates the poly-time version of complexity monotonicity from the parameterised one, is that the former extracts homomorphism counts for \textit{colourful} patterns. To break down this, a colourful graph $H$ is a graph that comes with a \textit{bijective} colouring $\nu_H$ of its vertices. For a recursively enumerable class of \textit{colourful} graphs, we write $\#\textsc{ColHom}(\mathcal{H})$ for the problem that is given a colourful graph $H \in \mathcal{H}$ and a \textit{coloured} graph $G$, to compute $\#\mathsf{ColHom}(H \to G)$, which is the number of all homomorphisms from $H$ to $G$ that also preserve the colour of the vertices. Hence, we derive a reduction from
$\#\textsc{ColHom}(\mathcal{H})$ to $\constclassHolant{d}(\{s\})$,
where $\mathcal{H}$ is the class of all $4$-regular \textit{colourful} graphs (resp. the class of all $r$-regular \textit{colourful} graphs) if $\{s\}$ is of type $\mathbb{T}[2]$ (resp. $\mathbb{T}[\infty]$).

Finally, we complete our proof by showing that counting homomorphisms from \textit{colourful} $R$-regular graphs, for any $R \geq 3$, is $\#\mathrm{P}$-hard, a result that 
% ---to the best of our knowledge--- 
was not known.

\subsubsection{\ensuremath{\#\mathrm{P}}-Hardness Results for \ensuremath{\constclassHolant{d}(\mathcal{S})}: Global Gadgets for \ensuremath{d > 2}} 

In \Cref{sec:HardnessCFIHypergraphs}, similarly to the parameterised setting, our goal ---at first--- is to establish a reduction from $\constclassHolant{2}(\mathcal{S})$ to $\constclassHolant{d}(\mathcal{S})$, for any $d \geq 3$. In fact, it can be easily verified that the reductions we showed in \Cref{sec:HardnessCFI} yielded by the gadgets we constructed, are polynomial time Turing-reductions. Hence, for any signature set $\mathcal{S}$ containing a signature $s$ such that $s(1) \neq 0$ or $s(2) \neq 0$, the problem $\constclassHolant{d}(\mathcal{S})$ is $\#\mathrm{P}$-hard.
% which is directly transferred from $\constclassHolant{2}(\mathcal{S})$ via the aforementioned reductions.

For the remaining cases, that is, for signatures $s$ such that $s(1) = s(2) = 0$, we note that a similar approach as in the parameterised setting, that is, using the framework of graph motif parameters, seems to be non-trivial. Although we have already argued that a polynomial time analogue of complexity monotonicity is feasible for graph motif parameters over graphs, it is non-trivial how such a result can be extended to graph motif parameters over hypergraphs, which is beyond the scope of this work.

Instead, we follow a different approach and construct global gadgets in order to reduce from the problem $\#\textsc{PerfectMatching}^{d}$. For the construction of our gadgets, we rely on the existence of infinitely many $d$-uniform $b$-regular graphs that are also connected, which we show in the same section. Finally, the reason as to why it is possible to construct such gadgets in the non-parameterised setting (that still cannot be used in the parameterised setting) is ---roughly speaking--- due to the fact we have more flexibility on the values of $k$, that is, $k$ can even be comparable to the size of the input which is not allowed by FPT Turing-reductions.

\section{Preliminaries}\label{sec:prelims}

\subsection{Hypergraphs}
A hypergraph is an unordered pair $(V, E)$ of sets, where $V$ is the set of \textit{vertices} and $E \subseteq 2^{V}$ is the set of \textit{hyperedges}. Given a hypergraph $G$, we write $V(G), E(G)$ for its vertex- and hyperedge-set respectively. We refer to the maximum size of a hyperedge in $E(G)$, as the \textit{rank} of $G$, denoted as $\mathsf{rank}(G)$. 

Hypergraphs may also come with a vertex-colouring $\nu : V(G) \to C$, for some finite set $C$ of colours. We call a hypergraph \textit{$C$-coloured} to indicate that its vertices are coloured with elements of $C$.

\subsubsection{Isomorphic Hypergraphs}
Two hypergraphs $F, H$ of rank $\hs$ are isomorphic (in symbols, $F \cong H)$ if there is a bijection $a : V(F) \to V(H)$ with the following property: for any $x \leq \hs$, $\{v_1, \dots, v_x\} \in E(F)$ if and only if $\{a(v_1), \dots, a(v_x)\} \in E(H)$. If $a$ satisfies the above then we call $a$ an \emph{isomorphism} (from $F$ to $H$). 

Similarly, two vertex-coloured hypergraphs $(F, \nu_F), (H, \nu_H)$ are isomorphic if there is an isomorphism $a$ from $F$ to $H$ that preserves the colours of the vertices, that is, for each $v \in V(F)$, $\nu_F(v) = \nu_H(a(v))$ is satisfied.

\subsubsection{Sub-hypergraphs}
In the present work, whenever we refer to a \emph{sub-hypergraph} $H$ of a hypergraph $G$ we assume that $H$ is \emph{induced} by some set $A$ of hyperedges of $G$, that is, $V(H) = \bigcup_{e \in A}e$ and $E(H) = A$. For $A \subseteq E(G)$, we write $G[A]$ for the hypergraph that is induced by $A$.

Given two hypergraphs $F, G$, we write $\subs{F}{G}$ for the set of all sub-hypergraphs $H$ of $G$ that are isomorphic to $F$.

If $G$ comes with a vertex-colouring $\nu_G$, then every sub-hypergraph $H$ of $G$ admits a natural vertex-colouring $\nu_H = \nu_G|_{V(H)}$, which is the restriction of $\nu_G$ to the vertices of $H$. We write $\subs{(F, \nu_F)}{(G, \nu_G)}$ for the set of all (vertex-coloured) sub-hypergraphs $(H, \nu_H)$ of $(G, \nu_G)$ that are isomorphic to $(F, \nu_F)$.

\subsubsection{Special Classes of Uniform Hypergraphs}
A hypergraph, all hyperedges of which have the same size $\hs$, is called $\hs$-\textit{uniform}. For $k, \hs \in \mathbb{N}$, we write $\mathcal{G}_k(\hs)$ for the class of all (isomorphism types) of $\hs$-uniform hypergraphs with $k$ hyperedges, without isolated vertices. We also set $\mathcal{G}_{\leq k}(\hs) = \bigcup_{\ell = 1}^{k}\mathcal{G}_\ell(\hs)$. 

Similarly, for the case of $C$-coloured hypergraphs, we write $\mathcal{G}_k(C;\hs)$ for the class of all (isomorphism types) of $C$-coloured $\hs$-uniform hypergraphs with $k$ hyperedges, without isolated vertices. We also set $\mathcal{G}_{\leq k}(C;\hs) = \bigcup_{\ell = 1}^{k}\mathcal{G}_{\ell}(C; \hs)$.  

\subsubsection{Gaifman Graphs and the Treewidth of Hypergraphs}
Let $G$ be a \emph{graph}. A \emph{tree decomposition} of $G$ is a pair $(\mathcal{T}, \beta)$ of a tree $\mathcal{T}$ and a mapping $\beta : V(\mathcal{T}) \to 2^{V(G)}$ with the following properties:
\begin{itemize}
\item For each $\{u, v\} \in E(G)$, there is $t \in V(\mathcal{T})$, such that $\{u, v\} \subseteq \beta(t)$.
\item For each $v \in V(G)$, the subtree of $\mathcal{T}$ consisted of all vertices $t \in V(\mathcal{T})$ such that $v \in \beta(t)$, is connected.
\end{itemize}
The \emph{width} of a tree decomposition $(\mathcal{T}, \beta)$ is given by $\max_{t \in V(T)}|\beta(t)|-1$. The \emph{treewidth} of $G$ is the minimum width among all possible tree decompositions of $G$.

The treewidth of a hypergraph $H$ is defined as the treewidth of the \emph{Gaifman graph} of $H$, which is a graph on the vertices of $H$ where two vertices $u, v$ are made adjacent if there is $e \in E(G)$ such that $\{u, v\} \subseteq e$ is satisfied.

\subsection{Homomorphisms, Embeddings and Automorphisms}
Given two hypergraphs $F, G$, a \emph{homomorphism} from $F$ to $G$ is a mapping $h : V(F) \to V(G)$ with the following property: for each hyperedge $e \in E(F)$ we have that $\bigcup_{v \in e}h(v) \in E(G)$ is satisfied. A homomorphism from $F$ to $G$ that is injective is called an \emph{embedding} and an embedding from a hypergraph to itself is called an \emph{automorphism}. We write $\homs{F}{G}$ for the set of all homomorphisms from $F$ to $G$ and $\embs{F}{G}$ for the set of all embeddings from $F$ to $G$. We also write $\auts(F)$ for the set of all automorphisms of $F$.

If case $F, G$ come with vertex-colourings $\nu_F, \nu_G$ respectively, then a homomorphism $h$ from $(F, \nu_F)$ to $(G, \nu_G)$ is a homomorphism from $F$ to $G$ that preserves the colours of the vertices, that is, for each $v \in V(F)$, $\nu_F(v) = \nu_G(h(v))$. 

The sets $\homs{(F, \nu_F)}{(G, \nu_G)}, \embs{(F, \nu_F)}{(G, \nu_G)}, \auts{(F, \nu_F)}$ are defined accordingly.

\subsection{Partitions of Sets and the Möbius function}
A \emph{partition} of a finite set $S$ is a collection of pair-wise disjoint subsets of $S$, the union of which is equal to $S$. The elements of a partition are called \emph{blocks}, and given a partition $\rho$ we write $|\rho|$ for the number of blocks of $\rho$. We write $\mathsf{Part}(S)$ for the set of all partitions of $S$. For two partitions $\sigma, \rho$ of $S$, we say that $\sigma$ \emph{refines} $\rho$, and write $\sigma \leq \rho$ if we can obtain $\sigma$ by further partitioning (some of) the blocks of $\rho$. We write $\bot_S$ for the \emph{finest} partition of $S$, that is, every element of $S$ consists a block of $\rho$. We also write $\top_S = \{S\}$ for the \emph{coarsest} partition. Whenever the set $S$ is clear from context, we write $\bot$ and $\top$ for $\bot_S$ and $\top_S$ respectively.

The \emph{Möbius function} $\mu_S : \mathsf{Part}(S) \times \mathsf{Part}(S) \to \mathbb{Z}$ of the poset $(\mathsf{Part}(S), \leq)$ is given explicitly below.\footnote{We refer the reader to \cite[Chapter 3.7]{Stanley11} for a detailed introduction to the Möbius function of posets, however we need not any other properties of Möbius functions.}
\begin{definition}[The M\"obius function of partitions \cite{aivasiliotis2024parameterisedholantproblems}]\label{def:mobius}
    Let $S$ be a finite set and let $\rho=\{B_1,\dots,B_t\}$ be a partition of $S$. Let furthermore $\sigma$ be a partition of $S$ with $\sigma \leq \rho$. For each $i\in [t]$, let $\sigma^i$ be the subpartition of $\sigma$ refining $\{B_i\}$, that is, $\sigma = \dot{\cup}_{i=1}^t \sigma^i$ and $\sigma^i$ is a partition of $B_i$ for all $i\in [t]$. Then the \emph{M\"obius function} of $\sigma$ and $\rho$ is defined as follows:
    \[ \mu_S(\sigma,\rho):=  \prod_{i=1}^t (-1)^{|\sigma^i|-1}(|\sigma^i|-1)! \,,\]
    where $0!=1$ as usual. We will drop the subscript $S$ of $\mu_S$ if it is clear from the context and we will further simplify notation by writing $\mu(\rho)$ instead of $\mu(\bot, \rho)$.
\end{definition}

\subsection{Quotient Hypergraphs}

Given a hypergraph $G$, we write $\mathsf{Part}(G)$ for the set of all partitions of $V(G)$. If $G$ comes with a vertex-colouring $\nu$, we write $\mathsf{Part}(G, \nu)$ for the set of all \emph{colour-consistent} partitions of $V(G)$, where a partition is colour-consistent if no two vertices of different colours are contained in the same block of the partition.

Given a hypergraph $G$ and a partition $\rho \in \mathsf{Part}(G)$, we write $G/\rho$ for the \emph{quotient hypergraph} defined as follows:

\begin{itemize}
\item The vertices of $G/\rho$ are given by the blocks of $\rho$ and

\item for any $r \leq \mathsf{rank}(G)$, blocks $B_1, \dots, B_r$ are made adjacent if and only if there is $e \in E(G)$ such that $\bigcup_{i = 1}^{r}B_i = e$ and for each $1 \leq i \leq r$, $e \cap B_i \neq \emptyset$. 
\end{itemize}
If $G$ comes with a vertex-colouring, then $G/\rho$ is defined similarly with respect to any \emph{colour-consistent} partition $\rho$. Naturally, the vertices (blocks) of the quotient hypergraph are coloured by the colour of the vertices contained in the respective blocks.

\subsection{Basic Transformations of Embedding- and Subgraph-Counts}
\begin{lemma}\label{claim:subposetMobius}
Let $(P, \leq)$ be a poset and let $(Q, \leq)$ be a subposet of $(P, \leq)$ such that for any $\sigma \notin Q$ and $\rho \geq \sigma$, it holds $\rho \notin Q$.
Let $f : Q \to \mathbb{R}, g : Q \to \mathbb{R}$ such that $f(\sigma) = \sum_{\rho \in Q,\, \rho \geq \sigma}g(\rho)$. Then, $g(\sigma) = \sum_{\rho \in Q,\, \rho \geq \sigma}\mu(\sigma, \rho)\cdot f(\rho)$, where $\mu : P \times P \to \mathbb{R}$ is the Möbius function of the poset $(P, \leq)$.
\end{lemma}
\begin{proof}
Consider the extensions of $f$ and $g$ to $P$, denoted as $\hat{f}$ and $\hat{g}$ respectively, such that for all $\sigma \in P \setminus Q$, $\hat{f}(\sigma) = \hat{g}(\sigma) = 0$. For $\sigma \in P\setminus Q$, we have \[\hat{f}(\sigma) = \sum_{\rho \in P, \, \rho \geq \sigma}\hat{g}(\rho) = 0\,,\]
since, by hypothesis, for any $\rho \geq \sigma$ we have $\rho \in P\backslash Q$ and hence $\hat{g}(\rho) = 0$.
In case $\sigma \in Q$, we have \[\hat{f}(\sigma) = f(\sigma) = \sum_{\rho \in Q,\, \rho \geq \sigma}g(\rho) = \sum_{\rho \in P,\, \rho \geq \sigma}\hat{g}(\rho)\,,\]
which follows again from the fact that for $\rho \in P\setminus Q$, it is $\hat{g}(\rho) = 0$.
Next, we invoke Möbius inversion over $(P, \leq)$ for $\hat{f}$ and $\hat{g}$ and deduce that for $\sigma \in P$ we have \[\hat{g}(\sigma) = \sum_{\rho \in P,\, \rho \geq \sigma}\mu(\sigma, \rho)\cdot \hat{f}(\rho)\,\]
and so for any $\sigma \in Q$ we have
\[g(\sigma) = \hat{g}(\sigma) = \sum_{\rho \in Q,\, \rho \geq \sigma}\mu(\sigma, \rho)\cdot f(\rho)\,.\]
\end{proof}

The following transformations are well-known for the setting of (uncoloured) \textit{graphs} and lifting them to vertex-coloured hypergraphs follows almost the same line of argumentation (see e.g., \cite[Chapter 5.2.3]{Lovasz12} for \Cref{lem:uncolHomsEmbs} in the setting of uncoloured graphs). For reasons of self-containment and to address a few subtleties that arise in the case of hypergraphs, we provide the complete proofs.

\begin{lemma}\label{lem:uncolHomsEmbs}
Let $(H, \nu_H), (G, \nu_G)$ be $\mathcal{S}$-coloured hypergraphs. For any $\sigma \in \mathsf{Part}(H, \nu_H)$ it holds,
\begin{equation}\label{eq:embeddings}
\Emb((H/\sigma, \nu_{H/\sigma}) \to (G, \nu_G)) = \sum_{\substack{\rho \in \mathsf{Part}(H, \nu_H) \\ \rho \geq \sigma}}\mu(\sigma, \rho)\cdot\Hom((H/\rho, \nu_{H /\rho}) \to (G, \nu_G))\,,   
\end{equation}
where $\mu : \mathsf{Part}(H) \times \mathsf{Part}(H) \to \mathbb{R}$ is the Möbius function of the partition lattice$(\mathsf{Part}(H), \leq)$.
\end{lemma}
\begin{proof}
Given $\phi \in \mathsf{Hom}((H/\sigma, \nu_{H/\sigma}) \to (G, \nu_G))$, we let $\rho_{\phi}$ denote the partition induced by $\phi$, that is, the partition having blocks given by $\phi^{-1}(u)$ for each $u \in V(G)$ such that $\phi^{-1}(u) \neq \emptyset$. By slightly abusing notation we can see $\rho_\phi$ both as a partition of $V(H/\sigma)$ and a partition of $V(H)$ that is a coarsening of $\sigma$. For the sake of simplicity, we write $H/\rho_\phi$ instead of $(H /\sigma) /\rho_\phi$.

Clearly, $\rho_{\phi}$ is color-consistent. Furthermore, the map $\psi : V(H/\rho_\phi) \to V(G)$ that maps each block $\phi^{-1}(u)$ to $u$ is an embedding. To see this, first observe that it is indeed injective. Then, let $\{\phi^{-1}(u_1), \ldots, \phi^{-1}(u_r)\} \in E(H/\rho_\phi)$. By definition, there exist non-empty subsets $V_1 \subseteq \phi^{-1}(u_1), \ldots, V_\hs \subseteq \phi^{-1}(u_r)$ such that $\bigcup_{i=1}^{r}V_i \in E(H/\sigma)$. Also, $\psi(\{\phi^{-1}(u_1), \dots, \phi^{-1}(u_r)\}) = \{u_1, \dots, u_r\}$ and $ \phi(\bigcup_{i=1}^{r}V_i) = \{u_1, \dots, u_r\}$. Since, $\bigcup_{i=1}^{r}V_i \in E(H/\sigma)$, it follows that $\{u_1, \dots, u_r\} \in E(G)$, and so $\psi \in \embs{(H/\rho_{\phi}, \nu_{H/\rho_{\phi}})}{(G, \nu_G)}$. 

On the other hand, the set of homomorphisms that induce the same partition $\rho \in \mathsf{Part}(H/\sigma, \nu_{H/\sigma})$ define an equivalence class with $\Emb((H/\rho, \nu_{H/\rho}) \to (G, \nu_G))$ many elements. For this, let $\rho \in \mathsf{Part}(H/\sigma, \nu_{H/\sigma})$. For $\psi \in \embs{(H/\rho, \nu_{H/\rho})}{(G, \nu_G)}$, let $\phi : V(H/\sigma) \to V(G)$ be defined such that $\phi(v) = \psi(B)$, where $B$ is the block of $\rho$ that contains $v$. By definition, $\phi$ is color-preserving. Given a hyperedge $e = \{v_1, \ldots, v_r\} \in E(H/\sigma)$, let $B_1, \ldots, B_{\ell}$ be the blocks of $\rho$ intersecting with $e$. By definition $\{B_1, \ldots, B_{\ell}\} \in E(H/\rho)$. We have $\phi(e) = \psi(\{B_1, \ldots, B_{\ell}\}) \in E(G)$. Hence, $\phi$ is a homomorphism in $\homs{(H/\sigma, \nu_{H/\sigma})}{(G, \nu_G)}$ that induces $\rho$. 

Let $\Hom((H/\sigma, \nu_{H/\sigma}) \to (G, \nu_G))[\rho]$ denote the number of homomorphisms that induce the same partition $\rho \in \mathsf{Part}(H/\sigma, \nu_{H/\sigma})$. Recall that we can equivalently treat $\rho \in \mathsf{Part}(H/\sigma, \nu_{H/\sigma})$ as a partition in $\mathsf{Part}(H, \nu_H)$ such that $\rho \geq \sigma$. We have

\begin{align*}
\Hom(((H/\sigma, \nu_{H/\sigma}) \to (G, \nu_G)) &= \sum_{\substack{\rho \in \mathsf{Part}(H, \nu_H) \\ \rho \geq \sigma}}\Hom(((H/\sigma, \nu_{H/\sigma}) \to (G, \nu_G))[\rho] \\ &= \sum_{\substack{\rho \in \mathsf{Part}(H, \nu_H) \\ \rho \geq \sigma}}\Emb((H/\rho, \nu_{H/\rho}) \to (G, \nu_G)).
\end{align*}
Finally, we observe that the conditions of \Cref{claim:subposetMobius} are satisfied by taking $P = \mathsf{Part}(H)$, $Q = \mathsf{Part}(H, \nu_H)$, $f : \sigma \mapsto \Hom(((H/\sigma, \nu_{H/\sigma}) \to (G, \nu_G))$ and $g : \sigma \mapsto \Emb(((H/\sigma, \nu_{H/\sigma}) \to (G, \nu_G))$. Indeed, no coarsening of a non-color-consistent partition is color-consistent. So, we get 
\begin{equation*}
\Emb((H/\sigma, \nu_{H/\sigma}) \to (G, \nu_G)) = \sum_{\substack{\rho \in \mathsf{Part}(H, \nu_H) \\ \rho \geq \sigma}}\mu(\sigma, \rho)\cdot\Hom((H/\rho, \nu_{H /\rho}) \to (G, \nu_G))\,
\end{equation*}
where $\mu$ is the Möbius function of the partition lattice $(\mathsf{Part}(H), \leq)$.
\end{proof}

\begin{proposition}\label{prop:subgraphsAndEmbs} Let $(H, \nu_H), (G, \nu_G)$ be hypergraphs. It holds,
\begin{equation}\label{eq:subgraphs}
\Sub((H, \nu_H) \to (G, \nu_G)) = \Emb((H, \nu_H) \to (G, \nu_G)) \cdot \Aut(H, \nu_H)^{-1}.
\end{equation}
\end{proposition}
\begin{proof}
Let $a \in \auts{(H, \nu_H)}$ and $\phi \in \embs{(H, \nu_H)}{(G, \nu_G)}$. We define the composition of mappings $\circ : \auts{(H, \nu_H)} \times \embs{(H, \nu_H)}{(G, \nu_G)} \mapsto \embs{(H, \nu_H)}{(G, \nu_G)}$ as the group action of the group $\auts{(H, \nu_H)}$ acting on $\embs{(H, \nu_H)}{(G, \nu_G)}$ that maps $(a, \phi)$ to $a\,\circ\,\phi$, where for each $v \in V(H), (a\,\circ\,\phi)(v) = \phi(a(v))$. The action $\circ$ is a free group action, so the orbit of each $\phi \in \embs{(H, \nu_H)}{(G, \nu_G)}$ is of size $\#\auts{(H, \nu_H)}$. Finally, since the image of $\phi$ is a sub-hypergraph of $G$ (that is, it is hyperedge-induced) and the group action preserves the image of $\phi$, it follows that the orbits of $\embs{(H, \nu_H)}{(G, \nu_G)}$ correspond uniquely to sub-hypergraphs of $(G, \nu_G)$ that are isomorphic to $(H, \nu_H)$, yielding \Cref{eq:subgraphs}.
\end{proof}

\subsection{Parameterised (Counting) Complexity Theory}

Recall that a counting problem is a function $P : \{0, 1\}^* \to \mathbb{Q}$. A \emph{parameterised counting problem} is a pair $(P, \kappa)$ where $P$ is a counting problem and $\kappa : \{0, 1\}^* \to \mathbb{N}$ is a computable function, called the \emph{parametrisation} of $P$. In particular, $\kappa$ is evaluated on the input $x \in \{0, 1\}^*$ of $P$. Parameterised (counting) complexity  theory\footnote{We refer the readers to \cite{CyganFKLMPPS15} for a detailed exposition on parameterised algorithms and complexity theory and in particular to \cite[Chapter 14]{FlumG06} and \cite{Curticapean15} for the counting counterpart.} relaxes the notion of efficiency as it is perceived in the classical (counting) complexity theory, where a counting problem $P$ is solvable efficiently if there is an algorithm that computes $P(x)$ in time $|x|^{O(1)}$, based on the following observation: in the parameterised world, it is often the case that the parameter $\kappa(x)$ is much smaller compared to the size of the input $x$. Hence, algorithms computing $P(x)$ in time $f(\kappa(x))\cdot |x|^{O(1)}$ are sought, for any computable function $f$ (that depends only on the parametrisation $\kappa$). The problems that are solvable in the aforementioned running time comprise the class of \emph{fixed-parameter tractable} problems (FPT).

For containment of a problem $(P, \kappa)$ in the class FPT, it suffices to construct an algorithm that computes $P(x)$ in time $f(\kappa(x))\cdot |x|^{O(1)}$ (for some computable function $f$), which can be done either explicitly or even implicitly via reducing $(P, \kappa)$ to a problem that is already known to be in the class FPT, via \emph{FPT Turing-reductions} defined below.

\begin{definition}[Parameterised Reductions]
Let $(P, \kappa), (P', \kappa')$ be two parameterised counting problems. An \emph{FPT Turing-reduction} from $(P, \kappa)$ to $(P', \kappa')$ is an algorithm $\mathbb{A}$ to compute $P(x)$ that has oracle access to $(P', \kappa')$ satisfying the following:

\begin{itemize}
\item There is a computable function $f$ such that $\mathbb{A}$ computes $P(x)$ in time $f(\kappa(x))\cdot |x|^{O(1)}$.
\item There is a computable function $g$ such that for any queried instance $y$, we have $\kappa'(y) \leq g(\kappa(x))$.
\end{itemize}

If in particular, $\mathbb{A}$ computes $P(x)$ in time $f(\kappa(x))\cdot O(|x|)$ making at most $f(\kappa(x))$ queries to the oracle, then we call $\mathbb{A}$ a \emph{linear FPT Turing-reduction}.

We write $(P, \kappa) \leq^{\mathsf{FPT}}_{\mathsf{T}} (P', \kappa')$ $($resp. $(P, \kappa) \fptlinred (P', \kappa'))$ if there is an FPT Turing-reduction $($resp. linear FPT Turing-reduction$)$ from $(P, \kappa)$ to $(P', \kappa')$. 
\end{definition}

It is in fact a folklore result in parameterised complexity theory that an FPT Turing-reduction implies that whenever the right-hand side problem of the reduction is in FPT, then so is the left-hand side. 

In particular, if the FPT Turing-reduction is linear then, we also have better control over the runtime complexity of $(P, \kappa)$, as shown in the following less-straightforward result (see \cite[Lemma 2.8]{aivasiliotis2024parameterisedholantproblems} for a proof), which is very useful if one is after a more fine-grained analysis.

\begin{lemma}[\cite{aivasiliotis2024parameterisedholantproblems}]
Let $d\geq 0$ and $c\geq 1$ be reals, and let $(P,\kappa)$ and $(P',\kappa')$ be parameterised counting problems such that $(P,\kappa)\fptlinred (P',\kappa')$. Assume that $(P',\kappa')$ can be solved in time $f(\kappa'(x))\cdot O(\log^d(|x|)\cdot |x|^c)$ for some computable function $f$. Then there is a computable function $g$ such that $(P,\kappa)$ can be solved in time $g(\kappa(x))\cdot O(\log^d(|x|)\cdot |x|^c)$.    
\end{lemma}

On the other hand, to argue about non-containment of a problem $(P, \kappa)$ in the class FPT, we show that $(P, \kappa)$ is hard for the class $\#\mathrm{W}$[1] (under FPT Turing-reductions), which can be thought of as the parameterised counting equivalent of the class $\mathrm{NP}$ (see \cite[Chapter 14]{FlumG06}). It is known that under the Exponential Time Hypothesis (ETH) \cite{ImpagliazzoP01,ImpagliazzoPZ01}, which asserts that 3-SAT cannot be solved in time $\mathsf{exp}(o(n))$ (where $n$ is the number of variables of the input formula), $\#\mathrm{W}[1]$-hard problems are not in FPT \cite{Chenetal05,Chenetal06,CyganFKLMPPS15}. 

The canonical $\#\mathrm{W}$[1]-complete problem is $\#\textsc{p-Clique}$ which, given a graph $G$ and an integer $k$, computes the number of cliques of size $k$ in $G$. Hence, we say that a problem $(P, \kappa)$ is $\#\mathrm{W}[1]$-hard if $\#\textsc{p-Clique} \leq^{\mathsf{FPT}}_{\mathsf{T}} (P, \kappa)$. If in addition, $(P, \kappa) \leq^{\mathsf{FPT}}_{\mathsf{T}} \#\textsc{p-Clique}$, then $(P, \kappa)$ is $\#\mathrm{W}[1]$-complete.

\subsection{Holants as Linear Combinations of Homomorphism-Counts}\label{sec:HolantsInHomBasis}

In this section, we show how the holant value of an instance $(\Omega, k) \in \uni{\constuncolholantprob{\mathsf{r}}}(\mathcal{S})$ can be expressed as a finite combinatorial sum of homomorphism-counts of the form $\#\homs{*}{(G, \nu_G)}$, where $G$ is the underlying hypergraph of $\Omega$ and $\nu_G$ is the $\mathcal{S}$-colouring of $G$ implied by $\Omega$.

For the sake of simplicity, we normalise our signatures based on the following remark.

\begin{remark}[On $s(0) = 1$]\label{remark:zeroValues}
Following \cite[Remark 6.6]{aivasiliotis2024parameterisedholantproblems}, we may assume without loss of generality (both with respect to parameterised and classical complexity), that we only consider signatures $s$ such that, whenever $s(0) \neq 0$, we have $s(0) = 1$.   
\end{remark}

\begin{lemma}\label{lem:uncolHomBasis}
Let $\mathcal{S}$ be a finite set of signatures and let $\Omega = (G, \{s_v\}_{v \in V(G)})$, $k \in \mathbb{N}$ be an instance of $\uni{\constuncolholantprob{\mathsf{r}}}(\mathcal{S})$, where $G$ is $\hs$-uniform. We have
\begin{equation}\label{eq:uncolHolantHomBasis}
\uncolholant(\Omega, k) = \sum_{\colGraph{F} \in \kgraphsGamma}\zeta_{k, \mathcal{S}}{\colGraph{F}}\cdot \Hom(\colGraph{F} \to \colGraph{G})\,,
\end{equation}
where
\begin{align}\label{eq:coeffColored}
\nonumber&\zeta_{k, \mathcal{S}}{\colGraph{F}}  &= \sum_{\colGraph{H} \in \kgraphsBHS}\frac1{\Aut\colGraph{H}}\prod_{v \in V(H)}\nu_H(v)(\deg_H(v))\sum_{\substack{\rho \in \mathsf{Part}\colGraph{H} \\ \colGraph{F} \cong \colGraph{H/\rho}}}\mu(\rho)\,.
\end{align}   
\end{lemma}

\begin{proof}
Recall, that
\[\uncolholant(\Omega) = \sum_{\substack{A \in \binom{E(G)}{k}}}\prod_{v\in V(G)}s_v(|A\cap E(v)|)\,.\]

Given a set $A \in \binom{E(G)}{k}$, we write $G[A]$ for the sub-hypergraph of $G$ induced by $A$. The $\mathcal{S}$-colouring of $G[A]$ is given by $\nu_{G[A]} = \nu_G|_{V(G[A])}$.

For $\colGraph{H} \in \kgraphsBHS$, let $[\colGraph{H}]$ denote the set of all $A \in \binom{E(G)}{k}$ such that $\colGraph{G[A]} \cong \colGraph{H}$; note, that $[\colGraph{H}]$ may be empty. 

Recall, that for each $s \in \mathcal{S}$, $s(0) = 1$. It is easy to verify that, given $\colGraph{H} \in \kgraphsBHS$, each $A \in [\colGraph{H}]$ satisfies
\[\prod_{v \in V(G)}s_v(|A \cap E(v)|) = \prod_{v \in V(H)}\nu_H(v)(\deg_H(v))\,.\]
Hence, we may write
\begin{equation}\label{eq:HolantToSubgraphs}
\uncolholant(\Omega) = \sum_{\colGraph{H} \in \kgraphsBHS}\big|[\colGraph{H}]\big|\prod_{v \in V(H)}\nu_H(v)(\deg_H(v))\,.
\end{equation}

We observe that $\big|[\colGraph{H}]\big|$ is equal to the number of sub-hypergraphs of $G$ that are isomorphic to $\colGraph{H}$. Formally, $\big|[\colGraph{H}]\big| = \#\subs{\colGraph{H}}{\colGraph{G}}$. Hence, by \Cref{prop:subgraphsAndEmbs} and \Cref{lem:uncolHomsEmbs}, it follows that
\begin{align}
\nonumber&\big|[\colGraph{H}]\big| \\ \nonumber&= \auts{\colGraph{H}}^{-1}\cdot\#\embs{\colGraph{H}}{\colGraph{G}} \\ &= \label{eq:intermediateHomBasis}\auts{\colGraph{H}}^{-1}\sum_{\rho \in \mathsf{Part}\colGraph{H}}\mu(\rho)\cdot\#\homs{\colGraph{H/\rho}}{\colGraph{G}}\,.
\end{align}

Next, we group terms of the above formula such that all partitions within the same group induce the same quotient hypergraph (up to isomorphisms). Note, that the quotient hypergraph $\colGraph{F}$ induced by partitions of some group may not be uniform (that is, it contains at least one hyperedge of size strictly less than $\hs$). However, it is easy to see that, for such hypergraphs, we have $\#\homs{(F, \nu_F)}{(G, \nu_G)} = 0$ since $G$ is assumed to be $\hs$-uniform. Hence, we may assume that $(F, \nu_F) \in \kgraphsGamma$ and write

\begin{align*}
&\big|[\colGraph{H}]\big| \\ &= \frac1{\#\auts{\colGraph{H}}}\sum_{\colGraph{F} \in \kgraphsGamma}\left(\sum_{\substack{\rho \in \mathsf{Part}\colGraph{H} \\ \colGraph{H/\rho}\cong\colGraph{F}}}\mu(\rho)\right)\#\homs{\colGraph{F}}{\colGraph{G}}\,.
\end{align*}

By replacing each term $\big|[\colGraph{H}]\big|$ in \eqref{eq:HolantToSubgraphs} with the expression above and rearranging the summands, we get the desired result.
\end{proof}

Note that, any hypergraph $G$ can be seen as a vertex-coloured hypergraph the vertices of which have all the the same colour (and vice versa). Hence, clearly, homomorphisms between hypergraphs that are vertex-coloured with the same single colour are essentially identical to those between the respective underlying uncoloured hypergraphs. The same is true for the colour-preserving partitions of such coloured hypergraphs. That said, we may simplify the notation in \Cref{lem:uncolHomBasis} for the case where $\mathcal{S} = \{s\}$ as follows. We have

\begin{lemma}[\Cref{lem:uncolHomBasis} restated for signature sets with a single signature]\label{lem:uncolHomBasisRESTATED}
Let $s$ be a signature and let $\Omega = (G, \{s\}_{v \in V(G)})$, $k \in \mathbb{N}$ be an instance of $\uni{\uncolholantprob}(\{s\})$, where $G$ is $\hs$-uniform. We have
\begin{equation}\label{eq:uncolHolantHomBasisSimplified}
\uncolholant(\Omega, k) = \sum_{F \in \mathcal{G}_{\leq k}(\hs)}\zeta_{k, s}(F)\cdot \Hom(F\to G)\,,
\end{equation}
where
\begin{align}\label{eq:coeffColoredRestated}
\nonumber\zeta_{k, s}(F)  = \sum_{H \in \mathcal{G}_k(\hs)}\frac1{\Aut(H)}\prod_{v \in V(H)}s(\deg_H(v))\sum_{\substack{\rho \in \mathsf{Part}(H) \\ F \cong H/\rho}}\mu(\rho)\,.
\end{align}   
\end{lemma}

\section{A Complete Parameterised Classification of \ensuremath{\uncolholantprob}}\label{sec:para}
In this section, we present a full classification of the problem $\uncolholantprob$. We approach the classification in two ways, depending on whether we restrict our signatures to satisfy $s(0) \neq 0$ or not, as follows. 

Firstly, we consider only signatures $s$ that satisfy $s(0) \neq 0$. Recall that any finite set $\mathcal{S}$ of such signatures can be of type $\mathbb{T}[1], \mathbb{T}[2]$ or $\mathbb{T}[\infty]$. 
For every finite set $\mathcal{S}$ of type $\mathbb{T}[1]$, we present an FPT-(near)-linear time algorithm for $\uncolholantprob(\mathcal{S})$ (in \Cref{lem:linearSigAlgo}). We complete the classification by showing that if $\mathcal{S}$ is not of type $\mathbb{T}[1]$, then $\uncolholantprob(\mathcal{S})$ is $\#\mathrm{W}[1]$-hard. In fact, we establish an even stronger result: whenever $\mathcal{S}$ is not of type $\mathbb{T}[1]$, then for any $d \in \mathbb{Z}_{\geq 2}$, the problem $\uni{\constuncolholantprob{d}}(\mathcal{S})$ is $\#\mathrm{W}[1]$-hard (see \Cref{thm:reduction2TodArity}).\footnote{We stress that the case $d = 2$ was shown in \cite{aivasiliotis_et_al:LIPIcs.ICALP.2025.7}.} We further complement our hardness results by providing lower bounds (under ETH) on the runtime complexity of $\uni{\constuncolholantprob{d}}(\mathcal{S})$, for signature sets $\mathcal{S}$ of type $\mathbb{T}[\infty]$. Clearly, the aforementioned hardness results transfer directly to $\uncolholantprob(\mathcal{S})$.

Next, we allow for signatures $s$ to satisfy $s(0) = 0$ and extend our tractability results above by showing that for any finite set of signatures $\mathcal{S}$ of type $\mathbb{T}[1]$, $\uncolholantprob(\mathcal{S})$ remains tractable even if we add to $\mathcal{S}$ any finite number of signatures $s$ with $s(0) = 0$. The aforementioned algorithm is shown in \Cref{lem:linearSigAlgoWithZeros}. On the other hand, it is easy to see that if $\mathcal{S}$ contains a subset that is of type $\mathbb{T}[2]$ or $\mathbb{T}[\infty]$, then all of the previously mentioned hardness results transfer directly to $\uncolholantprob(\mathcal{S})$, with which we obtain a complete classification of $\uncolholantprob$ for any finite set $\mathcal{S}$ of signatures with no restrictions, stated formally as follows.

\begin{theorem}\label{thm:mainParamThm}
Let $\mathcal{S}$ be a finite set of signatures. Let $\mathcal{S}_0 = \{s \in \mathcal{S} \mid s(0) = 0\}$. We have the following classification of the complexity of $\uncolholantprob(\mathcal{S})$.
\begin{itemize}
\item[1.] If $\mathcal{S} \setminus \mathcal{S}_0$ is of type $\mathbb{T}[1]$, then $\uncolholantprob(\mathcal{S})$ can be solved in FPT-(near)-linear time, that is, in time $f(k, \hs) \cdot \tilde{O}(|\Omega|)$, where $\hs$ is the rank of the underlying hypergraph of $\Omega$.
\item[2.] If $\mathcal{S}\setminus \mathcal{S}_0$ is not of type $\mathbb{T}[1]$, then $\uncolholantprob(\mathcal{S})$ is $\#\mathrm{W}[1]$-hard. In particular, for any $d \in \mathbb{Z}_{\geq 2}$, the problem $\uni{\constuncolholantprob{d}}(\mathcal{S})$ is $\#\mathrm{W}[1]$-hard. Furthermore, if $\mathcal{S} \setminus \mathcal{S}_0$ is of type $\mathbb{T}[\infty]$, then $\uni{\constuncolholantprob{d}}(\mathcal{S})$ cannot be solved in time $f(k)\cdot |V(\Omega)|^{o(k/\log(k))}$, for any computable function $f$ and any $d \geq 2$, unless ETH fails.
\end{itemize}
\end{theorem}

The tractability part of the theorem above is shown in \Cref{sec:TractabilitySubsection}, while the hardness part is shown in \Cref{sec:hardnessSubsection}.

\subsection{Tractability results for \ensuremath{\uncolholantprob}}\label{sec:TractabilitySubsection}

Recall that, unless stated otherwise, when we refer to a symmetric signature $s$, we assume that $s(0) \neq 0$. In \Cref{lem:linearSigAlgo}, we present an FPT-(near)-linear time algorithm for $\uncolholantprob(\mathcal{S})$, for finite sets $\mathcal{S}$ of symmetric signatures, of type $\mathbb{T}[1]$. In \Cref{lem:linearSigAlgoWithZeros}, we extend our algorithm so as to address the case in which $\mathcal{S}$ contains signatures $s$ with $s(0) = 0$. Note that the modified algorithm of \Cref{lem:linearSigAlgoWithZeros} still runs in FPT-(near)-linear time. 

Before we present our algorithms, we state and prove the following simple lemma that will allow us to restrict our analysis to uniform hypergraphs, thus simplifying the analysis.

\begin{lemma}\label{claim:uniformRestriction} Let $\mathcal{S}$ be a finite set of signatures and let $\text{\sf{one}} : \mathbb{N} \to \{1\}$. It holds,
\[\uncolholantprob(\mathcal{S}) \fptlinred \uni{\uncolholantprob}(\mathcal{\mathcal{S} \,\cup\,\{\text{\sf{one}}\}})\,.\]
In particular, the parameterised Turing reduction above is in fact a polynomial-time Turing reduction and so
\[
\unbndclassHolant(\mathcal{S}) \leq_\mathsf{T} \uni{\unbndclassHolant}(\mathcal{S}\,\cup\,\{\mathsf{one}\})\,.
\]
\end{lemma}
\begin{proof}
Let $\Omega = (G, \{s_v\}_{v \in V(G)}) \in \uncolholantprob(\mathcal{S})$ be a signature grid. Let $\hs$ denote the rank of $G$. Consider the signature grid $\Omega' = (G', \{s'_v\}_{v \in V(G')})$, where $G'$ also has rank $\hs$, obtained from $\Omega$ as follows. We obtain $G'$ from $G$ by adding $\hs - \min_{e \in E(H)}|e|$ new vertices and, then, for each hyperedge $e \in E(G)$, adding to $e$ any $\hs - |e|$ of the new vertices, yielding the hyperedge $e'$. We equip each new vertex with the constant signature $\text{\sf{one}} : \mathbb{N} \to \{1\}$, while the rest of the vertices keep their signatures. Observe that $G'$ is $\hs$-uniform with size $|G'| \leq \mathsf{poly}(|G|)$, and that $\Omega' \in \uncolholantprobuni(\mathcal{S}\,\cup\,\{\text{\sf{one}}\})$. Furthermore, it is easy to verify that
\[\uncolholant(\Omega', k) = \uncolholant(\Omega, k)\,,\]
since the contribution of the new vertices to any assignment in the the holant value is always equal to one.
\end{proof}

\begin{lemma}\label{lem:linearSigAlgo}
Let $\mathcal{S}$ be a finite set of symmetric signatures of type $\mathbb{T}[1]$. Let $(\Omega, k)$ be an instance of $\uncolholantprob(\mathcal{S})$ and let $\hs$ denote the rank of the underlying graph $G$ of $\Omega$. There is an algorithm to compute $\uncolholant(\Omega, k)$ in time $f(k, \hs) \cdot \tilde{O}(|G|)$, for some computable function $f$. In particular, for any constant $c$, $f(k, c) \in \mathsf{poly}(k)$. 
\end{lemma}

\begin{proof}
Recall that for each $s \in \mathcal{S}$, we have $s(0) = 1$. Furthermore, since $\mathcal{S}$ is of type $\mathbb{T}[1]$, it follows from \cite[Lemma 6.15]{aivasiliotis2024parameterisedholantproblems} that for each $s \in \mathcal{S}$ we have $s(n) = s(1)^{n}, n \geq 1$. We assume that we are only given instances of $\uncolholantprobuni(\mathcal{S})$, since from \Cref{claim:uniformRestriction} we know that $\uncolholantprob(\mathcal{T})$ reduces to $\uncolholantprobuni(\mathcal{T} \,\cup\, \{\mathsf{one}\})$ in FPT-(near)-linear time, and $\mathcal{T}, \mathcal{T\,\cup\,\{\mathsf{one}\}}$ have the same type.

We further assume that there are $z$ different signatures assigned to the vertices of $G$. Since every signature $s$ is completely determined by $s(1)$, we assume that the values of the $z$ distinct signatures on input 1 are $a_1, \ldots, a_z$. 

For $1 \leq i \leq z$, let $N_i = \{v \in V(G) \mid s_v(1) = a_i\}$, where $s_v \in \mathcal{S}$ denotes the signature assigned to $v$. Let $\mathcal{L}$ denote the set of all vectors $\lambda \in \{0, 1, \dots, \hs\}^z$ such that $\sum_{i=1}^{z}\lambda(i) = \hs$.  Given $\lambda \in \mathcal{L}$, let $E_{\lambda} = \{e \in E(G) \mid \forall 1\leq i\leq z, |e\,\cap\,N_i| = \lambda(i)\}$. Clearly, the sets $\{E_\lambda\}_{\lambda \in \mathcal{L}}$ partition $E(G)$, that is, $\bigcup_{\lambda \in \mathcal{L}}E_\lambda = E(G)$ and $E_\lambda \cap E_{\lambda'} = \emptyset$, whenever $\lambda \neq \lambda'$ (and $E_\lambda, E_{\lambda'}$ are not empty).

Let $\{k_\lambda\}_{\lambda \in \mathcal{L}} \in \mathbb{N}^{|\mathcal{L}|}$ be a partition of $k$, that is, $k = \sum_{\lambda \in \mathcal{L}}k_\lambda$ (where $k_\lambda$ may be zero). Let $A \subseteq E(G)$ be a set of $k$ hyperedges, such that, for each $\lambda \in \mathcal{L}$, $|A\,\cap\,E_\lambda| = k_\lambda$. We say that $A$ \textit{induces} the partition $\{k_\lambda\}_{\lambda \in \mathcal{L}}$ of $k$. The contribution of $A$ to $\uncolholant(\Omega)$ is computed as follows:

\begin{equation*}
   \prod_{v \in V(G)}s_v(|A \cap E_G(v)|) = \prod_{i = 1}^{z}\prod_{v \in N_i}a_i^{|A\,\cap\, E_G(v)|} = \prod_{i = 1}^{z}a_i^{\sum_{v \in N_i}|A\,\cap\,E_G(v)|} = \prod_{i=1}^{z}a_i^{\sum_{\lambda \in \mathcal{L}}\lambda(i)\cdot k_{\lambda}}\,,
\end{equation*}
where the last equation follows from the observation that for every hyperedge $e \in A\,\cap\,E_\lambda$, $\lambda \in \mathcal{L}$, there are exactly $\lambda(i)$ vertices $v \in N_i$ such that $e \in A\,\cap\,E_G(v)$ and the number of those hyperedges is $k_\lambda$.

Furthermore, all subsets of $k$ hyperedges that induce the same partition of $k$ into $\sum_{\lambda \in \mathcal{L}}k_{\lambda}$ yield the same contribution. The number of those subsets is precisely $\prod_{\lambda \in \mathcal{L}}\binom{|E_{\lambda}|}{k_\lambda}$. 

Let $[k]^+_{\mathcal{L}}$ denote all possibilities we can write $k$ as $\sum_{\lambda \in \mathcal{L}}k_{\lambda}$ (where $k_{\lambda}$ may be zero). We have,

\begin{equation*}
    \uncolholant(\Omega) = \sum_{\{k_\lambda\}_{\lambda \in \mathcal{L}} \in [k]^+_{\mathcal{L}}}\prod_{\lambda \in \mathcal{L}}\binom{|E_{\lambda}|}{k_\lambda}\prod_{i=1}^{z}a_i^{\sum_{\lambda \in \mathcal{L}}\lambda(i)\cdot k_{\lambda}}\,.
\end{equation*}
The size of $\mathcal{L}$ is $\mathsf{poly}(\hs)$ since $z$ is upper-bounded by $|\mathcal{S}|$ which is a constant. Hence, the sum above has $k^{\mathsf{poly}(\hs)}$ many summands, each of which can be computed in $\mathsf{poly}(k, \hs)\cdot\tilde{O}(|G|)$ time.
\end{proof}

We extend \Cref{lem:linearSigAlgo} so as to address the case where $\mathcal{S}$ also contains signatures $s$ with $s(0) = 0$.
\begin{lemma}\label{lem:linearSigAlgoWithZeros}
Let $\mathcal{S}$ be a finite set of symmetric signatures that may contain signatures $s$ with $s(0) = 0$. We set $\mathcal{S}_0 = \{s \in S \mid s(0) = 0\}$. We assume that $\mathcal{S} \setminus \mathcal{S}_0$ is of type $\mathbb{T}[1]$. Let $(\Omega, k)$ be an instance of $\uncolholantprob(\mathcal{S})$ and let $\hs$ denote the rank of the underlying graph $G$ of $\Omega$. There is an algorithm to compute $ \uncolholant(\Omega, k)$ in FPT-(near)-linear time, that is, in time $f(k, \hs) \cdot \tilde{O}(|G|)$, for some computable function $f$.
\end{lemma}

\begin{proof}
As in the proof of the previous lemma, due to \Cref{claim:uniformRestriction}, we may assume that $G$ is uniform. We may also assume that, without loss of generality, the number of vertices equipped with signatures $s \in \mathcal{S}_0$ is at most $k\cdot\hs$, as otherwise the value of the holant would be zero. To see this, note that, $k\cdot\hs$ is the maximum number of vertices that $k$ hyperedges (of maximum size $\hs)$ can cover. 

As in \Cref{lem:linearSigAlgo}, we assume that $\mathcal{S} \setminus \mathcal{S}_0$ contains $z$ distinct signatures, each associated with their $s(1)$ value, which we denote as $a_i$. For $1 \leq i \leq z$, let $N_i = \{v \in \mathcal{S} \setminus \mathcal{S}_0 \mid s_v(1) = a_i\}$, where $s_v \in \mathcal{S} \setminus \mathcal{S}_0$ denotes the signature assigned to $v$. Let $N_0 = \{w_1, \dots, w_x\} \subseteq V(G)$ denote the vertices with signatures $s \in \mathcal{S}_0$. For $0 \leq j \leq |N_0|$, we define $\mathcal{L}_j = \{\mu \in \{0, 1, \dots, \hs\}^z \mid \sum_{i=1}^{z}\mu(i) = \hs - j\}$, whenever $j \leq \hs$ and $\mathcal{L}_j = \emptyset$, otherwise. For $W \subseteq N_0$ and $\mu \in \mathcal{L}_{|W|}$, we also define $\hat{E}(W, \mu) = \{e \in E(G) \mid e\,\cap\,N_0 = W \,\land\, \forall 1\leq i \leq z, |e \,\cap\, N_i| = \mu(i)\}$.

The sets $\{\hat{E}(W, \mu)\}_{W\subseteq N_0, \mu \in \mathcal{L}_{|W|}}$ partition $E(G)$ as follows. We first partition $E(G)$ according to the intersection of hyperedges $e \in E(G)$ with $N_0$ and then partition further each block (corresponding to a subset $W \subseteq N_0)$ according to the multiplicities of signatures $s \in S \setminus\mathcal{S}_0$, described by the vectors $\mu \in \mathcal{L}_{|W|}$.

Let $\{\hat{k}(W, \mu)\}_{W \subseteq N_0, \mu \in \mathcal{L}_{|W|}}$ be a partition of $k$, that is, $k = \sum_{W \subseteq N_0, \mu \in \mathcal{L}_{|W|}}\hat{k}(W, \mu)$ (where $\hat{k}(W, \mu)$ may be zero). Let $A \subseteq E(G)$ be a set of $k$ hyperedges, such that, for each $W \subseteq N_0$, $\mu \in \mathcal{L}_{|W|}$, we have $|A \,\cap\, \hat{E}(W, \mu)| = \hat{k}(W, \mu)$.

We compute the contribution of $A$ to $\uncolholant(\Omega)$ as the product of the contribution of vertices in $N_0$ and the contribution of vertices in $V(G) \setminus N_0$, the latter of which is computed using similar arguments as in \Cref{lem:linearSigAlgo} as follows:

\begin{equation}\label{eq:productOfContributions}
   \prod_{v \in V(G) \setminus N_0}s_v(|A\,\cap\,E_G(v)|) = \prod_{i = 1}^{z}\prod_{v \in N_i}a_i^{|A\,\cap\, E_G(v)|} = \prod_{i = 1}^{z}a_i^{\sum_{v \in N_i}|A\,\cap\,E_G(v)|} = \prod_{i = 1}^{z}a_i^{\mathsf{exp}(a_i, A)},  
\end{equation}
where \[\mathsf{exp}(a_i, A) = \sum_{\substack{W \subseteq N_0}}\sum_{\mu \in \mathcal{L}_{|W|}}\mu(i)\cdot\hat{k}(W,\mu)\] is precisely the number of ways that any hyperedge $e \in A$ satisfies $e \in E_G(v)$, for some $v \in N_i$.

Let $G_0$ denote the hypergraph with $V(G_0) = N_0$ and $E(G_0) = \{e \in E(G) \mid e \subseteq N_0\}$. Let $\hat{A} = A\,\cap\,E(G_0)$. We have $|\hat{A}|$ = $\sum_{W \subseteq N_0,\,|W| = \hs}\hat{k}(W, \mathbf{0})$, where $\mathbf{0}$ is the zero vector. To see this, observe that for $W \subseteq N_0$ with $|W| = \hs$, we have $\mathcal{L}_{|W|} = \{\mathbf{0}\}$. Hence, if in addition $W \in E(G_0)$, we have $\hat{E}(W, \mathbf{0}) = \{W\}$ and $\hat{k}(W, \mathbf{0}) = |A\,\cap\,\{W\}|$ which is 1 if $W \in A$ and 0 otherwise, the former of which implies that $W \in \hat{A}$. Equivalently, we have $k-|\hat{A}| = \sum_{W \subseteq N_0,\,|W| < \hs}\sum_{\mu \in \mathcal{L}_{|W|}}\hat{k}(W, \mu)$.

Next, we write the contribution of vertices in $N_0$ to $\uncolholant(\Omega, k)$, in terms of $\hat{A}$ and $\hat{k}(W,\mu)$ as follows, so that in turn the overall contribution of $A$ to the holant value depends only on $\hat{A}$ and $\hat{k}(W,\mu)$:

\begin{equation}\label{eq:N0contribution}
\prod_{1 \leq q \leq x}s_{w_q}(|A \cap E_G(w_q)|) = \prod_{1 \leq q \leq x}s_{w_q}\left(|\hat{A} \cap E_G(w_q)| + \sum_{\substack{\{w_q\} \subseteq W \subseteq N_0 \\ |W| < \hs}}\sum_{\mu \in \mathcal{L}_W}\hat{k}(W, \mu)\right)\,.    
\end{equation}
To see this, note that
\begin{align*}
|A\,\cap\,E_G(w_q)| &= \sum_{\substack{\{w_q\} \subseteq W \subseteq N_0}}\sum_{\mu \in \mathcal{L}_{|W|}}\hat{k}(W, \mu) \\ &= \sum_{\substack{\{w_q\} \subseteq W \subseteq N_0 \\ |W| = \hs}}\sum_{\mu \in \mathcal{L}_{|W|}}\hat{k}(W, \mu) + \sum_{\substack{\{w_q\} \subseteq W \subseteq N_0 \\ |W| < \hs}}\sum_{\mu \in \mathcal{L}_{|W|}}\hat{k}(W, \mu)
\end{align*}
as well as
\[
\sum_{\substack{\{w_q\} \subseteq W \subseteq N_0 \\ |W| = \hs}}\sum_{\mu \in \mathcal{L}_{|W|}}\hat{k}(W, \mu) = \sum_{\substack{\{w_q\} \subseteq W \subseteq N_0 \\ |W| = \hs}}\hat{k}(W, \mathbf{0}) = |\hat{A}\,\cap\,E_G(w_q)|\,.
\]

Furthermore, the number of all $k$-hyperedge subsets $A'$ with $A'\,\cap\,E(G_0) = \hat{A}$ that induce the same partition of $k$ into $\sum_{\substack{W \subseteq N_0}}\sum_{\mu \in \mathcal{L}_{|W|}}\hat{k}(W, \mu)$ is

\begin{equation}\label{eq:partitionsHighArity}
\prod_{\substack{W \subseteq N_0 \\ |W| < \hs}}\prod_{\mu \in \mathcal{L}_{|W|}}\binom{|\hat{E}(W, \mu)|}{\hat{k}(W,\mu)}\,
\end{equation}
since as we showed previously, $k-|\hat{A}| = \sum_{W \subseteq N_0,\,|W| < \hs}\sum_{\mu \in \mathcal{L}_{|W|}}\hat{k}(W, \mu)$. Clearly, the aforementioned subsets yield the same contribution to the total holant value.

Given $\hat{A}$, we write $\mathsf{Cont}(\hat{A})$ for the contribution to $\uncolholant(\Omega, k)$ of all $A \subseteq E(G)$ with $|A| = k$ and $A\,\cap\,E(G_0) = \hat{A}$. We have

\begin{equation}\label{eq:UpperBoundsHolant}
\uncolholant(\Omega, k) = \sum_{\substack{\hat{A} \subseteq E(G_0) \\  |\hat{A}| \leq k}} \mathsf{Cont}(\hat{A})\,. 
\end{equation}

Now we are ready to present our algorithm, which works as follows. First, we partition $E(G)$ into $\bigcup_{W \subseteq N_0, \mu \in \mathcal{L}_{|W|}}{\hat{E}(W, \mu)}$. Clearly, we can decide in linear time, for each $e \in E(G)$, the block $\hat{E}(W, \mu)$ that $e$ belongs to and note that the number of those blocks depends only on $k$ and $\hs$.\footnote{Note that, it also depends on $z$, however, $z$ is however upper-bounded by $|\mathcal{S}|$ which is constant.}

We proceed by computing $\mathsf{Cont}(\hat{A})$, for each $\hat{A} \subseteq E(G_0)$ such that $|\hat{A}| \leq k$. Note that, $|E(G_0)|$ is at most $\binom{|N_0|}{\hs}$ which depends only on $k, \hs$ since we have assumed that $|N_0| \leq k\cdot\hs$. Hence, we compute $\mathsf{Cont}(\hat{A})$ for at most $h(k, \hs)$ subsets $\hat{A}$, for some computable function $h$.

Given $\hat{A}$, note that  we compute $\mathsf{Cont}(\hat{A})$ by iterating over all possibilities of partitioning $k - |\hat{A}|$ into $\sum_{\substack{W \subseteq N_0,\,|W| < \hs}}\sum_{\mu \in \mathcal{L}_{|W|}}\hat{k}(W, \mu)$ such that $\hat{k}(W, \mu) \leq \min\{k-|\hat{A}|, |\hat{E}(W, \mu)|\}$. The number of those possibilities ---with respect to a given partitioning--- is given by \eqref{eq:partitionsHighArity} and the contribution of the corresponding edge-subsets is given by the product of \eqref{eq:productOfContributions} and \eqref{eq:N0contribution}, which can be computed 
% in a straightforward way 
in FPT-(near)-linear time, that is, in time $g(k, \hs)\cdot\tilde{O}(|G|)$, for some computable function $g$.
\end{proof}

\begin{remark}[On allowing multiple copies of a hyperedge]Recall that all of our holant problems are defined over hypergraphs that feature \textit{no} multiple hyperedges. To lift our upper-bounds for the more general case in which multiple hyperedges are allowed, we work as follows: We first note that the analysis of \Cref{lem:linearSigAlgo} applies verbatim to the more general setting, hence yielding the same upper-bounds. If in addition, signatures $s$ with $s(0) = 0$ are considered, the analysis of \Cref{lem:linearSigAlgoWithZeros} can still be adapted in an easy way, but we would like to point out the only non-trivial point, which is \Cref{eq:UpperBoundsHolant} and in particular, the computation implied by it. The reason is that, in this case, the number $|\{\hat{A} \subseteq E(G_0) : |\hat{A}| \leq k\}|$ is not bounded by a function of $k, \hs$ and so we can cannot compute $\mathsf{Cont}(\hat{A})$ for \textit{each} $\hat{A} \subseteq E(G_0), |\hat{A}| \leq k$. Instead, we will group $\{\hat{A} \subseteq E(G_0) : |\hat{A}| \leq k\}$ into groups, the number of which is bounded by a function of our parameters, such that elements within the same group yield the same contribution. Concretely, we compute $\uncolholant(\Omega, k)$ as follows: First, we let $\{e_1^*, \dots, e_y^*\}$ denote the unique hyperedges of $E(G_0)$ and we let $C(e_i^*) \subseteq E(G_0)$ denote the multiset containing all copies of $e_i^*$ in $E(G_0)$. We have
\[\uncolholant(\Omega, k) = \sum_{\hat{k} = 0}^{k}\sum_{\{\hat{k}_i\}_{1\leq i \leq y}}\prod_{i = 1}^{y}\binom{|C(e_i^*)|}{\hat{k}_i}\mathsf{Cont}(\hat{A}(\{\hat{k}_i\}))\,,
\]
where the second sum is over all partitions $\{\hat{k}_i\}$ of $\hat{k}$ into $y$ non-negative integers (that may also be zero) and $\hat{A}(\{\hat{k}_i\})$ is \textit{any} of the subsets of $E(G_0)$ satisfying for each $1 \leq i \leq y$, $|\hat{A}(\{\hat{k}_i\}) \cap C(e_i^*)| = \hat{k}_i$. Since $y$ is clearly bounded by a function of $k, \hs$, we derive the same upper bounds for the setting in which multiple hyperedges are allowed as well.
\end{remark}

\subsection{\ensuremath{\uni{\constuncolholantprob{d}}}\label{sec:hardnessSubsection} is \ensuremath{\#\mathrm{W}[1]}-hard for signatures sets of type \ensuremath{\mathbb{T}[2]} or \ensuremath{\mathbb{T}[\infty]}}
In this section, we deal with the hard cases of $\uni{\constuncolholantprob{d}}$. In a nutshell, our approach towards $\#\mathrm{W}[1]$-hardness works as follows. Let $\mathcal{S}$ be a finite set of symmetric signatures $s$ such that $s(0) \neq 0$. In \cite{aivasiliotis_et_al:LIPIcs.ICALP.2025.7} it was shown that, whenever $\mathcal{S}$ is not of type $\mathbb{T}[1]$, $\uni{\constuncolholantprob{2}}(\mathcal{S})$ is $\#\mathrm{W}[1]$-hard. For certain signature sets, we rely on the aforementioned result and show hardness by reducing $\uni{\constuncolholantprob{2}}$ to  $\uni{\constuncolholantprob{d}}$, for any $d \geq 3$ which is achieved with the help of suitably defined gadgets. For the cases for which the aforementioned approach does not seem to be applicable, we show instead that we can use $\uni{\constuncolholantprob{d}}$ as an oracle in order to evaluate in FPT time certain homomorphism counts $\#\mathsf{Hom}(H, \star) : G \mapsto \#\mathsf{Hom}(H, G)$ for hypergraphs $H$ with large treewidth, a problem which is known to be $\#\mathrm{W}[1]$-hard \cite{dalmau2004complexity}. 

Before we proceed with the proof, we need to introduce some additional technical background.

\subsubsection{Tensor Product of Uniform Hypergraphs}\label{sec:Dedekind}

In this section, we define the tensor product of $\hs$-uniform hypergraphs, such that, for any $\hs \in \mathbb{Z}_{\geq 2}$, the class $\mathcal{G}(\hs)$ of all $\hs$-uniform hypergraphs is closed under taking tensor products. There are further useful properties of our tensor product, which will eventually allow us to apply Dedekind Interpolation for linear combinations of homomorphism counts over the semigroup $(\mathcal{G}(\hs), \otimes)$, to be shown momentarily. We note that the tensor product between hypergraphs and the properties discussed below are not new (see e.g., \cite{bressan2025complexitycountingsmallsubhypergraphs}). However, we need to slightly modify the way the tensor product is defined in the special case of uniform hypergraphs, and adapt the properties we will use accordingly.

\begin{definition}[Tensor Product of Uniform Hypergraphs]
Let $\hs \in \mathbb{Z}_{\geq 2}$. For any two $\hs$-uniform hypergraphs $G, H$, we write $G\otimes H$ for the hypergraph with vertex set $V(G) \times V(H)$ and hyperedges that satisfy the following: $\{(u_1, v_1), \ldots, (u_\hs, v_\hs)\}$ is a hyperedge of $G \otimes H$ if and only if $\{u_1, \ldots, u_\hs\} \in E(G)$ and $\{v_1, \ldots, v_\hs\} \in E(H)$.
\end{definition}

Note that, for $r=2$, the Tensor product coincides with the usual Tensor product for graphs.

\begin{proposition}
For any $\hs \in \mathbb{Z}_{\geq 2}$, the class $\mathcal{G}(\hs)$ of all (isomorphism classes) of $\hs$-uniform hypergraphs equipped with the tensor product of uniform hypergraphs, is a semigroup.   
\end{proposition}

\begin{proof}
By definition, the tensor product of two $\hs$-uniform hypergraphs $G, H$ is also $\hs$-uniform. Hence, what remains is to show that the tensor product is associative, that is, for any $\hs$-uniform hypergraphs $G, H, F$, we have $(G \otimes H)\otimes F = G \otimes (H \otimes F)$. It can be readily verified that the map that maps $((a, b), c) \in V((G \otimes H)\otimes F)$ to $(a, (b, c)) \in V(G \otimes (H\otimes F))$ is an isomorphism, which shows the claim.
\end{proof}

\begin{proposition}\label{prop:tensorprod} Let $F, G, H$ be $\hs$-uniform hypergraphs for some $\hs \in \mathbb{Z}_{\geq 2}$. We have
\begin{align*}
\Hom(F \to G \otimes H) = \Hom(F \to G) \cdot \Hom(F \to H)\,.
\end{align*}
\end{proposition}
\begin{proof}
Consider $b$ that maps elements of $\homs{F}{G\,\otimes\,H}$ to elements of $\homs{F}{G}\,\times\,\homs{F}{H}$ as follows. For a homomorphism $h \in \homs{F}{G\,\otimes\,H}$, we define $h_G$ and $h_H$ such that $h_G$ (resp. $h_H$) maps a vertex $v \in V(F)$ to $h(v)_1$ (resp. $h(v)_2$), where $h(v)_i$ denotes the $i$-th entry of the tuple $h(v)$. For any hyperedge $\{v_1, \ldots, v_\hs\} \in E(F)$, we have $\{h(v_1), \ldots, h(v_\hs)\} \in E(G\otimes H)$ which implies that $\{h(v_1)_1, \ldots, h(v_\hs)_1\} \in E(G)$ and $\{h(v_1)_2, \ldots, h(v_\hs)_2\} \in E(H)$. Clearly, $h_G \in \homs{F}{G}$ and $h_H \in \homs{F}{H}$.

We have that $b$ is a bijection. 
To see that it is injective, let $h, h' \in \homs{F}{G\,\otimes\,H}$ and assume that $b(h) = b(h') = (h_G, h_H)$. For each $v \in V(F)$ we have $h(v) = (h_G(v), h_H(v)) = h'(v)$. 
Then, for surjectivity, we take $h_G \in \homs{F}{G}$ and $h_H \in \homs{F}{H}$ and define $h$ that maps $v \in V(F)$ to $(h_G(v), h_H(v))$. It can be verified that indeed $h \in \homs{F}{G\,\otimes\,H}$ and $b(h) = (h_G, h_F)$.
\end{proof}

\begin{proposition}\label{prop:lovaszNonIso}
For non-isomorphic $\hs$-uniform hypergraphs $F, H$ without isolated vertices there exists $X$ that is a sub-hypergraph\footnote{Recall that a sub-hypergraph is by definition hyperedge-induced.} either of $F$ or $H$ such that $\Hom(F \to X) \neq \Hom(H \to X)$. 
\end{proposition}
\begin{proof}
    For any two hypergraphs $A,B$ we set
    \[ \surhoms{A}{B} := \{ \varphi \in \homs{A}{B} \mid \varphi \text{ is edge-surjective}\}\,\]
    where a homomorphism $\varphi$ from $A$ to $B$ is \textit{edge-surjective} if and only if $\varphi(E(A)) = E(B)$.
    Recall that given a subset $J \subseteq E(B)$, $B[J]$ is the sub-hypergraph of $B$ induced by $J$, that is, $V(B[J])=\bigcup_{e \in J}e$ and $E(B[J])=J$. Clearly, if $B$ is $\hs$-uniform then $B[J]$ is also $\hs$-uniform for any $J \subseteq V(B)$.
    By the principle of Inclusion-Exclusion, we have
    \[\#\surhoms{A}{B} = \sum_{J \subseteq E(B)} (-1)^{|E(B)\setminus J|} \cdot \#\homs{A}{B[J]} \,.\]
    Towards contradiction, assume that that for all sub-hypergraphs $X$ of $F$ and $H$ we have
    \[ \#\homs{F}{X} = \#\homs{H}{X} \,.\]
    Then, we compute
    \begin{align*}
        \#\surhoms{F}{H} &= \sum_{J \subseteq E(H)} (-1)^{|E(H)\setminus J|} \cdot \#\homs{F}{H[J]}\\
        ~& = \sum_{J \subseteq E(H)} (-1)^{|E(H)\setminus J|} \cdot \#\homs{H}{H[J]}\\
        ~&= \#\surhoms{H}{H} > 0\,.
    \end{align*}
    Similarly, we have $\#\surhoms{H}{F}>0$. It is well-known that two hypergraphs $F, H$ (without isolated vertices) are isomorphic if and only if $\#\surhoms{F}{H} > 0$ as well as $\#\surhoms{H}{F} > 0$. Hence, $F \cong H$, which is a contradiction. So, there is a sub-hypergraph $X$ of either $F$ or $H$ such that $\#\homs{F}{X} \neq \#\homs{H}{X}$. 
\end{proof}

The last ingredient is the following technical lemma that follows from \cite{dalmau2004complexity} and \cite{Marx10}. The formal argument requires the consideration of homomorphisms between relational structures in an intermediate step and is provided in \cref{sec:ProofOfLemma}.

\begin{lemma}\label{lem:lowerBoundsHoms}
Let $r\in \mathbb{Z}_{\geq 2}$ be a fixed integer, and let $\mathcal{H}$ be a recursively enumerable class of $r$-uniform hypergraphs. If $\mathcal{H}$ has unbounded treewidth, then $\#\text{\sc{Hom}}(\mathcal{H})$ is $\#\mathrm{W}[1]$-hard and cannot be solved in time $f(|H|)\cdot |V(G)|^{o(\mathsf{tw}(H)/\log(\mathsf{tw}(H))))}$, for any function $f$, unless ETH fails.  This remains true even if each input $(H,G)$ to $\#\text{\sc{Hom}}(\mathcal{H})$ satisfies that $G$ is $r$-uniform as well.
\end{lemma}

We are now ready to prove the theorem.

\begin{theorem}\label{thm:reduction2TodArity}
If $\{s\}$ is of type $\mathbb{T}[2]$ or $\mathbb{T}[\infty]$, $\uni{\constuncolholantprob{d}}(\{s\})$ is $\#\mathrm{W}[1]$-hard, for any number $d \in \mathbb{Z}_{\geq 2}$. Furthermore, if $\{s\}$ is of type $\mathbb{T}[\infty]$, then $\uni{\constuncolholantprob{d}}$ cannot be solved in time $f(k)\cdot |V(\Omega)|^{o(k/\log k)}$, unless ETH fails.
\end{theorem}
\begin{proof}
For $d = 2$, the theorem follows from \cite[Theorem 6.2]{aivasiliotis2024parameterisedholantproblems}. So, for what follows we assume that $d \geq 3$. 

Let $b > 0$ such that $s(b) \neq 0$ and for each $0 < i < b$, $s(i) = 0$. We distinguish between the cases $b = 1, b = 2,$ and $b > 2$. For the first two cases, we show the following reduction
\[\uni{\constuncolholantprob{2}}(\{s\}) \fptlinred \uni{\constuncolholantprob{d}}(\{s\})\,,\]
from which follows that $\constuncolholantprob{d}(\{s\})$ is $\#\mathrm{W}[1]$-hard. Also, note that the known lower bounds for $\uni{\constuncolholantprob{2}}(\{s\})$ (that is, for the case where $\{s\}$ is of type $\mathbb{T}[\infty]$) transfer directly to $\uni{\constuncolholantprob{d}}(\{s\})$, for any $d \geq 3$.

For the last case, we rely on \Cref{lem:lowerBoundsHoms} and instead of reducing from $\uni{\constuncolholantprob{2}}(\{s\})$ we reduce from $\#\textsc{HOM}(\mathcal{F}_{d,b})$ where we write $\mathcal{F}_{d, b}$ for the class of all $d$-uniform $b$-regular hypergraphs. As we will show, $\mathcal{F}_{d, b}$ has unbounded treewidth and it contains hypergraphs the treewidth of which is linear in their size, with which observations we can then apply \Cref{lem:lowerBoundsHoms} and derive our hardness results.

We also note that in the case $b > 2$, $\{s\}$ can only be of type $\mathbb{T}[\infty]$. To see this, note that by \Cref{def:fingerprint_intro}, for any signature $s$ satisfying $s(1) = s(2) = 0$, $\{s\}$ is either of type $\mathbb{T}[1]$ or $\mathbb{T}[\infty]$ (and we have assumed that we only consider signatures $s$ such that $\{s\}$ is not of type $\mathbb{T}[1]$). We are now ready to proceed with the details of each case.
\begin{itemize}
    \item[(1)]$\boldsymbol{b=1}:\quad$ Let $(\Omega, k)$ be an instance of $\uni{\constuncolholantprob{2}}(\{s\})$ with underlying graph $G$. Let $G'$ denote the $d$-uniform hypergraph obtained from $G$ by adding to each edge $e \in E(G)$, $d-2$ new vertices (that is, all new vertices have degree 1). We further equip the new vertices with signature $s$ and obtain a signature hypergrid $\Omega'$. Note that $(\Omega', k) \in \uni{\constuncolholantprob{d}}(\{s\})$. We have
    \begin{align*}
    \uncolholant(\Omega', k) & = \sum_{\substack{A \subseteq E(G') \\ |A| = k}}\prod_{v \in V(G')}s(|A \cap E_{G'}(v)|) \\ & = s(1)^{k(d-2)} \underbrace{\sum_{\substack{A \subseteq E(G') \\ |A| = k}}\prod_{v \in V(G)}s(|A\cap E_{G'}(v)|)}_{= \uncolholant(\Omega, k)}\,.\end{align*}
    Since $s(1) \neq 0$, it follows that $\uncolholant(\Omega, k) = s(1)^{-k(d-2)}\cdot\uncolholant(\Omega', k)$ which thus can be computed in polynomial time using our oracle for $\uncolholant(\Omega', k)$, showing the reduction.

    \item[(2)]$\boldsymbol{b = 2}: \quad$ Let $(\Omega, k)$ and $G$ be defined as above. Since $s(1) = 0$, the above reduction cannot be used as it is in this case. However, we devise appropriate gadgets that will help us overcome this. Let $B$ be a connected $d$-uniform hypergraph with precisely two vertices $x, y$ having degree 1 (the existence of which is established below). We assume that $x, y$ are adjacent (that is, they are contained in the same hyperedge) and that every other vertex has degree 2.
    Let $G'$ denote the hypergraph obtained by replacing each edge $\{u, v\} \in E(G)$ with a copy of $B$ by identifying $v$ (resp. $u$) with $x$ (resp. $y$). We also equip each new vertex with signature $s$, obtaining a new signature grid $\Omega'$.
    
    To compute $\uncolholant(\Omega', k')$, for any $k'$, we observe that any $k'$-set $A \subseteq E(G')$ that contains at least one but not all of the hyperedges of any copy of $B$, has zero contribution to $\uncolholant(\Omega', k')$. To see this, recall that $s(1) = 0$ and note that there would exist at least one vertex $w \in B$ ($w \neq x, y$) such that $|A \cap E_{G'}(w))| = 1$, since $B$ is connected.

    We set $k' = k \cdot |E(B)|$ and for each $e \in E(G)$, we write $B^e \subseteq G'$ for the copy of $B$ that $e$ is replaced with in $G'$. We let $\mathcal{A}'$ denote the set of all $k'$-sets $A' \subseteq E(G')$ for which there exist $e_1, \dots, e_k \in E(G)$ such that $A' = \bigcup_{i = 1}^{k}E(B^{e_i})$. It follows from the discussion above that

    \[
    \uncolholant(\Omega', k') = \sum_{ A'\in \mathcal{A}'}\prod_{v \in V(G')}s(|A' \cap E_{G'}(v)|)\,.
    \]
    We let $\binom{E(G)}{k}$ denote all $k$-subsets of $E(G)$.
    Let $a:\mathcal{A}' \to \binom{E(G)}{k}$ denote the mapping assigning a set $A'= \bigcup_{i = 1}^{k}E(B^{e_i})$ to the set $\{e_1, \dots, e_k\}$. By construction of $\mathcal{A}'$ it is easy to see that $a$ is a bijection. Moreover, observe that 
    \[
    \prod_{v \in V(G')}s(|A' \cap E_{G'}(v)|) = s(2)^{k(|V(B)|-2)}\prod_{v \in V(G)}s(|a(A') \cap E_G(v)|)\,.\] 
    So, we have 
    \begin{align*}
    \uncolholant(\Omega', k') & = \sum_{\substack{A' \in a(\mathcal{A}')}}s(2)^{k(|V(B)|-2)}\prod_{v \in V(G)}s(|A' \cap E_{G}(v)|) \\ & = s(2)^{k(|V(B)|-2)} \cdot \uncolholant(\Omega, k)\,,
    \end{align*}
    and since $s(2) \neq 0$, we have $\uncolholant(\Omega, k) = s(2)^{-k(|V(B)|-2)}\cdot \uncolholant(\Omega', k')$, which thus can be computed in polynomial time using our oracle for $\uncolholant(\Omega', k')$ (since $|V(B)| + |E(B)|$ is constant). 

    To complete the reduction, what remains is to show that such a hypergraph $B$ always exists.
    By \cite[Theorem 3]{frosini2021new}, it follows that any degree sequence with sufficiently many entries of $2$ and precisely $d-2$ entries of $1$ (and no entries larger than $2$) is realizable by a $d$-uniform hypergraph.
    Let hence $W$ be a $d$-uniform hypergraph of maximum degree 2 having precisely $d-2$ vertices of degree 1. First, we delete all connected components of $W$ in which each vertex has degree $2$; consequently, each remaining connected component contains at least one vertex of degree $1$.
    Finally, we obtain $B$ by first adding two fresh vertices $x$ and $y$, and then forming a new hyperedge that contains all $(d-2)$ degree-1 vertices along with $x$ and $y$. Since $d$ is a constant, it follows that we can also compute $B$ in constant time (irrespective of whether the mapping $d \mapsto B$ that constructs $B$ is even decidable).

    \item[(3)] $\boldsymbol{b\geq 3}:\quad$ For the last case, we note that the existence of the gadgets (similar to the ones) employed above cannot be guaranteed. So, we need to follow a different approach. To this end, recall from \Cref{lem:uncolHomBasisRESTATED}, that for any $(\Omega, k) \in \uni{\constuncolholantprob{d}}(\{s\})$ with underlying hypergraph $G$ we have \[
\mathsf{Holant}(\Omega, k) = \sum_{F \in \mathcal{G}_{\leq k}(d)}\zeta_{k, s}(F)\cdot \Hom(F \to G)\,,
\]
where the coefficients of the above expression are given as follows
\begin{align*}
\zeta_{k, s}(F) = \sum_{H \in \mathcal{G}_k(d)}\left(\Aut(H)^{-1}\prod_{v \in V(H)}s(\deg_H(v))\sum_{\substack{\rho \in \mathsf{Part}(H) \\ F \cong H/\rho, }}\mu(\rho)\right).
\end{align*}

We also need to recall some simple properties of uniform and regular hypergraphs. For any $b$-regular hypergraph $X$, we have $\sum_{v \in V(X)}\deg_X(v) = |V(X)|\cdot b$. Additionally, if $X$ is also $d$-uniform, it follows from the generalization of the Handshaking lemma for hypergraphs that $\sum_{v \in V(X)}\deg_X(v) = d\cdot |E(X)|$. Hence, for any $d$-uniform $b$-regular hypergraph $X$, we have $|V(X)|\cdot b = d \cdot |E(X)|$. 

Recall that we write $\mathcal{F}_{d, b}$ for the class of all $d$-uniform $b$-regular hypergraphs. Let $F^* \in \mathcal{F}_{d, b}$ and $G$  be an instance of $\#\textsc{Hom}(\mathcal{F}_{d, b})$. We set $k = |E(F^*)|$. We observe that any hypergraph $H \in \mathcal{G}_k(d)$ that has at least one vertex of degree strictly less than $b$, has zero contribution to $\zeta_{k, s}(F^*)$, since $s(x) = 0$, for any $0 < x < b$ (recall that $H$ has no isolated vertices). In fact, the only hypergraph $H \in \mathcal{G}_k(d)$ for which there is a partition $\rho \in \mathsf{Part}(H)$ such that $F^* \cong H/\rho$ is $F^*$ itself (with the corresponding partition being the finest partition). To see this, recall that we have safely assumed that the vertices of $H$ all have degree at least $b$ and so, $|V(H)|\cdot b \leq \sum_{v \in V(H)}\deg_H(v) = d\cdot k = |V(F^*)|\cdot b$, which implies that $|V(H)| \leq |V(F^*)|$. Since $F^*$ is a quotient hypergraph of $H$, we have $|V(F^*)| \leq |V(H)|$. Hence, $|V(H)| = |V(F^*)|$ or equivalently $\rho = \bot$, which in turn implies that $H = F^*$ and thus $\zeta_{k, s}(F^*) = \#\auts(F^*)^{-1}\cdot s(b)^{|V(F^*)|} \neq 0$ (since $\mu(\bot) = 1$). 

For any $d$-uniform hypergraph $X$, let $\Omega_X$ denote the signature grid with underlying hypergraph $G \otimes X$, the vertices of which are all equipped with $s$. We have 
\begin{align*}
\uncolholant(\Omega_X, k) &= \sum_{F \in \mathcal{G}_{\leq k}(d)}\zeta_{k, s}(F)\cdot \Hom(F \to G \otimes X) \\ &= \sum_{F \in \mathcal{G}_{\leq k}(d)}\underbrace{\zeta_{k, s}(F)\cdot \Hom(F \to G)}_{=\hat{\zeta}_{k,s}(F)}\cdot\Hom(F \to X)\,,
\end{align*}
where the coefficients $\zeta_{k, s}(F)$ are as given above, and the second inequality follows from \Cref{prop:tensorprod}.
With the observations of \Cref{sec:Dedekind} in hand, we can apply Dedekind Interpolation (see e.g., the full version of \cite[Theorem 18]{10.1145/3564246.3585204}) for the map \[p_G : X \in \mathcal{G}(d) \mapsto \sum_{F \in \mathcal{G}_{\leq k}(d)}\hat{\zeta}_{k, s}(F)\cdot \#\homs{F}{X}\,,\] where $\mathcal{G}(d)$ is the class of all (isomorphism types) of $d$-uniform hypergraphs (with no isolated vertices). Clearly, for any $X \in \mathcal{G}(d)$, $p_G(X) = \uncolholant(\Omega_X, k)$. Thus, we can compute $\hat{\zeta}_{k, s}(F)$ for any $F \in \mathcal{G}_{\leq k}(d)$. If in addition $\zeta_{k, s}(F) \neq 0$, we can also compute $\Hom(F \to G) = \hat{\zeta}_{k, s}(F)/\zeta_{k,s}(F)$. Recall that $\zeta_{k, s}(F^*) \neq 0$, which implies that we can compute $\Hom(F^* \to G)$ using Dedekind Interpolation. In particular, it can be verified that this can be done in FPT-(near)-linear time, with which we deduce that
\[
\#\text{\sc{HOM}}(\mathcal{F}_{d, b}) \fptlinred \uni{\constuncolholantprob{d}}(\{s\})\,.
\] 

Next, in order to apply \Cref{lem:lowerBoundsHoms}, we first need to show that $\mathcal{F}_{d, b}$ has unbounded treewidth. To derive lower-bounds we also show that $\mathcal{F}_{d, b}$ contains hypergraphs the treewidth of which is linear in their size (recalling that we have taken $k = |E(F)|$).

\begin{claim}
Fix $d \geq 2$ and $b \geq 3$. There is a family $(F_n)_{n\in \mathbb{N}}$ of $d$-uniform $b$-regular hypergraphs such that $|V(F_n)|,|E(F_n)|\in \Theta(n)$, and $|\mathsf{tw}(F_n)| \geq \Omega(n)$.   
\end{claim}
\begin{proof}
Given a graph $G$, we write $i(G)$ for the edge-expansion (also called the isoperimetric number) of $G$, and we write $\mathsf{vx}(G)$ for the vertex expansion of $G$. For $b$-regular graphs, edge end vertex expansion are equal up to a constant; in particular, $\mathsf{vx}(G) \geq \frac{i(G)}{b}$. It is well known that expander graphs have large treewidth; concretely we will use $\mathsf{tw}(G)\in \Omega(\mathsf{vx}(G)\cdot |V(G)|)$ (see e.g., \cite[Proposition 1]{Grohe&2009expansion} for $\alpha=1/2$).

We now provide a construction of the $F_n$. First, assume w.l.o.g.\ that $2d$ divides $n$ --- otherwise just set $F_n = F_{n+1} = \dots = F_{n+j}$ such that $2d|(n+j)$; clearly $n+j \in \Theta(n)$ since $d$ is a constant.

By a classical result of Bollob{\'{a}}s~\cite{Bollobas88}, for every $n$ with $bn$ even, a random $b$-regular $n$-vertex \emph{graph} $G$ has edge expansion $i(G)\geq b/20$ with probability converging to $1$ as $n \to \infty$.\footnote{Note that the explicit bound $i(G)\geq b/20$ can be obtained by invoking \cite[Theorem 1]{Bollobas88} with $\eta=0.9$.} In particular, there is hence a positive integer $N_1$ such that for all $n\geq N_1$ with $bn$ even, there exists a $b$-regular graph $G_n$ with $i(G_n)\geq b/20$. 

For all $n\geq N_1$ such that $2d$ divides $n$, we first construct $\hat{F}_n$ from $G_n$ by adding $(d-2)$ fresh vertices to each edge of $G_n$. Next, let $m=|E(G_n)|$ and note that $m=\frac{1}{2}nb$ by the Handshaking Lemma. 

Now observe that we added $n':= m\cdot(d-2)$ fresh vertices. Observe furthermore that $d$ divides $m$ since $2d$ divides $n$. In particular, we obtain that $d$ divides $n'$. By \cite[Theorem 6]{frosini2021new}, there is hence a positive integer $N_2$ such that, if $n'\geq N_2$, there exists a $d$-uniform $(b-1)$-regular hypergraph $H_{n'}$ on $n'$ vertices and $(b-1)n'/d$ edges.

Finally, we construct $F_n$ from $\hat{F}_n$ by considering a bijection between the degree-1 vertices of $\hat{F}_n$ and the vertices of $H_{n'}$ and then proceed with the identification of corresponding vertices. Clearly, this results in a $d$-uniform $b$-regular hypergraph.
Now observe that 
\begin{align*}
    |V(F_n)| &=n+n' = n+(d-2)m= n+(d-2)\frac{1}{2}nb \in \Theta(n)\,,\\
    |E(F_n)|&= m + \frac{(b-1)n'}{d} = \frac{1}{2}nb + \frac{b-1}{d}(d-2)\frac{1}{2}nb \in \Theta(n)\,,\\
    \mathsf{tw}(F_n)&\stackrel{(\ast)}{\geq} \mathsf{tw}(G_n)\geq \Omega(\mathsf{vx}(G_n)\cdot n) \geq \Omega(i(G_n)\cdot n) \geq \Omega\left(\frac{b}{20}n\right)=\Omega(n)\,,
\end{align*}
where $(\ast)$ holds since $G_n$ is a subgraph of $F_n$ and treewidth cannot increase by taking subgraphs.
This concludes our construction for the cases $n\geq N_1$ and $n' \geq N_2$. Since $N_1,N_2$ are constants, we can hence, for smaller $n$ (and corresponding $n'$), just set $F_n$ to be any fixed $b$-regular $d$-uniform hypergraph.
\end{proof}
By the claim above, we have that $\mathcal{F}_{d, b}$ has unbounded treewidth and in particular contains hypergraphs with treewidth linear in their size. Thus, due to \Cref{lem:lowerBoundsHoms}, $\uni{\constuncolholantprob{d}}(\{s\})$ is $\#\mathrm{W}[1]$-hard and cannot be solved in time $f(k) \cdot |V(\Omega)|^{o(k/\log(k))}$.
\end{itemize}
\end{proof}

\section{\ensuremath{\#\mathrm{P}}-hardness for \ensuremath{\constclassHolant{d}}}\label{sec:NonParamHardness}

In this part of the paper, we will establish $\#\mathrm{P}$-hardness for all instances of $\ensuremath{\constclassHolant{d}}$ that are not solvable in polynomial-time via our algorithm in \Cref{lem:linearSigAlgoWithZeros}.

We start with instances of $\ensuremath{\constclassHolant{d}}$ in which at least one of the signatures satisfies $s(0)=0$.

\subsection{Hardness for signatures sets containing $s$ with $s(0)=0$}\label{sec:HardnessZeroSigs}
We have seen in \Cref{thm:mainParamThm} that $\classHolant(\mathcal{S})$ and $\constclassHolant{d}(\mathcal{S})$ are fixed-parameter tractable if all signatures $s\in \mathcal{S}$ satisfy $s(0)=0$. In this section, we will show that those instances are ``real'' FPT cases in the sense that even a single signature $s$ with $s(0)=0$ makes our holant problems $\#\mathrm{P}$-hard --- recall that no signature is allowed to be the constant zero function, hence any signature $s$ with $s(0)=0$ satisfies that there is some $b>0$ with $s(b)\neq 0$. 

We will obtain $\#\mathrm{P}$-hardness via reducing from the problem $\#\textsc{PerfectMatching}^d$ of counting perfect matchings in $d$-uniform hyperrgaphs. Recall that a perfect matching of a (hyper)graph $G$ is a set $A$ of hyperedges such that each vertex of $G$ is contained in precisely one edge of $A$. In case of graphs, i.e., for $d=2$, we will just drop the $d$ in the superscript. We write $\#\mathsf{PerfMatch}(G)$ for the number of perfect matchings of $G$. For $d=2$, the $\#\mathrm{P}$-hardness result of $\#\textsc{PerfectMatching}^d$ was established in the seminal paper by Valiant on computational counting~\cite{valiant1979complexity}. For $d>2$, $\#\mathrm{P}$-hardness of $\#\textsc{PerfectMatching}^d$ was established in~\cite{creignou1996complexity}.

We begin our reductions by establishing hardness for $s(0)=0$ for the case of graphs, i.e., for $d=2$.
\begin{theorem}\label{thm:PerfectMatchingReduction}
For any signature $s$ satisfying $s(0) = 0$ that is not the zero function, $\classHolant(\{s\})$ is $\#\mathrm{P}$-hard.    
\end{theorem}
\begin{proof}
We show the following reduction
\[\#\textsc{PerfectMatching} \leq_{\mathsf{T}} \classHolant(\{s\})\,,\]
which yields $\#\mathrm{P}$-hardness for $\classHolant(\{s\})$ via~\cite{valiant1979complexity}. 

We first show the reduction assuming that $s(1) \neq 0$. To this end, let $G$ be an instance of $\#\textsc{PerfectMatching}$ and assume that $n = |V(G)|$ is even (since otherwise, $\#\textsf{PerfMatch}(G) = 0)$. Let $\Omega$ denote the signature grid with underlying graph $G$ the vertices of which are equipped with signature $s$. It can be verified, that
\[
\uncolholant(\Omega, n/2) = s(1)^n\#\mathsf{PerfMatch}(G).
\]
To see this note, that since $s(0) = 0$, any set $A \subseteq E(G)$ that is not an edge-cover\footnote{Given a graph $G$, a set $A \subseteq E(G)$ is an edge-cover if $\bigcup_{e \in A}e = V(G)$.} has zero contribution to $\uncolholant(\Omega, k)$. Also, note, that any edge-cover of size $n/2$ of a $n$-vertex graph is a perfect matching and that for any perfect matching $A \subseteq E(G)$, we have $\prod_{v \in V(G)}s(|A \cap E(v)|) = \prod_{v \in V(G)}s(1) = s(1)^n$.

To generalize for any signature, let $b \in \mathbb{Z}_{>0}$, such that $s(b) \neq 0$ and for each $i < b, s(i) = 0$ (such $b$ exists since we have assumed that $s$ is not the zero function). Let $K_{b-1}$ denote the clique on $b-1$ vertices. We construct a graph $G'$ obtained from $G$ as follows. For each $v \in V(G)$, we add a copy of $K_{b-1}$, denoted as $K_{b-1}^{(v)}$ and we add an edge between $v$ and every vertex of $K_{b-1}^{(v)}$. Next, we arbitrarily group $V(G)$ into $n/2$ pairs and for each pair $v, u$ we connect $K_{b-1}^{(v)}$ and $K_{b-1}^{(u)}$ with a matching of size $b-1$. Note, that since $b$ is a constant, $|G'| = O(|G|)$.

Let $\Omega'$ denote the signature grid with underlying graph $G'$ the vertices of which are equipped with signature $s$. Observe, that every vertex $v \in V(G') \setminus V(G)$ has degree $b$. Hence, for any $k$, the contribution to $\uncolholant(\Omega',k)$ of any $k$-set $A \subseteq E(G')$ that misses at least one edge from $E(G') \setminus E(G)$ is trivially zero. That said, for $k = n/2 + |E(G') \setminus E(G)|$, it can be verified that
\[
\uncolholant(\Omega', k) = s(b)^{|V(G')|}\#\mathsf{PerfMatch}(G)\,.
\]
To see this, recall that, as argued above, every (non-trivial) $k$-set $A \subseteq E(G')$ contains all edges in $E(G') \setminus E(G)$, which implies that the remaining $n/2$ edges should consist an edge-cover of $G$, since every vertex $v \in V(G)$ has exactly $b-1$ neighbors in $G' \setminus G$ and $s(b-1) = 0$. Recall that an edge-cover of size $n/2$ of an $n$-vertex graph is a perfect matching. So, the contribution of $A$ to $\uncolholant(\Omega', k)$ is precisely $\prod_{v \in V(G')}s(|A \cap E_{G'}(v)|) = s(b)^{|V(G')|}$.

Clearly, $s(b)^{|V(G')|}$ can be computed in time that is polynomial in $|G|$ by using standard techniques, which concludes the proof.
\end{proof}

Next, we establish $\#\mathrm{P}$-hardness for $d$-uniform hypergraphs with $d>2$. Since we have little control over $s$ --- we only assume that $s(0)=0$ --- it turns out to be easier to reduce directly from counting perfect matchings in $d$-uniform hypergraphs rather than to reduce from the case of $d=2$.
\begin{theorem}\label{thm:HyperPerfectMatchingReduction}
For any signature $s$ satisfying $s(0) = 0$ that is not the zero function, and any number $d \in \mathbb{Z}_{\geq 2}$, $\constclassHolant{d}(\{s\})$ is $\#\mathrm{P}$-hard.    
\end{theorem}
\begin{proof}
Recall that $\#\textsc{PerfectMatching}^{d}$ denotes the problem that takes as input a $d$-uniform hypergraph $G$ (where $d \in \mathbb{Z}_{\geq 2}$ is a constant) and computes the number $\#\mathsf{PerfMatch}(G)$ of perfect matchings of $G$, that is, $\#\mathsf{PerfMatch}(G)$ is the number of all partitions of $V(G)$ into (pair-wise disjoint) sets $e \in E(G)$. Following the lines of the proof of \Cref{thm:PerfectMatchingReduction}, we show the following reduction 
\[\#\textsc{PerfectMatching}^{d} \leq_{\mathsf{T}} \constclassHolant{d}(\{s\})\,,\]
which yields $\#\mathrm{P}$-hardness for $\constclassHolant{d}(\{s\})$ since $\#\textsc{PerfectMatching}^{d}$ is $\#\mathrm{P}$-hard~\cite{creignou1996complexity}.

We first show the reduction assuming that $s(1) \neq 0$. To this end, let $G$ be an instance of $\#\textsc{PerfectMatching}^{d}$ and assume that $n = |V(G)|$ is a multiple of $d$ (since otherwise, $\#\textsf{PerfMatch}(G) = 0)$. Let $\Omega$ denote the signature grid with underlying hypergraph $G$ the vertices of which are equipped with signature $s$. It can be verified, that
\[
\uncolholant(\Omega, n/d) = s(1)^n\#\mathsf{PerfMatch}(G).
\]
To see this note, that since $s(0) = 0$, any set $A \subseteq E(G)$ that is not a hyperedge-cover\footnote{Given a hypergraph $G$, a set $A \subseteq E(G)$ is a hyperedge-cover if $\bigcup_{e \in E(G)}e = V(G)$.} has zero contribution to $\uncolholant(\Omega, k)$. Also, note, that any hyperedge-cover of size $n/d$ of a $n$-vertex graph is a perfect matching and that for any perfect matching $A \subseteq E(G)$, we have $\prod_{v \in V(G)}s(|A \cap E(v)|) = \prod_{v \in V(G)}s(1) = s(1)^n$.

To generalize for any signature, let $b \in \mathbb{Z}_{> 0}$, such that $s(b) \neq 0$ and for each $i < b, s(i) = 0$ (such $b$ exists since we have assumed that $s$ is not the zero function). We may assume that $b \geq 2$. Let $F(d, b-1)$ be a $d$-uniform $(b-1)$-regular hypergraph. If $F(d, b-1)$ has sufficiently many vertices, then is guaranteed to exist for any $d, b$ (see, \cite[Theorem 3]{frosini2021new}). Hence, we may assume that $F(d, b-1)$ has at least $(d-1)(b-1)$ vertices. We construct a hypergraph $G'$ obtained from $G$ as follows. For each $v \in V(G)$, we add a copy of $F(d, b-1)$, denoted as $F(d, b-1;v)$. We pick $(d-1)(b-1)$ vertices of $F(d, b-1;v)$ arbitrarily, which we partition into $b-1$ sets of size $d-1$, namely $S_1, \dots, S_{b-1}$. For each $1 \leq i \leq b-1$, we create a new hyperedge $S_i \cup \{v\}$.

Next, we proceed by increasing by one the degree of each vertex of $F(d, b-1; v)$ having degree $b-1$ (that is, of each vertex of $F(d, b-1;v)$ that is not adjacent to $v$). To this end, recalling that $n$ is a multiple of $d$, we group $V(G)$ into pair-wise disjoint sets of size $d$. Let $v_1, \dots, v_{d}$ be such a group. We partition the vertices of $\bigcup_{i =1}^{d}F(d, b-1;v_i)$ having degree $b-1$ into pair-wise disjoint sets of size $d$, which we include into the hyperedges of $G'$. We repeat the aforementioned procedure for all groups and thus we obtain $G'$. Note, that since $b$ and $d$ are constants, $|G'| = O(|G|)$.

Let $\Omega'$ denote the signature grid with underlying hypergraph $G'$ the vertices of which are equipped with signature $s$. Observe, that every vertex $v \in V(G') \setminus V(G)$ has degree $b$. Hence, for any $k$, the contribution to $\uncolholant(\Omega',k)$ of any $k$-set $A \subseteq E(G')$ that misses at least one hyperedge from $E(G') \setminus E(G)$ is trivially zero. 

That said, for $k = n/d + |E(G') \setminus E(G)|$, since every (non-trivial) $k$-set $A \subseteq E(G')$ contains all hyperedges in $E(G') \setminus E(G)$, it follows that the remaining $n/d$ hyperedges should consist a hyperedge-cover of $G$. To see this, note that every vertex $v \in V(G)$ has exactly $b-1$ incident hyperedges intersecting with $G' \setminus G$ and $s(b-1) = 0$. Recall that a hyperedge-cover of size $n/d$ of an $n$-vertex $d$-uniform hypergraph is a perfect matching. So, the contribution of $A$ to $\uncolholant(\Omega', k)$ is precisely $\prod_{v \in V(G')}s(|A \cap E_{G'}(v)|) = s(b)^{|V(G')|}$. Hence,
\[
\uncolholant(\Omega', k) = s(b)^{|V(G')|}\#\mathsf{PerfMatch}(G)\,.
\]

Clearly, $s(b)^{|V(G')|}$ can be computed in time that is polynomial in $|G|$ by using standard techniques, which concludes the proof.
\end{proof}

\subsection{Hardness for Signatures of Types $\mathbb{T}[2]$ or $\mathbb{T}[\infty]$ on Graphs: Reduction via Curticapean's CFI-Filters}\label{sec:HardnessCFI}
We continue with signatures of types $\mathbb{T}[2]$ or $\mathbb{T}[\infty]$ and first handle the case of graphs ($d=2$). Our overall strategy is to use again the homomorphism basis of our holant problems on graphs: while this framework (``complexity monotonicity'') usually only works for establishing $\#\mathrm{W}[1]$-hardness of counting problems, we will rely here on a recent modification of the framework due to Curticapean~\cite{curticapean2024count}. In a nutshell, Curticapean's approach is to use so-called ``CFI-graphs''\footnote{Named after the seminal Cai-Fürer-Immerman construction for bounding the expressivity of the Weisfeiler-Leman algorithm \cite{cai1992optimal}.} to filter out \emph{in polynomial time} certain hard terms in the homomorphism basis of counting problems. 
For the purpose of our work, we present Curticapean's result as a black box tool for constructing polynomial-time Turing reductions, and we refer the interested reader to Curticapean's paper~\cite{curticapean2024count} for the details.

For the statement of the framework, we need to introduce a few further concepts and notations.

A graph $H$ is colourful, if its vertex-colouring $\nu_H$ is a bijection; note that we may identify $\nu_H(V(H))$ with $V(H)$. Given a colourful graph $H$ and a vertex-coloured graph $G$, the vertices of which are coloured by some colouring $\nu_G : V(G) \to V(H)$, we write $\mathsf{ColHom}(H \to G)$ for the set $\homs{(H, \nu_H)}{(G, \nu_G)}$.

\begin{definition}
Let $\mathcal{H}$ be a recursively enumerable class of colourful graphs. We write $\#\textup{\textsc{ColHom}}(\mathcal{H})$ for the problem that takes as input a colourful graph $H \in \mathcal{H}$ and a vertex-coloured graph $G$, and computes $\#\mathsf{ColHom}(H \to G)$.    
\end{definition}

\begin{remark} While we choose to adapt to the notation of \cite{curticapean2024count} for the reader's easier access to the results therein, we would like to point out that another common way of defining $\mathsf{ColHom}(H \to G)$ is via \emph{colour-prescribed} homomorphisms. Given two graphs $H, G$ and a vertex-colouring $h \in \homs{G}{H}$ of $G$, a colour-prescribed homomorphism $\phi$ from $H$ to $G$ satisfies $\phi \in \homs{H}{G}$ and for each $v \in V(H)$, $h(\phi(v)) = v$. The difference in our case is that the vertex-colouring $\nu_G$ of $G$ might not be a homomorphism in $\homs{G}{H}$. However, it is easy to verify, that for any edge $e \in E(G)$ that is not mapped by $\nu_G$ to an edge of $E(H)$, no homomorphism in $\mathsf{ColHom}(H \to G)$ maps edges of $E(H)$ to $e$ either, since $H$ is colourful. Hence, $\mathsf{ColHom}(H \to G) = \mathsf{ColHom}(H \to (G \setminus e))$. Letting $G'$ denote the graph obtained from $G$ by removing all edges that are not mapped by $\phi$ to an edge in $E(H)$, we have $\mathsf{ColHom}(H \to G) = \mathsf{ColHom}(H \to G')$. Also, note that $\nu_{G'} = \nu_G \in \homs{G'}{H}$. So, the problem of evaluating $\#\mathsf{ColHom}(H \to G)$ reduces to the problem of evaluating $\#\mathsf{cp}$-$\mathsf{Hom}(H \to G')$, where the last quantity counts the number of colour-prescribed homomorphisms from $H$ to $G'$.
\end{remark}

We are now able to state Curticapean's polynomial-time reduction for isolating constituents of a graph motif parameter. Recall that for a graph motif parameter $p$ and a graph $F$, we write $F \triangleleft p$ to denote that the coefficient of $F$ in the expansion of $p$ in the homomorphism basis, is non-zero. 

\begin{theorem}[\cite{curticapean2024count}]\label{thm:raduPoly}
There is an algorithm for the following problem: output the value $\#\mathsf{ColHom}(H \to G)$ when given
\begin{itemize}
\item[$\bullet$] as input a number $a \in \mathbb{N}$, a colorful graph $H$, a colored graph $G$, and
\item[$\bullet$] oracle access for a graph motif parameter $p$ with $\max_{F\,\triangleleft\,p}|V(F)| \leq a$ such that $H^\circ\,\triangleleft\, p$, where $H^\circ$ is the underlying (uncoloured) graph corresponding to $H$.
\end{itemize}
The running time of the algorithm is bounded by $4^{\Delta(H)}\cdot\mathsf{poly}(|V(G)|, a).$
\end{theorem}

For applying the previous result, we first need to state some properties of the expression of Holants as graph motif parameters established in~\cite{aivasiliotis2024parameterisedholantproblems}. Even before we do so, we need to introduce some more preliminaries. In order to avoid notational clutter, we keep the notation we need to introduce at a minimum, skipping details whenever this does not hurt the accessibility of our proofs. We refer the interested reader to the proof of \cite[Theorem 6.12]{aivasiliotis2024parameterisedholantproblems} for a complete exposition of the notions discussed below.

\subsubsection{Integer Partitions and Generalised Fingerprints}
\begin{definition}[Integer partitions \cite{aivasiliotis2024parameterisedholantproblems}]
For positive integers $d, d_1, d_2, \dots, d_{\ell} \in \mathbb{N}$ such that $\sum_{i = 1}^{\ell}d_i = d$, we say that $\lambda = d_1 + d_2 + \ldots + d_{\ell}$ is a \emph{partition} of $d$. We write $|\lambda|$ for the sum of the summands of $\lambda$ (that is, $|\lambda| = d$) and we write $\mathsf{len}(\lambda)$ for the number of summands of $\lambda$ (that is, $\mathsf{len}(\lambda) = \ell$). We write $\mathcal{P}$ for the set of all partitions. For any finite collection $\{\lambda_i\}_{1 \leq i \leq k} \in \mathcal{P}^k$ of partitions (of possibly different positive integers), we let $\bigcup_{1 \leq i \leq k}\lambda_i$ denote the partition whose summands are the elements of the disjoint union of the summands of each $\lambda_i, 1 \leq i \leq k$. 
\end{definition}

An integer partition $\lambda$ is associated with what was defined in \cite{aivasiliotis2024parameterisedholantproblems} as the \textit{multiplicity} of $\lambda$, denoted by $\mathsf{mult}(\lambda) \in \mathbb{Z}$, which we will be using as a black-box. In particular, for the purposes of this work, it is only the sign of $\mathsf{mult}(\lambda)$ that matters, which is given as follows: for any integer partition $\lambda$, $\mathsf{sign}(\mathsf{mult}(\lambda)) = (-1)^{|\lambda|-\mathsf{len}(\lambda)}$. 

The last ingredient is a generalisation of fingerprints from \Cref{def:fingerprint_intro} that lifts the $\chi$ function which has been defined over a positive integer $d$ and a signature $s$  to the generalised $\chi$ function defined over  a partition $\lambda \in \mathcal{P}$ and a signature $s$ instead. While we refer the reader to \cite{aivasiliotis2024parameterisedholantproblems} for a formal definition, the needs of our analysis only require the following property of the generalised fingerprints.

\begin{proposition}[\cite{aivasiliotis2024parameterisedholantproblems}] \label{Prop:chi_vanishing_property}
Let $\{s\}$ be of type $\mathbb{T}[2]$, then
we have
\begin{enumerate}
    \item[a)] $\chi(\lambda, s) = 0$ if $|\lambda| < 2\cdot  \mathsf{len}(\lambda)-2$,
    \item[b)] $\chi(\lambda, s) = a_\lambda \cdot (s(2)-s(1)^2)^{\mathsf{len}(\lambda)-1}$ for some $a_\lambda \in \mathbb{Z}_{>0}$ if $|\lambda| = 2 \cdot \mathsf{len}(\lambda)-2$.
\end{enumerate}
\end{proposition}

Now we are ready to state the main results from \cite{aivasiliotis2024parameterisedholantproblems} that we need. Note that, in what follows we make use of both definitions of fingerprints. Recall from \Cref{lem:uncolHomBasisRESTATED}, that for any signature $s$, and any instance $\Omega = (G, \{s\}_{v \in V(G)})$ and $k \in \mathbb{N}$ of $\uni{\constuncolholantprob{2}}(\{s\})$, we have $\uncolholant(\Omega, k) = \sum_{F \in \mathcal{G}_{\leq k}(2)}\zeta_{k, s}(F)\cdot\Hom(F \to G)$.

\begin{theorem}[\cite{aivasiliotis2024parameterisedholantproblems}]\label{thm:infinityCoeffRestated}
Let $s$ be a signature with $s(0) = 1$, $k \in \mathbb{N}$ and $F \in \mathcal{G}_k$. Then we have
\[\zeta_{k, s}(F) = \frac1{\Aut(F)}\prod_{v \in V(F)}\chi(\deg_F(v), s)\,.\]
\end{theorem}

\begin{theorem}[\cite{aivasiliotis2024parameterisedholantproblems}]\label{thm:zeta_at_most_k_edges}
Let $s$ be a signature with $s(0) = 1$, $k \in \mathbb{N}$, and $F \in \mathcal{G}_{\leq k}$. Then we have
\begin{equation} \label{eqn:Theorem_general1}
    \zeta_{k, s}(F)=\frac{1}{\#\mathsf{Aut}(F)} \cdot \sum_{\substack{\lambda: E(F) \to \mathcal{P}\\\sum_e |\lambda(e)| = k }}\ \  \prod_{e \in E(F)} \frac{\mathsf{mult}(\lambda(e))}{|\lambda(e)|!}  \prod_{v \in V(F)} \chi(\mathsf{deg}(F, v, \lambda), s)\,,
\end{equation}
where the partitions $\mathsf{deg}(F, v, \lambda)$ are defined as
\[
\mathsf{deg}(F, v, \lambda) = \bigcup_{\substack{e \in E(F):\\ e \text{ incident to }v}} \lambda(e)\,. 
\]
\end{theorem}

Next we need to establish the survival of regular graphs as constituents in the graph motif parameter expression of our holant problem --- this is necessary since the running time of Curticapean's reduction has exponential dependency on the maximum degree of the graphs we wish to isolate (see the factor of $4^{\Delta(H)}$ in Theorem~\ref{thm:raduPoly}). We will see later that proving the survival of regular graphs will be easy for signatures of type $\mathbb{T}[\infty]$. For signatures of type $\mathbb{T}[2]$, on the other hand, the situation is slightly more tricky. However, we are able to show that $4$-regular graphs survive in the subsequent lemma.
\begin{lemma}\label{lem:non-zero-coeff} Let $\{s\}$ be of type $\mathbb{T}[2]$. Let $F$ be a 4-regular graph on $n_F$ vertices and $m_F$ edges with no isolated vertices. We set $k_F = 3\cdot n_F$. We have, $\zeta_{k_F, s}(F) \neq 0$.    
\end{lemma}
\begin{proof}
By the Handshaking Lemma, we have $\sum_{v \in V(F)}\deg_F(v) = 2\cdot m_F$. Furthermore, since $F$ is 4-regular (and has no isolated vertices), we have $\sum_{v \in V(F)}\deg_F(v) = 4\cdot n_F$. Hence, $m_F = 2\cdot n_F$ and so $k_F > m_F$. Let $\lambda : E(F) \to \mathcal{P}$ such that $\sum_{e \in E(F)}|\lambda(e)| = k_F$. We have
\[\sum_{v \in V(F)}|\deg(F, v, \lambda)| = 2\sum_{e \in E(F)}|\lambda(e)| = 6\cdot n_F\,,\]
and so the average weight $\sum_{v}|\deg(F, v, \lambda)|/n_F$ is 6. Clearly, if there is a vertex $u$ with $|\deg(F, u, \lambda)| > 6$, then there is also a vertex $v$ with $|\deg(F, v, \lambda)| < 6$. In that case, we have $|\deg(F, v, \lambda)| < 6 = 2\cdot 4 -2 \leq 2\cdot\mathsf{len}(\deg(F, v, \lambda)) -2$ and so by \Cref{Prop:chi_vanishing_property}, we have $\chi(\deg(F, v, \lambda), s) = 0$. Note that, we used $\mathsf{len}(\deg(F, v, \lambda)) \geq \deg_F(v) = 4$ which follows directly from the definition of $\deg(F, v, \lambda)$. Hence, we may consider only such $\lambda$ for which it holds that, $|\deg(F, v, \lambda)| = 6$, for each $v \in V(F)$. Furthermore, if for some $v \in V(F)$, we have $\mathsf{len}(\deg(F, v, \lambda)) > 4$, then similar arguments as above show that $\chi(\deg(F, v, \lambda), s) = 0$. Since $\mathsf{len}(\deg(F, v, \lambda)) \geq 4$, it follows that, we may further restrict our focus to those $\lambda$ such that, $\mathsf{len}(\deg(F, v, \lambda)) = 4$, for each $v \in V(F)$, which is equivalent to requiring that $\mathsf{len}(\lambda(e)) = 1$, for each $e \in E(F)$, since $F$ is 4-regular.

For any such $\lambda$, we compute
\[
\mathsf{sign}\left(\prod_{e \in E(F)}\mathsf{mult}(\lambda(e))\right) = \prod_{e \in E(F)}(-1)^{|\lambda(e)|-1} = (-1)^{k_F-m_F} = (-1)^{n_F}\,.
\]
Also, for each $v \in V(F)$, we have $\chi(\deg(F, v, \lambda), s) = a_{\deg(F, v, \lambda)}(s(2) - s(1)^2)^3$, where $a_{\deg(F, v, \lambda)} \in \mathbb{Z}_{> 0}$ depends only on $\deg(F, v, \lambda)$.

The above analysis show that no terms of $\zeta_{k_F, s}(F)$ cancel out (since all terms have same sign). Hence, what remains is to show that there is at least one $\lambda$ such that for each $e \in E(G), \mathsf{len}(\lambda(e)) = 1$ and for each $v \in V(F), |\deg(F, v, \lambda)| = 6$ (note that the above implies that we also have $\sum_{e \in E(F)}|\lambda(e)| = 1/2\sum_{v \in V(F)}|\deg(F, v, \lambda)| = 3\cdot n_F = k_F$). 

Let $F'$ be a 2-regular spanning subgraph of $F$, the existence of which follows from Petersen's 2-factor theorem \cite[Theorem 2]{mulder1992julius}. We define $\lambda$ as follows.
\[
\lambda(e) = 
\begin{cases}
    2, & e \in E(F')\\
    1, & e \in E(F) \setminus E(F')
\end{cases}
\]
It is straightforward to verify that for each $v \in V(F)$, $|\deg(F, v, \lambda)| = 6$, which concludes the proof.
\end{proof}

Our final ingredient for the application of Theorem~\ref{thm:raduPoly} is the $\#\mathrm{P}$-hardness $\#\text{\sc{ColHom}}(\mathcal{R}_r)$, for every $r\geq 3$ where $\mathcal{R}_r$ denotes the class of all colourful $r$-regular graphs. We point out that $\#\mathrm{P}$-hardness under \emph{randomised} polynomial-time reductions follows immediately from Curticapean's paper~\cite{curticapean2024count}, since $\mathcal{R}_r$ has unbounded treewidth. The reason for the need of randomisation is the necessity of a polynomial-time algorithm that finds a minor-model of a large enough grid (or wall) in a graph of high treewidth. Specifically, it is necessary to obtain a polynomial dependence between the grid (or wall) minor and the treewidth for the overall reduction to run in polynomial time.  Unfortunately, at the time of writing this paper, there is no \emph{deterministic} algorithm known for finding such grid or wall minors.  However, for our specific case of $r$-regular graphs we can make the reduction deterministic by providing grid-minors explicitly. To this end, we will need to take a small detour to grids, walls, and minors.

\subsubsection{Sub-walls and grid minors in regular graphs}

We write $\boxplus_{k\times \ell}$ for the $k$-by-$\ell$ grid. The $k$\emph{-by-}$\ell$\emph{-wall}, denoted by $W_{k \times \ell}$ is obtained from $\boxplus_{k \times \ell}$ via the following three steps:
\begin{enumerate} 
    \item Delete every odd-indexed edge in every even-indexed-column.
    \item Delete every even-indexed edge in every odd-indexed column.
    %\item Contract every vertex of degree $2$ which does not lie on the outer face.\todo{I think we need to remove that. For our toroidal wall construction we should not have contracted vertices in the wall to be transformed into a toroidal wall.}
\end{enumerate}
Finally, for odd $k$ and odd $\ell$, the $k$\emph{-by-}$\ell$-\emph{toroidal wall}, denoted by $\circledcirc_{k\times\ell}$ is obtained from the wall $W_{k\times \ell}$ by identifying the left and the right sides, and by identifying the top and the bottom sides.\footnote{Note that $k$ must be odd and $\ell$ must be even for the left and right sides, and the top and bottom sides to ``fit'' in order to produce a $3$-regular graph.} Formally, the toroidal wall of dimension $k \times \ell$ can also be defined from the Cayley graph of the direct product $\mathbb{Z}_k\times \mathbb{Z}_\ell$ with generators $\pm(0,1),\pm(1,0)$ (which yields the toroidal grid), and then deleting edges and contracting degree-$2$-vertices similar as in the construction of the wall from the grid. However, we will only rely on very basic properties of the toroidal wall, hence we omit the technical details and invite the reader to instead consider Figure~\ref{fig:donutwall}.

\begin{figure}[t]
    \centering
    \includegraphics[width=\linewidth]{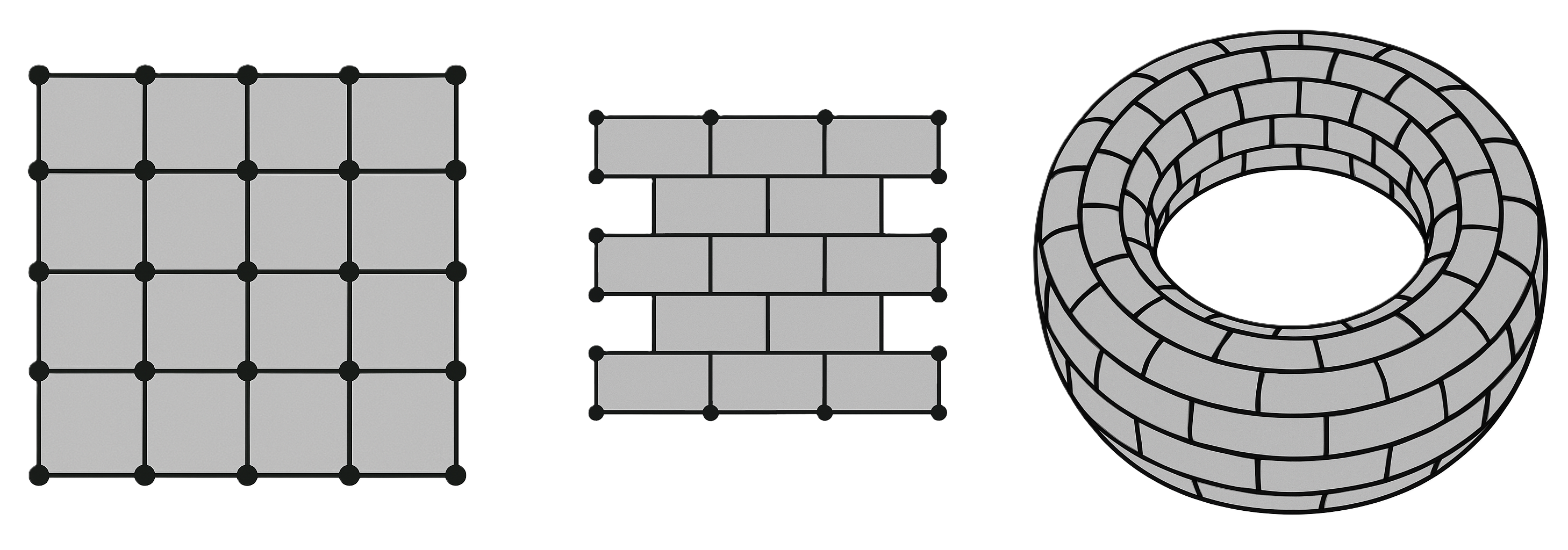}
    \caption{\emph{(Left:)} The $5$-by-$5$-grid $\boxplus_{5\times 5}$. \emph{(Center:)} The $6$-by-$7$-wall $W_{6\times 7}$. \emph{(Right:)} Illustration of a toroidal wall. This illustration was generated by ChatGPT4o.}
    \label{fig:donutwall}
\end{figure}

For the next reduction, we need an explicit construction of $r$-regular graphs together with large wall-subgraphs.

\begin{lemma}
    For every positive integer $r\geq 3$, there is a deterministic polynomial time algorithm that, on input a natural number $t$ in unary, outputs a pair of graphs $(H,W)$ such that
    \begin{itemize}
        \item $H$ is an $r$-regular graph.
        \item $W\cong W_{t\times t}$ and $W$ is a subgraph of $H$, that is, $V(W)\subseteq V(H)$.
    \end{itemize}
\end{lemma}
\begin{proof}
    Assume w.l.o.g.\ that $t\geq2r$; otherwise we perform the construction for $t'=2r$.
    If $r=3$ then we set (and construct) $H$ to be the toroidal wall of dimensions $t \times t$ if $t$ is odd, and $(t+1)\times (t+1)$ if $t$ is even. We pick for $W$ any of the $t \times t$ sub-walls of the toroidal wall.

    If $r>3$, we construct $2r$ copies $H_1,\dots,H_{2r}$ of the toroidal wall $\circledcirc_{t\times t}$ (if $t$ is odd) or of $\circledcirc_{(t+1)\times(t+1)}$ (if $t$ is even). For $i\in[2r]$, we let $v_i$ denote the $i$-th copy of vertex $v$ in the toroidal wall. Since $2r(r-3)$ is even and $2r\geq (r-3)+1$, there is an $(r-3)$-regular graph $F$ with $V(F)=[2r]$.\footnote{Recall that for even $dn$, and $n\geq d+1$ there always exists a $d$-regular $n$-vertex graph.}  For each vertex $v$ of the toroidal wall, we add a copy of $F$ by identifying $v_i$ with vertex $i\in V(F)$ for all $i\in [2k]$. The resulting graph is set to be $H$. Clearly, $H$ is $3+(r-3)=r$-regular. Moreover, we can pick for $W$ any of the $t \times t$ sub-walls of the first copy $H_1$ of the toroidal wall.
\end{proof}

We will furthermore rely on the following hardness result as an intermediate step for our overall reduction. We remark that this result follows easily from combining the previous lemma with the analysis of $\#\mathrm{P}$-hardness of colourful homomorphism counting in~\cite{curticapean2024count}; we only add an explicit proof for reasons of self-containment.
\begin{lemma}
    Let $r>2$ be a positive integer and let $\mathcal{R}_r$ denote the class of all colorful $r$-regular graphs. Then the problem $\#\text{\sc{ColHom}}(\mathcal{R}_r)$ is $\#\mathrm{P}$-hard.
\end{lemma}
\begin{proof}
    Let $\boxplus$ be the class of all square grids.
    As shown in~\cite[Theorem 2.2]{curticapean2024count}, $\#\text{\sc{ColHom}}(\boxplus)$ is $\#\mathrm{P}$-hard. We construct a polynomial-time Turing reduction from $\#\text{\sc{ColHom}}(\boxplus)$ as follows: given a colourful $t\times t$ square grid $\boxplus_{t\times t}$ and a (coloured) graph $G$, the goal is to compute the number of colour-preserving homomorphisms from $\boxplus_{t\times t}$ to $G$. 
    We use the algorithm in the previous lemma to construct in polynomial time in $t$ a pair $(H,W)$ such that $H$ is $r$-regular and $W$ is a $t\times t$ sub-wall of $H$. Observe that, given $W$, it is easy to construct a minor-model of $\boxplus_{t\times t}$ in $H$ in polynomial time. Then, by~\cite[Lemma 2.1]{curticapean2024count}, we can in polynomial time construct a coloured graph $G'$ such that $\#\homs{\boxplus_{t\times t}}{G}=\#\homs{H}{G'}$ (where $H$ is considered vertex-coloured with the identity function on $V(H)$). Since $H\in \mathcal{R}_r$, the proof is concluded.
\end{proof}

For reasons of compatibility with the notation of \Cref{thm:raduPoly}, we choose to equivalently see $\uncolholant(\Omega, k)$ as a map \[p_{k, s} : G \mapsto \sum_{F \in \mathcal{G}_{\leq k}(2)}\zeta_{k, s}(F)\cdot \Hom(F \to G)\,,\]
where $G$ is a simple \emph{graph}. Clearly, for any instance of $\uni{\constuncolholantprob{2}}(\{s\})$ consisted of a signature grid $\Omega$ with underlying graph $G$ and $k \in \mathbb{N}$, we have $p_{k, s}(G) = \uncolholant(\Omega, k)$.
Recall that we write $F \triangleleft\,p_{k, s}$ whenever $\zeta_{k, s}(F) \neq 0$.

\begin{theorem}\label{thm:hardnessInfinity}
If $\{s\}$ is of type $\mathbb{T}[\infty]$, $\classHolant(\{s\})$ is $\#\mathrm{P}$-hard.    
\end{theorem}
\begin{proof}
Since $\{s\}$ is of type $\mathbb{T}[\infty]$, there exists $r \in \mathbb{Z}_{> 2}$, such that $\chi(r, s) \neq 0$ by \Cref{def:fingerprint_intro}. Also, there is a algorithm to compute such $r$ (that depends only on $s$).

Let $\mathcal{R}_r$ denote the class of all colorful $r$-regular graphs. Since $r>2$, we have that $\#\text{\sc{ColHom}}(\mathcal{R}_r)$ is $\#\mathrm{P}$-hard by the previous lemma.

We show the following reduction
\[\#\text{\sc{ColHom}}(\mathcal{R}_r) \leq_{\mathsf{T}} \classHolant(\{s\})\,.\]
To this end, let $H \in \mathcal{R}_r$ and $G$ be an instance of $\#\textsc{ColHom}(\mathcal{R}_r)$. We set $k = |E(H)|$ and observe that from \Cref{thm:infinityCoeffRestated}, it follows that $\zeta_{k, s}(H^\circ) = \Aut(H^{\circ})^{-1}\cdot\chi(r,s)^{|V(H^\circ)|} \neq 0$. 

Hence, the conditions for the application of \Cref{thm:raduPoly} are met and thus we can compute $\#\mathsf{ColHom}(H \to G)$ with oracle access to $p_{k, s}(\cdot)$, in time $4^r\cdot \mathsf{poly}(|V(G)|, a)$, where $a$ is the maximum number of vertices appearing in $\mathcal{G}_{\leq k}$ which is by definition $2k$. Since $k = |E(H)|$ and $r$ is a constant, we deduce that the reduction requires $\mathsf{poly}(|V(G)|, |V(H)|)$ time. This concludes the reduction and hence the claim is proved.
\end{proof}

\begin{theorem}\label{thm:hardnessOmega}
If $\{s\}$ is of type $\mathbb{T}[2]$, $\classHolant(\{s\})$ is $\#\mathrm{P}$-hard.    
\end{theorem}

\begin{proof}
Let $\mathcal{R}_4$ denote the class of all colorful $4$-regular graphs. Similarly to the proof of \Cref{thm:hardnessInfinity} we show the following reduction
\[\#\text{\sc{ColHom}}(\mathcal{R}_4) \leq_{\mathsf{T}} \classHolant(\{s\})\,.\]

To this end, let $H \in \mathcal{R}_4$ and $G$ be an instance of $\#\textsc{ColHom}(\mathcal{R}_4)$. We set $k = 3|V(H)|$ and observe that from \Cref{lem:non-zero-coeff} we have $\zeta_{k, s}(H^{\circ}) \neq 0$. Hence, the conditions for the application of \Cref{thm:raduPoly} are met and thus we can compute $\#\mathsf{ColHom}(H \to G)$ with oracle access to $p_{k, s}(\cdot)$ in time $\mathsf{poly}(|V(G)|, |V(H)|)$.
\end{proof}

\subsection{Hardness for signatures of types $\mathbb{T}[2]$ or $\mathbb{T}[\infty]$ on hypergraphs}\label{sec:HardnessCFIHypergraphs}

\begin{lemma}\label{lem:connectedHypergraphs}
For any $d, b \in \mathbb{Z}_{\geq 2}$, there are infinitely many $d$-uniform $b$-regular connected hypergraphs.    
\end{lemma}
\begin{proof}
For any $d \geq 2$, we prove by induction on $b$, that there is at least one $d$-uniform $b$-regular connected hypergraph. In particular, for any $b \geq 2$, instead of merely showing the existence of one such hypergraph, we construct an infinite sequence of $d$-uniform $b$-regular connected hypergraphs.

For the induction base, we note that, for $b = 1$ and $d \geq 2$, there is a unique $d$-uniform $1$-regular connected hypergraph, which is the single hyperedge (of rank $d$). 

For the induction step, let $b \geq 2$ and let $F_0$ be a $d$-uniform $(b-1)$-regular connected hypergraph. We define an infinite sequence $F_1, F_2, \dots$ of $d$-uniform $b$-regular connected hypergraphs (of increasing size) as follows.

For $i \geq 2$, let $C^1, C^2, \dots, C^{i\cdot d}$ be $i\cdot d$ copies of $F_0$. Since $|V(\bigcup_{j = 1}^{i\cdot d}C^j)|$ is a multiple of $d$ we can partition the vertices of all copies $C^j$ into sets of size $d$. Then, we use those sets to form new hyperedges that increase the degree of each vertex by one. So, what remains is to ensure that we form hyperedges in such a way that the resulting hypergraph is connected. For this, let $x, y \in \mathbb{Z}_{\geq 1}$ such that $x + y = d$. For each $1 \leq j < i \cdot d$, we form a hyperedge containing $x$ vertices from $C^j$ and $y$ vertices from $C^{j+1}$, in such a way that no vertex of a copy belongs to more than one newly formed hyperedges. The above is feasible since, for each copy, at most $d$ of its vertices appear in the new hyperedges (and each copy contains at least $d$ vertices), thus, yielding a path intersecting with all copies. Finally, the number of vertices that do not appear in any of the new hyperedges is still a multiple of $d$, and so we partition them into hyperedges of size $d$ arbitrarily. It is straightforward to verify that the resulting hypergraph is indeed connected and $b$-regular.
\end{proof}

\begin{theorem}\label{thm:SharpPHardHyper}
If $\{s\}$ is of type $\mathbb{T}[2]$ or $\mathbb{T}[\infty]$, $\constclassHolant{d}(\{s\})$ is $\#\mathrm{P}$-hard, for any number $d \in \mathbb{Z}_{> 2}$. 
\end{theorem}
\begin{proof}
Let $b > 0$ such that $s(b) \neq 0$ and for each $0 < i < b$, $s(i) = 0$. We distinguish between the cases $b = 1$, $b = 2$ and $b > 2$. For the first two cases, we observe that the reduction
\[\constclassHolant{2}(\{s\}) \leq_{\mathsf{T}} \constclassHolant{d}(\{s\})\,,\]
follows directly from the two first cases of \Cref{thm:reduction2TodArity} respectively. This is true, since the parameterised Turing reductions appearing in both of the aforementioned cases, are in fact polynomial time Turing reductions. Furthermore, since $\constclassHolant{2}(\{s\})$ is $\#\mathrm{P}$-hard whenever $\{s\}$ is of type $\mathbb{T}[2]$ (resp. $\mathbb{T}[\infty]$) as shown in \Cref{thm:hardnessOmega} (resp. \Cref{thm:hardnessInfinity}), it follows that $\constclassHolant{d}(\{s\})$ is $\#\mathrm{P}$-hard.

For the remaining case, we follow the lines of the proof of \Cref{thm:HyperPerfectMatchingReduction}. However, we present the complete proof as it requires non-trivial modifications. We show the following reduction
\[\#\textsc{PerfectMatching}^{d} \leq_{\mathsf{T}} \constclassHolant{d}(\{s\})\,.\]

To this end, let $G$ be an instance of $\#\textsc{PerfectMatching}^{d}$ and assume that $n = |V(G)|$ is a multiple of $d$ (since otherwise, $\#\textsf{PerfMatch}(G) = 0)$. Let $F(d, b-1)$ be a $d$-uniform $(b-1)$-regular hypergraph, with the additional requirement that $F(d, b-1)$ is connected. By \Cref{lem:connectedHypergraphs}, there are infinitely many choices $F(d, b-1)$ for any $d$ and $b \geq 3$. Hence, we may also assume that $F(d, b-1)$ has at least $(d-1)(b-1)$ vertices as well as that $|E(F(d, b-1))| > |E(G)|$. We may also construct $F(d, b-1)$ such that $E(F(d, b-1)) \in O(|E(G)|)$, which can be done in polynomial time.

To see this note that, by the construction shown in \Cref{lem:connectedHypergraphs} the number of vertices of $F(d, b-1)$ can be expressed as $j \cdot c$, for $j \in \mathbb{N}$ divisible by $d$ and some constant $c$. It is easy to verify that for $d$-uniform $(b-1)$-regular hypergraphs (with no isolated vertices) we have $|V(F(d,b-1))|\cdot (b-1) = d \cdot |E(F(d,b-1))|$. Hence, it follows that $j\cdot c \cdot (b-1) = d\cdot |E(F(d, b-1))|$ and so $|E(F(d,b-1))| = j \cdot(c\cdot(b-1)/d)$. Note that the term $c\cdot(b-1)/d$ is a constant and thus $|E(F(d,b-1))|$ depends only on $j$. Hence, by simply picking $j = d\cdot |E(G)|$, we have $|E(F(d, b-1))| > |E(G)|$ and $|E(F(d,b-1))| \in O(|E(G)|)$.

We construct a hypergraph $G'$ obtained from $G$ as follows. For each $v \in V(G)$, we add a copy of $F(d, b-1)$, denoted as $F(d, b-1;v)$. We pick $(d-1)(b-1)$ vertices of $F(d, b-1;v)$ arbitrarily, which we partition into $b-1$ sets of size $d-1$, namely $S_1, \dots, S_{b-1}$. For each $1 \leq i \leq b-1$, we create a new hyperedge $S_i \cup \{v\}$.

Next, we proceed by increasing by one the degree of each vertex of $F(d, b-1; v)$ having degree $b-1$ (that is, of each vertex of $F(d, b-1;v)$ that is not adjacent to $v$). We will perform this step as follows. Recalling that $n$ is a multiple of $d$, we group $V(G)$ into pairwise disjoint sets of size $d$. Let $v_1, \dots, v_{d}$ be such a group. We partition the vertices of $\bigcup_{i =1}^{d}F(d, b-1;v_i)$ having degree $b-1$ into pair-wise disjoint sets of size $d$, which we add as new hyperedges to $G'$. We repeat the aforementioned procedure for all groups and thus we obtain $G'$. Note, that since $|F| \in O(|G|)$, we have $|G'| = |G|^{O(1)}$.

\begin{figure}
    \centering
    \includegraphics[width=\linewidth]{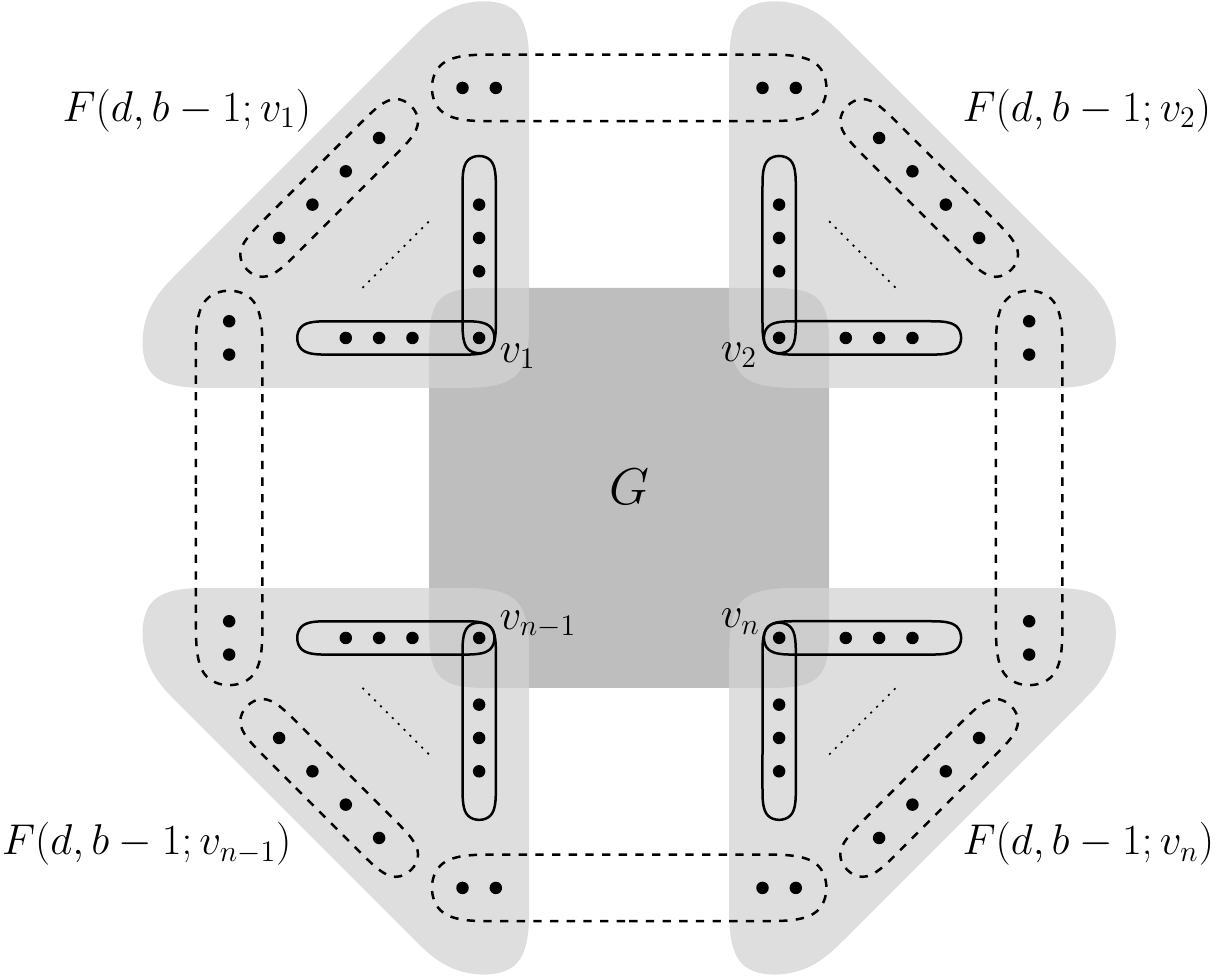}
    \caption{Illustration of the construction of $G'$ in the proof of Theorem~\ref{thm:SharpPHardHyper} for the case $d=4$. The solid hyperedges in each gadget represent the sets $S_1\cup\{v\},\dots,S_{b-1}\cup\{v\}$. The dashed hyperedges represent the hyperedges added in the last step in order to enforce $G'$ to be $b$-regular.}
    \label{fig:gadgetreduction}
\end{figure}

Let $\Omega'$ denote the signature grid with underlying hypergraph $G'$ the vertices of which are equipped with signature $s$. Recall that $s(0) \neq 0$. Observe, that every vertex $v \in V(G') \setminus V(G)$ has degree $b$. 

We also observe that for any $k$-set $A \subseteq E(G')$, if there is $v \in V(G') \setminus V(G)$ such that $0 < |A \cap E_{G'}(v)| < |E_{G'}(v)| = b$, then the contribution of $A$ to $\uncolholant(\Omega', k)$ (for any $k$) is trivially zero, since then $s(|A \cap E_{G'}(v)|) = 0$. In other words, we may assume that we only consider $k$-sets $A$, such that, for each $v \in V(G)$, $A$ contains either all or none of the hyperedges incident to vertices of $F(d, b-1;v)$. Following the lines of \Cref{thm:HyperPerfectMatchingReduction} we want to ensure that every $k$-set $A$ with non-zero contribution to $\uncolholant(\Omega', k)$ must contain all hyperedges in $E(G') \setminus E(G)$.\footnote{In the proof of \Cref{thm:HyperPerfectMatchingReduction} this was easily achieved using that $s(0) = 0$.} To achieve this, it suffices to take $F(d, b-1)$ such that $|E(F(d, b-1))| > |E(G)|$ along with $k \geq |E(G') \setminus E(G)|$. To see this, assume for contradiction that a $k$-set misses one gadget completely. Then even if all other hyperedges are taken, the total number of hyperedges is at most $|E(G)|+ |E(G')\setminus E(G)| - |E(F(d,b-1))|< |E(G')\setminus E(G)| \leq k$, which is a contradiction.

Hence, we can assume that, for $k \geq |E(G') \setminus E(G)|$, all edges in $E(G')\setminus E(G)$ are taken. In particular, for $k = n/d + |E(G') \setminus E(G)|$, since every (non-trivial) $k$-set $A \subseteq E(G')$ contains all hyperedges in $E(G') \setminus E(G)$, it follows that the remaining $n/d$ hyperedges should consist a hyperedge-cover of $G$. To see this, not that every vertex $v \in V(G)$ has exactly $b-1$ incident hyperedges intersecting with $G' \setminus G$ and $s(b-1) = 0$. As it has already been mentioned, a hyperedge-cover of size $n/d$ of an $n$-vertex $d$-uniform hypergraph is a perfect matching. So, the contribution of $A$ to $\uncolholant(\Omega', k)$ is precisely $\prod_{v \in V(G')}s(|A \cap E_{G'}(v)|) = s(b)^{|V(G')|}$. Hence, we have
\[
\uncolholant(\Omega', k) = s(b)^{|V(G')|}\#\mathsf{PerfMatch}(G)\,.
\]

Clearly, $s(b)^{|V(G')|}$ can be computed in time that is polynomial in $|G|$ by using standard techniques, which concludes the proof.

\end{proof}

\section{\ensuremath{\mathrm{W}[2]}-hardness for \ensuremath{\uncolunbndholantprob}}\label{sec:W2hardness}

We will show, in a very easy reduction, that $\uncolunbndholantprob(\mathcal{S})$ is $\mathsf{W}[2]$-hard even for the signature $\mathsf{hw}_{\geq 1}$ of having Hamming weight at least $1$ --- note that for this signature, $\uncolholantprob(\mathcal{S})$ is fixed-parameter tractable (see \Cref{thm:mainParamThm}).

\begin{lemma}
$\uncolunbndholantprob(\{\mathsf{hw}_{\geq 1}\})$ is $\mathsf{W}[2]$-hard.
\end{lemma}
\begin{proof}
    We reduce from the problem $\textsc{p-HittingSet}$ which asks, on input a hypergraph $G$ and a positive integer $k$, to decide whether $G$ contains a $k$-hitting set, that is, a vertex-subset $S$ of size $k$ such that each hyperedge is incident to at least one vertex of $S$. The parameter is $k$. It is well-known that $\textsc{p-HittingSet}$ is a $\mathrm{W}[2]$-complete problem (see e.g., \cite[Theorem 7.14]{FlumG06}).

    To this end, let $(G,k)$ be an instance of $\textsc{p-HittingSet}$. Note that, if there are two vertices $u, v \in V(G)$ with $E_G(v) = E_G(u)$, then we can safely remove one of them (arbitrarily), in the sense that this operation does not affect the existence or not of a $k$-hitting set. To see this, first note that we only remove vertices and not hyperedges. It is also clear that deleting $u$ or $v$ does not remove hyperedges either. Therefore, deletion of $u$ or $v$ cannot lead to forming new $k$-hitting sets. Furthermore, for any $k$-hitting set $X$ of $G$, the following are true:
    \begin{enumerate}
    \item If $\{u, v\} \subseteq X$, then for $z \in \{u, v\}$ and any $w \notin \{u, v\}$, we have that $(X \setminus\{z\}) \cup \{w\}$ is a $k$-hitting set of $G$ (assuming we are not in the degenerate case where $V(G) = \{u, v\}$ and $k = 2$).
    \item If $\{u, v\}\,\cap\,X = \{u\}$, then $(X \backslash \{u\}) \,\cup\,\{v\}$ is a $k$-hitting set of $G$ (equivalently, $(X \backslash \{v\})\,\cup\,\{u\}$ is a $k$-hitting set if $\{u, v\}\,\cap\,X = \{v\}$). 
    \end{enumerate}
    It is also easy to see that we may repeat the deletion of vertices as many times as required (which can be done in polynomial time). Hence, given $G$ we may assume that there are no vertices $u, v \in V(G)$ such that $E_G(u) = E_G(v)$.
    
    To proceed with our reduction, we construct an instance of $\uncolunbndholantprob(\{\mathsf{hw}_{\geq 1}\})$ as follows. Let $G'$ be the hypergraph with vertices $E(G)$; moreover, for each $v\in V(G)$, we add a hyperedge $e'_v=\{e\in E(G)=V(G')\mid v \in e \}$. Note that, $G'$ contains no multi-hyperedges which follows from our previous assumption. Finally, each vertex in $V(G')=E(G)$ is equipped with signature $\mathsf{hw}_{\geq 1}$.

    Then it is easy to see that $\uncolholant(G',k)$ counts precisely the $k$-hitting sets of $G$, showing the reduction, since we can decide whether at least one $k$-hitting set of $G$ exists.
\end{proof}

\bibliographystyle{plain}
\bibliography{biblio}

\begin{thebibliography}{10}

\bibitem{aivasiliotis_et_al:LIPIcs.ICALP.2025.7}
Panagiotis Aivasiliotis, Andreas G\"{o}bel, Marc Roth, and Johannes Schmitt.
\newblock {Parameterised Holant Problems}.
\newblock In {\em ICALP 2025}, volume 334, pages 7:1--7:14, 2025.
\newblock Full version : \url{ https://arxiv.org/abs/2409.13579}.

\bibitem{aivasiliotis2024parameterisedholantproblems}
Panagiotis Aivasiliotis, Andreas Göbel, Marc Roth, and Johannes Schmitt.
\newblock Parameterised holant problems, 2024.

\bibitem{Bollobas88}
B{\'{e}}la Bollob{\'{a}}s.
\newblock The isoperimetric number of random regular graphs.
\newblock {\em Eur. J. Comb.}, 9(3):241--244, 1988.

\bibitem{bressan2025complexitycountingsmallsubhypergraphs}
Marco Bressan, Julian Brinkmann, Holger Dell, Marc Roth, and Philip Wellnitz.
\newblock The complexity of counting small sub-hypergraphs, 2025.

\bibitem{10.1145/3564246.3585204}
Marco Bressan, Matthias Lanzinger, and Marc Roth.
\newblock The complexity of pattern counting in directed graphs, parameterised by the outdegree.
\newblock In {\em Proceedings of the 55th Annual ACM Symposium on Theory of Computing}, STOC 2023, page 542–552, New York, NY, USA, 2023. Association for Computing Machinery.

\bibitem{Bulatov13}
Andrei~A. Bulatov.
\newblock The complexity of the counting constraint satisfaction problem.
\newblock {\em J. {ACM}}, 60(5):34:1--34:41, 2013.

\bibitem{Bulatov17}
Andrei~A. Bulatov.
\newblock A {D}ichotomy {T}heorem for {N}onuniform {CSP}s.
\newblock In {\em Proc. {FOCS} 2017}, pages 319--330, 2017.

\bibitem{BulatovSurvey18}
Andrei~A. Bulatov.
\newblock Constraint satisfaction problems: complexity and algorithms.
\newblock {\em ACM SIGLOG News}, 5(4):4–24, November 2018.

\bibitem{bulatov2014constraint}
Andrei~A Bulatov and D{\'a}niel Marx.
\newblock Constraint satisfaction parameterized by solution size.
\newblock {\em SIAM Journal on Computing}, 43(2):573--616, 2014.

\bibitem{CaiC17}
Jin{-}Yi Cai and Xi~Chen.
\newblock Complexity of counting {CSP} with complex weights.
\newblock {\em J. {ACM}}, 64(3):19:1--19:39, 2017.

\bibitem{cai1992optimal}
Jin-Yi Cai, Martin F{\"u}rer, and Neil Immerman.
\newblock An optimal lower bound on the number of variables for graph identification.
\newblock {\em Combinatorica}, 12(4):389--410, 1992.

\bibitem{CaiG21}
Jin{-}Yi Cai and Artem Govorov.
\newblock The complexity of counting edge colorings for simple graphs.
\newblock {\em Theor. Comput. Sci.}, 889:14--24, 2021.

\bibitem{CaiG22}
Jin{-}Yi Cai and Artem Govorov.
\newblock Perfect matchings, rank of connection tensors and graph homomorphisms.
\newblock {\em Comb. Probab. Comput.}, 31(2):268--303, 2022.

\bibitem{CaiGW16}
Jin{-}Yi Cai, Heng Guo, and Tyson Williams.
\newblock A complete dichotomy rises from the capture of vanishing signatures.
\newblock {\em {SIAM} J. Comput.}, 45(5):1671--1728, 2016.

\bibitem{cai2025holant}
Jin-Yi Cai and Jin~Soo Ihm.
\newblock Holant* dichotomy on domain size 3: A geometric perspective.
\newblock In {\em 52nd International Colloquium on Automata, Languages, and Programming (ICALP 2025)}, pages 148--1. Schloss Dagstuhl--Leibniz-Zentrum f{\"u}r Informatik, 2025.

\bibitem{cai2009holant}
Jin-Yi Cai, Pinyan Lu, and Mingji Xia.
\newblock Holant problems and counting csp.
\newblock In {\em Proceedings of the forty-first annual ACM symposium on Theory of computing}, pages 715--724, 2009.

\bibitem{Chenetal05}
Jianer Chen, Benny Chor, Mike Fellows, Xiuzhen Huang, David~W. Juedes, Iyad~A. Kanj, and Ge~Xia.
\newblock Tight lower bounds for certain parameterized {N}{P}-hard problems.
\newblock {\em Inf. Comput.}, 201(2):216--231, 2005.

\bibitem{Chenetal06}
Jianer Chen, Xiuzhen Huang, Iyad~A. Kanj, and Ge~Xia.
\newblock Strong computational lower bounds via parameterized complexity.
\newblock {\em J. Comput. Syst. Sci.}, 72(8):1346--1367, 2006.

\bibitem{creignou1996complexity}
Nadia Creignou and Miki Hermann.
\newblock Complexity of generalized satisfiability counting problems.
\newblock {\em Information and computation}, 125(1):1--12, 1996.

\bibitem{CreignouV15}
Nadia Creignou and Heribert Vollmer.
\newblock Parameterized complexity of weighted satisfiability problems: Decision, enumeration, counting.
\newblock {\em Fundam. Informaticae}, 136(4):297--316, 2015.

\bibitem{Curticapean13}
Radu Curticapean.
\newblock Counting matchings of size k is \#{W}[1]-hard.
\newblock In {\em Proc.\ of ICALP}, volume 7965, pages 352--363, 2013.

\bibitem{Curticapean15}
Radu Curticapean.
\newblock {\em The simple, little and slow things count: {O}n parameterized counting complexity}.
\newblock PhD thesis, Saarland University, 2015.

\bibitem{curticapean2024count}
Radu Curticapean.
\newblock {Count on CFI graphs for \#P-hardness}.
\newblock In {\em Proceedings of the 2024 Annual ACM-SIAM Symposium on Discrete Algorithms (SODA)}, pages 1854--1871. SIAM, 2024.

\bibitem{CurticapeanDM17}
Radu Curticapean, Holger Dell, and D{\'{a}}niel Marx.
\newblock Homomorphisms are a good basis for counting small subgraphs.
\newblock In {\em Proc.\ of ACM STOC}, pages 210--223, 2017.

\bibitem{CyganFKLMPPS15}
Marek Cygan, Fedor~V. Fomin, Lukasz Kowalik, Daniel Lokshtanov, D{\'{a}}niel Marx, Marcin Pilipczuk, Michal Pilipczuk, and Saket Saurabh.
\newblock {\em Parameterized {A}lgorithms}.
\newblock Springer, 2015.

\bibitem{dalmau2004complexity}
V{\'\i}ctor Dalmau and Peter Jonsson.
\newblock The complexity of counting homomorphisms seen from the other side.
\newblock {\em Theoretical Computer Science}, 329(1-3):315--323, 2004.

\bibitem{DiazST02}
Josep D{\'{\i}}az, Maria~J. Serna, and Dimitrios~M. Thilikos.
\newblock Counting {H}-colorings of partial k-trees.
\newblock {\em Theor. Comput. Sci.}, 281(1-2):291--309, 2002.

\bibitem{DowneyFVW99}
Rodney~G. Downey, Michael~R. Fellows, Alexander Vardy, and Geoff Whittle.
\newblock The parametrized complexity of some fundamental problems in coding theory.
\newblock {\em {SIAM} J. Comput.}, 29(2):545--570, 1999.

\bibitem{DyerR13}
Martin~E. Dyer and David Richerby.
\newblock An effective dichotomy for the counting constraint satisfaction problem.
\newblock {\em {SIAM} J. Comput.}, 42(3):1245--1274, 2013.

\bibitem{FederV98}
Tom{\'{a}}s Feder and Moshe~Y. Vardi.
\newblock The {C}omputational {S}tructure of {M}onotone {M}onadic {SNP} and {C}onstraint {S}atisfaction: {A} {S}tudy through {D}atalog and {G}roup {T}heory.
\newblock {\em {SIAM} J. Comput.}, 28(1):57--104, 1998.

\bibitem{FlumG04}
J{\"{o}}rg Flum and Martin Grohe.
\newblock The {P}arameterized {C}omplexity of {C}ounting {P}roblems.
\newblock {\em {SIAM} J. Comput.}, 33(4):892--922, 2004.

\bibitem{FlumG06}
J{\"{o}}rg Flum and Martin Grohe.
\newblock {\em {P}arameterized {C}omplexity {T}heory}.
\newblock Springer, 2006.

\bibitem{frosini2021new}
Andrea Frosini, Christophe Picouleau, and Simone Rinaldi.
\newblock New sufficient conditions on the degree sequences of uniform hypergraphs.
\newblock {\em Theoretical Computer Science}, 868:97--111, 2021.

\bibitem{Grohe07}
Martin Grohe.
\newblock The complexity of homomorphism and constraint satisfaction problems seen from the other side.
\newblock {\em J. {ACM}}, 54(1):1:1--1:24, 2007.

\bibitem{Grohe&2009expansion}
Martin Grohe and Dániel Marx.
\newblock On tree width, bramble size, and expansion.
\newblock {\em J.\ Comb.\ Theory, Ser.\ B}, 99(1):218 -- 228, 2009.

\bibitem{ImpagliazzoP01}
Russell Impagliazzo and Ramamohan Paturi.
\newblock On the {C}omplexity of k-{S}{A}{T}.
\newblock {\em J. Comput. Syst. Sci.}, 62(2):367--375, 2001.

\bibitem{ImpagliazzoPZ01}
Russell Impagliazzo, Ramamohan Paturi, and Francis Zane.
\newblock Which {P}roblems {H}ave {S}trongly {E}xponential {C}omplexity?
\newblock {\em J. Comput. Syst. Sci.}, 63(4):512--530, 2001.

\bibitem{jerrum2017counting}
Mark Jerrum.
\newblock Counting constraint satisfaction problems.
\newblock 2017.

\bibitem{kolmogorov2017complexity}
Vladimir Kolmogorov, Andrei Krokhin, and Michal Rol{\'\i}nek.
\newblock The complexity of general-valued csps.
\newblock {\em SIAM Journal on Computing}, 46(3):1087--1110, 2017.

\bibitem{Ladner75}
Richard~E. Ladner.
\newblock On the structure of polynomial time reducibility.
\newblock {\em J. {ACM}}, 22(1):155--171, 1975.

\bibitem{LinW17}
Jiabao Lin and Hanpin Wang.
\newblock The complexity of holant problems over boolean domain with non-negative weights.
\newblock In Ioannis Chatzigiannakis, Piotr Indyk, Fabian Kuhn, and Anca Muscholl, editors, {\em 44th International Colloquium on Automata, Languages, and Programming, {ICALP} 2017, July 10-14, 2017, Warsaw, Poland}, volume~80 of {\em LIPIcs}, pages 29:1--29:14. Schloss Dagstuhl - Leibniz-Zentrum f{\"{u}}r Informatik, 2017.

\bibitem{Lovasz12}
L{\'{a}}szl{\'{o}} Lov{\'{a}}sz.
\newblock {\em Large {N}etworks and {G}raph {L}imits}, volume~60 of {\em Colloquium Publications}.
\newblock American Mathematical Society, 2012.

\bibitem{Marx05}
D{\'{a}}niel Marx.
\newblock Parameterized complexity of constraint satisfaction problems.
\newblock {\em Comput. Complex.}, 14(2):153--183, 2005.

\bibitem{Marx10}
D{\'{a}}niel Marx.
\newblock Can {Y}ou {B}eat {T}reewidth?
\newblock {\em Theory Comput.}, 6(1):85--112, 2010.

\bibitem{meng2025fp}
Boning Meng, Juqiu Wang, and Mingji Xia.
\newblock The fp\({}^{\mbox{np}}\) versus \#p dichotomy for \#eo.
\newblock In {\em Proceedings of the 57th Annual ACM Symposium on Theory of Computing}, pages 1795--1806, 2025.

\bibitem{MengWXZ25}
Boning Meng, Juqiu Wang, Mingji Xia, and Jiayi Zheng.
\newblock From an odd arity signature to a holant dichotomy.
\newblock In Srikanth Srinivasan, editor, {\em 40th Computational Complexity Conference, {CCC} 2025, Toronto, Canada, August 5-8, 2025}, volume 339 of {\em LIPIcs}, pages 23:1--23:20. Schloss Dagstuhl - Leibniz-Zentrum f{\"{u}}r Informatik, 2025.

\bibitem{mulder1992julius}
Henry~Martyn Mulder.
\newblock Julius petersen's theory of regular graphs.
\newblock {\em Discrete mathematics}, 100(1-3):157--175, 1992.

\bibitem{Roth19}
Marc Roth.
\newblock {\em Counting {P}roblems on {Q}uantum {G}raphs: {P}arameterized and {E}xact {C}omplexity {C}lassifications}.
\newblock PhD thesis, Saarland University, 2019.

\bibitem{schaefer1978complexity}
Thomas~J Schaefer.
\newblock The complexity of satisfiability problems.
\newblock In {\em Proceedings of the tenth annual ACM symposium on Theory of computing}, pages 216--226, 1978.

\bibitem{Stanley11}
Richard~P. Stanley.
\newblock {\em Enumerative {C}ombinatorics: {V}olume 1}.
\newblock Cambridge University Press, 2011.

\bibitem{valiant1979complexity}
Leslie~G Valiant.
\newblock The complexity of enumeration and reliability problems.
\newblock {\em siam Journal on Computing}, 8(3):410--421, 1979.

\bibitem{Valiant08}
Leslie~G. Valiant.
\newblock Holographic {A}lgorithms.
\newblock {\em {SIAM} J. Comput.}, 37(5):1565--1594, 2008.

\bibitem{Zhuk17}
Dmitriy Zhuk.
\newblock A {P}roof of {CSP} {D}ichotomy {C}onjecture.
\newblock In {\em Proc. {FOCS} 2017}, pages 331--342, 2017.

\end{thebibliography}

\newpage

\appendix

\section{Proof of Lemma~\ref{lem:lowerBoundsHoms}}\label{sec:ProofOfLemma}
Recall that the goal is to show that for any $r \in \mathbb{Z}_{\geq 2}$ and every class of $r$-uniform hypergraphs $\mathcal{H}$ of unbounded treewidth, the problem $\Homprob(\mathcal{H})$ is $\#\mathrm{W}[1]$-hard and cannot be solved in time $f(|H|)\cdot |V(G)|^{o(\mathsf{tw}(H)/\log(\mathsf{tw}(H)))}$ for any function $f$, unless ETH fails, even if all inputs $(H,G)$ to $\Homprob(\mathcal{H})$ satisfy that $G$ is $r$-uniform as well. For what follows, fix $r$ and $\mathcal{H}$.

To prove this claim, we need to take a detour to constraint satisfaction problems and relational structures, which is encapsulated here since we do not need it anywhere else in the paper.

A vocabulary $\tau$ is a tuple $(R_1,\dots,R_\ell)$ of relation symbols, each equipped with an arity $a_i$. A structure $A$ over $\tau$ consists of a finite universe $V(A)$ and relations $R^A_1,\dots R^A_\ell$ over $V(A)$ such that, for each $i\in [\ell]$, $R^A_i$ has arity $a_i$.
Given two structures $A$ and $B$ over vocabulary $\tau$, a homomorphism from $A$ to $B$ is a mapping $\varphi: V(A) \to V(B)$ such that, for each $i\in[\ell]$ and tuple $t\in R^A_i$, we have that $\varphi(t)\in R^B_i$. The Gaifman graph of a structure is defined similarly to the Gaifman graph of a hypergraph, and the treewidth of a structure is defined to be the treewidth of the Gaifman graph as well.

Now, we associate each $H\in \mathcal{H}$ with a structure $A=A(H)$ as follows. Let $k=|V(H)|$ --- assume w.l.o.g.\ that $V(H)=\{v_1,\dots,v_k\}$ --- and set $\tau_H=(U_1,\dots,U_k,E)$, where the arity of each $U_i$ is $1$, and the arity of $E$ is $r$. We set $V(A)=V(H)$, and $U^A_i=\{v_i\}$ for each $i\in [k]$. Furthermore, for every hyperedge $\{v_{j_1},\dots,v_{j_r}\}$ with $j_1 < \dots < j_r$ we add the tuple $(v_{j_1},\dots,v_{j_k})$ to $E$.

We make two trivial observations:
\begin{enumerate}
    \item[(a)] $\mathsf{tw}(H)=\mathsf{tw}(A)$.
    \item[(b)] The only homomorphism from $A$ to $A$ is the identity (due to the unary relations~$U^A_i$).
\end{enumerate}
Condition (b) implies that $A$ is a core: there is no homomorphism from $A$ to a \emph{proper} substructure of $A$.

Now let $\mathcal{A}=\{A(H)\mid H \in \mathcal{A}\}$ and let $\Homprob(\mathcal{A})$ denote the problem of, given $A \in \mathcal{A}$ and a structure $B$ over the same vocabulary as $A$, computing the number $\#\homs{A}{B}$ of homomorphisms from $A$ to $B$. Since, by (a), the treewidth of $\mathcal{A}$ is unbounded, we can invoke the lower bounds of Dalmau and Jonsson~\cite{dalmau2004complexity}, and of Marx~\cite[Corollary 5.3]{Marx10} to obtain that $\Homprob(\mathcal{A})$ is $\#\mathrm{W}[1]$ hard and cannot be solved in time $f(|H|)\cdot |V(G)|^{o(\mathsf{tw}(H)/\log(\mathsf{tw}(H)))}$ for any function $f$, unless ETH fails.

For the remainder of this proof, it hence suffices to construct a tight reduction from $\Homprob(\mathcal{A})$ to $\Homprob(\mathcal{H})$. To this end, we consider an intermediate coloured problem; for what follows, recall that we do not allow isolated vertices in our hypergraphs.

Let $H$ be a hypergraph, and let $(G,c)$ be a pair of a hypergraph $G$ and a homomorphism $c: G \to H$; we call $(G,c)$ an $H$-coloured hypergraph. A \emph{colour-prescibed} homomorphism from $H$ to $G$ is a homomorphism $\varphi \in \homs{H}{G}$ such that, for every $v \in V(H)$, we have $c(\varphi(v))=v$. We write $\mathsf{cpHom}(H\to(G,c))$ for the set of all colour-prescribed homomorphisms from $H$ to $(G,c)$, and we denote with $\#\textsc{cpHom}(\mathcal{H})$ the problem of, given $H \in \mathcal{H}$ and an $H$-coloured hypergraph $G$, computing $\#\mathsf{cpHom}(H \to (G,c))$.

\paragraph*{Reducing $\Homprob(\mathcal{A})$ to $\#\textsc{cpHom}(\mathcal{H})$}
Let $A=A(H)$ and $B$ be the input to $\Homprob(\mathcal{A})$. Let $k=V(H)$ and recall that $A$ and $B$ must be over vocabulary $\tau_H$. We construct the tensor $B \otimes A$: the universe is $V(B \otimes A)=V(B) \times V(A)$, and for each $r$-tuple $\vec{t}=(t_1,\dots,t_r) \in V(B \otimes A)^r$, we add $\vec{t}$ to $E^{B \otimes A}$ if and only if the projections to the first and second elements of $\vec{t}$ yield tuples in $E^B$ and $E^A$, respectively. Moreover, for each $i\in [k]$, we add $(x,y)\in V(B \otimes A)$ to $U_i^{B\times A}$ if and only if $x \in U^B_i$ and $y \in U^A_i$. It is well-known that homomorphism counts are linear w.r.t.\ the tensor product of relational structures. Hence
\begin{equation*}
    \#\homs{A}{B \otimes A} = \#\homs{A}{B} \cdot \#\homs{A}{A} \stackrel{(b)}{=}  \#\homs{A}{B}
\end{equation*}
Next, we associate $B \otimes A$ with an $H$-coloured hypergraph $G(B,A)$ as follows. First, we start with vertices $V(B \otimes A)$ and note that each vertex is of the form $(x,v_i)$ for some $i \in [k]$. Hence $(x,v_i)\in U^{B \otimes A}_i$ if $x \in U^B_i$. Otherwise $(x,y_i)$ is contained in none of the $U^{B \otimes A}_i$; we delete all such vertices that are not contained in any of the unary relations.

The remaining set $V'$ of vertices can be partitioned into 
\[V' = U^{B \otimes A}_1\,\dot \cup\, \dots\,\dot \cup \,U^{B \otimes A}_k\]
(note that some of the $U^{B\otimes A}_i$ might be empty.)

We are now able to construct $G(B,A)$. First, set $V(G(B,A))=V'$. Next, for each tuple $(v_{j_1},\dots,v_{j_r})\in E^A$, we add to $G(B,A)$ all hyperedges of the form $\{(x_1,v_{j_1}),\dots,(x_r,v_{j_r})\}$ such that $(x_1,\dots,x_r)\in E^B$ is satisfied. Moreover, set $c$ as the $H$-colouring of $G(B, A)$ that maps $(u,v)$ to $v$ (recall that $V(A)=V(H)$).

\begin{claim}
We have $\#\mathsf{cpHom}(H \to (G(B,A),c)) = \#\homs{A}{B}$.
\end{claim}
\begin{proof}
For a tuple $t$, we write $\pi_i(t)$ for the $i$-th entry of $t$. Let $\phi \in \mathsf{cpHom}(H \to (G(B,A),c))$ and consider a mapping $b$ that maps $\phi$ to $b(\phi) = \psi : V(A) \to V(B)$ such that for each $v_i \in V(A), \psi(v_i) = \pi_1(\phi(v_i))$. In fact, $\psi \in \homs{A}{B}$. To see this, first note that $\phi(v_i)$ is by definition in $U^{B\times A}_i$, which implies that $\psi(v_i) = \pi_1(\phi(v_i)) \in U^B_i$ hence the unary relations are preserved. For $e = \{v_{j_1}, \dots, v_{j_r}\} \in E(H)$, we have $\phi(e) = \{(x_1, v_{j_1}), \dots, (x_r, v_{j_r})\}$ for some $x_1, \dots, x_r$. Assuming that $j_1 < \dots < j_r$, we also have $(x_1, \dots, x_r) \in E(B)$, which implies that $\psi$ preserves the relation $E$ as well. Hence, $\psi \in \homs{A}{B}$.

To prove our claim, it suffices to show that $b$ defines a bijection. Indeed, $b$ is injective. To see this, let $\phi, \phi'$ such that $b(\phi) = b(\phi') = \psi$. For $v_i \in V(H)$, we have $\phi(v_i) = (x, v_i)$ and $\phi'(v_i) = (x', v_i)$. Since $\psi(v_i) = \pi_1(\phi(v_i)) = \pi_1(\phi'(v_i))$, we have $x = x'$.

Finally, to show that $b$ is also surjective, let $\psi \in \homs{A}{B}$ and define $\phi$ such that for each $v_i \in V(H)$, $\phi(v_i) = (\psi(v_i), v_i)$. Clearly, $\phi \in \mathsf{cpHom}(H \to (G(B,A),c))$ and $b(\phi) = \psi$, completing the proof.
\end{proof}

Moreover, by construction, $G(B,A)$ is $r$-uniform. Hence we can use our oracle to query $\#\mathsf{cpHom}(H \to (G(B,A),c))$, completing the computation.
Note that $(G(B,A),c)$ can easily be constructed in FPT time from $B$ and $A$, and the parameter of the oracle query ($|H|$) is bounded linearly by the parameter of the input ($|A|$); thus the conditional lower bound under ETH transfers as well.

\paragraph*{Reducing $\#\textsc{cpHom}(\mathcal{H})$ to $\Homprob(\mathcal{H})$}
The second reduction follows verbatim the standard argument for graphs via inclusion-exclusion (see~\cite[Lemma 2.49]{Roth19} for a detailed presentation). The only trivial modification that has to be made is the deletion of isolated vertices in the terms arising in the intermediate steps.

\end{document}